\documentclass[a4paper,twocolumn,16pt,accepted=2024-03-30]{quantumarticle}
\pdfoutput=1

\usepackage[utf8]{inputenc}
\usepackage[english]{babel}
\usepackage[T1]{fontenc}
\usepackage{amsmath}
\usepackage{hyperref}[breaklinks=true]
\usepackage{tikz}
\usepackage{lipsum}
\usepackage[english]{babel}
\usepackage{graphicx}					
\usepackage[pdftex]{epsfig}
\usepackage{ragged2e}
\usepackage[caption=false]{subfig}
\usepackage{amsmath,amssymb}
\usepackage{dsfont}
\usepackage{siunitx}
\usepackage{mathtools}
\usepackage{amsmath}
\usepackage{amssymb}
\usepackage{amsthm}

\usepackage{algorithm}
\usepackage{algpseudocode}

\algblock{Input}{EndInput}
\algnotext{EndInput}
\algblock{Output}{EndOutput}
\algnotext{EndOutput}
\newcommand{\Desc}[2]{\State \makebox[2em][l]{#1}#2}

\usepackage{algorithm}
\usepackage{algorithmicx}
\usepackage{algpseudocode}

\makeatletter

\newenvironment{abstractalgorithm}[1][htb]{%
    \floatname{algorithm}{Abstract algorithm}
   \begin{algorithm}[#1]%
  }{\end{algorithm}}\newenvironment{quantumalgorithm}[1][htb]{%
    \floatname{algorithm}{Quantum algorithm}
   \begin{algorithm}[#1]%
  }{\end{algorithm}}
\makeatother

\usepackage{eufrak}
\DeclareMathAlphabet{\mathpzc}{OT1}{pzc}{m}{it}

\usepackage{dsfont}
\newcommand{\1}{\mathds{1}}
\newcommand{\id}{\1}
\newcommand{\R}{\mathbb{R}}
\newcommand{\CC}{\mathbb{C}}

\newcommand{\tr}{\mathrm{tr}}

\newcommand{\X}{X}
\newcommand{\Y}{Y}
\newcommand{\Z}{Z}

\newcommand{\J}{\mathcal{J}}

\newcommand{\li}{\mathcal{L}}

\newcommand{\W}{\mathcal{W}}
\renewcommand{\H}{\mathcal{H}}
\newtheorem{theorem}{Theorem}
\newtheorem{lemma}[theorem]{Lemma}
\newtheorem{proposition}[theorem]{Proposition}
\newtheorem{definition}[theorem]{Definition}
 \newtheorem{corollary}[theorem]{Corollary}

\renewcommand{\i}{\ensuremath\mathrm{i}}
\newcommand{\ket}[1]{\left.\left|{#1}\right.\right\rangle}

\newcommand{\bra}[1]{\left.\left\langle{#1}\right.\right|}

\newcommand{\ketbra}[2]{\ket{#1} \!\! \bra{#2}}

\newcommand{\sandwich}[3]
  {\left\langle  #1 \right| #2 \left| #3 \right\rangle}

\newcommand\n{\hat n}

\DeclareMathAlphabet{\mathpzcc}{OT1}{pzc}{m}{it}
\DeclareMathAlphabet{\mathpzc}{T1}{pzc}{m}{it}{\huge}
\DeclareMathAlphabet{\mathpzg}{T1}{pzc}{m}{n}{\huge}

  \newcommand{\p} {\mathpzcc{p}}

  \newcommand{\w} {\mathpzcc{w}}

\def\d{\text{d}}

\renewcommand{\u}[1]{\underline {#1}}
\def\H{\hat H}

\renewcommand\L{\hat L}
\newcommand{\qu}{\hat Q}
\renewcommand{\X}{\hat X}
\renewcommand{\Y}{\hat Y}
\renewcommand{\Z}{\hat Z}
\renewcommand\d[1]{\Delta(#1)}
\renewcommand\o[1]{\sigma(#1)}
\newcommand\etak[1]{\eta^{\mtiny{(#1)}}}
\def\w{\hat W}
\def\h{\hat H}
\def\D{\hat D}
\def\U{\hat{ \mathcal U}}
\def\H{{\hat{ \mathcal H}}}
\def\W{\hat{ \mathcal W}}
\def\dell{{\delta\ell}}
\def\dell{{\mathpzc {S}}}
\def\dell{{s}}
\def\S{\mathbf{s}}

\def\p{\Z}

\def\Hell{{{\H}_\ell}}

\def\Well{{\W_\ell}}

\def\oHell{\o{\H_\ell}}

\def\dHell{\d{\H_\ell}}

\def\Uell{{\U_\ell}}
\def\u{{\hat U}}

\def\Q{\hat Q}
\def\D{\hat D}
\def\Dbf{\mathbf{\hat D}}

\def\wcan{\w^{(\mathrm{can})}}
\def\VDBI{\mathcal T}
\def\UDBI{\mathcal V}
\def\RDBI{\mathcal R}
\def\PDBI{\mathcal P}
\def\ugci{\hat U^{(\mathrm {GCI})}}

\def\a{\hat a}
\def\I{\mathcal I}
\def\delln{{\dell^{\mtiny{(N)}}}}
\def\elln{\ell^{\mtiny{(N)}}}
\def\wn{\w^{\mtiny{(N)}}}
\def\wnm{\w^{\mtiny{(N-1)}}}
\def\un{\u^{\mtiny{(N)}}}
\def\unm{\u^{\mtiny{(N-1)}}}
\def\hn{{\h^{\mtiny{(N)}}}}
\def\hnm{{\h^{\mtiny{(N-1)}}}}

\def\dellnm{\dell^{\mtiny{(N-1)}}}
\def\D{\mathcal D}
\def\Z{\hat Z}
\def\Pr{\mathrm{Pr}}
\newcommand\sntc[2]{\Delta^{#1}(#2)}
\def\snt{\sntc n\ell}

\def\dtn {\mathrm{d}^ns}
\def\dr {\mathrm{d} r}

\def\dd {\mathrm{d}}
\def\bc{\{0,1\}^{\times L}}
\def\zm{\Z_\mu}
\newcommand\zmc[1]{\Z_{\etak {#1}}}
\def\xn{\X_\nu}
\def\zmp{\Z_{\mu'}}
\def\xnp{\X_{\nu'}}
\def\mnbc{{\mu,\nu \in \bc}}

\def\hmn{\overline {H}_{\mu,\nu}}

\def\A{\hat K}
\def\B{\hat J}

\def\linops{{\mathcal L(\CC^{\times D})}}
\def\A{\hat A}
\def\B{\hat B}
\def\C{\hat C}
\def\D{\hat D}
\def\NC{\mathcal N}
\def\J{J}
\def\n{K}
\def\elk{\ell_k}
\def\elkm{{{\ell_{k-1}}}}

\def\v{\hat V}
\def\vgc{\v^{\mtiny{\mathrm{(GC)}}}}

\def\vd{\v^{\mtiny{(\Delta)}}}
\def\elk{\ell_k}
\def\elkm{\ell_{k-1}}
\def\wz{{\w^{\mtiny{(Z)}}}}
\def\wmu{{\w^{\mtiny{(\mu)}}}}
\def\dzz{\Delta^{\mtiny{(Z)}}}
\def\dmu{\overline{\Delta^{\mtiny{(\eta)}}}}

\def\b{B}

\def\J{\hat J}

\def\sgd{$\sigma$-decreasing }
\def\pch{\hat Q}
\newcommand\del[1]{}
\newtheorem*{theorem*}{Theorem}

\newcommand\mtiny[1]{\mbox{\tiny\ensuremath{#1}}}
\newcommand\msmall[1]{\mbox{\small\ensuremath{#1}}}
\usepackage{mathcomp}
\begin{document}

\title{Double-bracket quantum algorithms for diagonalization}
\author{Marek Gluza}
\orcid{0000-0003-2836-9523}
\email{marekludwik.gluza@ntu.edu.sg}
\affiliation{School of Physical and Mathematical Sciences, Nanyang Technological University, 21 Nanyang Link, 637371 Singapore, Republic of Singapore}

\begin{abstract}
This work proposes double-bracket iterations as a framework for obtaining diagonalizing quantum circuits. 
Their implementation on a quantum computer consists of interlacing evolutions generated by the input Hamiltonian with diagonal evolutions which can be chosen variationally.
No qubit overheads or controlled-unitary operations are needed but the method is recursive which makes the circuit depth grow exponentially with the number of recursion steps. 
To make near-term implementations viable, the proposal includes  optimization of diagonal evolution generators and of recursion step durations. Indeed, thanks to this numerical examples show that the expressive power of double-bracket iterations suffices to approximate  eigenstates of relevant quantum models with few recursion steps.
Compared to brute-force optimization of unstructured circuits double-bracket iterations do not suffer from the same trainability limitations.
Moreover, with an implementation cost  lower than required for quantum phase estimation they are more suitable for near-term quantum computing experiments.
More broadly, this work opens a pathway for constructing purposeful quantum algorithms based on so-called double-bracket flows also for tasks different from diagonalization and thus enlarges the quantum computing toolkit geared towards practical physics problems.
\end{abstract}

\maketitle

Studying quantum many-body systems is an area where quantum computing may lead to practical advances outside the scope of what we can compute numerically.
We often gain physics understanding by analyzing eigenstates and eigenvalues of quantum models so quantum algorithms for diagonalization could be key for achieving new insights.

To approximate eigenstates on a quantum computer, a basic idea is to variationally find appropriate parameters for a sequence of quantum gates.
However, the optimization needed for that is made difficult by a phenomenon referred to as barren plateaus, see Ref.~\cite{KishorRevModPhys.94.015004} for review and Refs.~\cite{PhysRevLett.127.120502,stilck2021limitations} for relevant NP hardness results.
Unstructured ansatzae, e.g. circuits composed of CNOT gates and single-qubit rotations, have been studied extensively but there are very few exceptional cases when it has been possible to find parametrizations that generalize to system sizes larger than a handful of qubits. 
Recently, a more algorithmic approach  based on the classical Lanczos algorithm~\cite{lanczos1950iteration}, widely employed for numerical diagonalization, has found  application in a quantum device~\cite{motta2020determining}.
Its innovation consists in implementing imaginary-time evolutions using accessible unitary operations.
Ref.~\cite{kokail2019self} provides another example for experimental preparation of eigenstate approximations for relatively large systems.
Similarly to the approach taken here, evolutions under the input Hamiltonian were composed with simple evolutions to achieve eigenstate approximations. 
As we will see it is possible to use analytical means to  guide such eigenstate preparation as opposed to optimization.

This work will establish that the Głazek-Wilson-Wegner (GWW) flow~\cite{PhysRevD.48.5863,PhysRevD.49.4214,wegner1994flow} can play a conceptual role for quantum computing, namely it provides a feasible solution to the task of compiling diagonalizing quantum circuits.
Applications of GWW flow in condensed-matter physics have been explored using classical computers  and are showcased in the monograph by Kehrein~\cite{kehrein_flow}, see Ref.~\cite{wegner2006flow} for a concise review.
The GWW flow is an example of non-linear differential equations called double-bracket flows~\cite{deift1983ordinary,BROCKETT199179,Chu_iterations} whose mathematics is covered in the monograph by Helmke and Moore~\cite{helmke_moore_optimization}.

Having near-term quantum computing applications in mind,  we will depart from the GWW flow.
Instead, we will consider double-bracket iterations which can recover continuous flows as a special limit but give flexibility to attempt lowering implementation cost.
It appears difficult to characterize the precise efficacy of variational double-bracket iterations analytically but numerical simulations show that just a handful of recursion steps can yield surprisingly good approximations of low-energy eigenstates of the quantum Ising model.

The manuscript is structured as follows.
Sec.~\ref{VDB} \emph{i)} defines general double-bracket iterations, \emph {ii)} expounds on data structures for diagonalization in quantum computing, \emph{iii)} formulates the specific diagonalization double-bracket iteration ansatz, \emph{iv)} proves its diagonalizing monotonicity, \emph{v)} discusses the quantum circuit components needed to implement double-bracket iterations and \emph{vi)} formalizes the quantum algorithm by means of a recursive transpiling algorithm.
More broadly, Sec.~\ref{VDB} demonstrates that in the double-bracket approach it is clear why and how diagonal-dominance iteratively increases.

The recursive character of double-bracket iterations leads to an exponential runtime of the quantum algorithms in the number of iteration steps.
Sec.~\ref{GWWnumsec} explores numerically what can be achieved with a small number of recursion steps and how to reduce their number.
It demonstrates that few steps of a diagonalization double-bracket iteration suffice to achieve relevant state preparations.
Double-bracket iterations can be approximated by quantum circuits using Hamiltonian simulation \cite{Su}.
Simulating evolution under input Hamiltonians is being actively explored in experiments~\cite{martinez2016real,google2020hartree,somhorst2021quantum} and so the proposed quantum algorithm lends itself towards near-term experiments.
In particular, no controlled-unitary operations are needed.

Sec.~\ref{GWWreview} focuses on continuous double-bracket flows, in particular a variant of a double-bracket iteration is proposed which converges to the GWW flow.
The obtained runtime is not efficient but the presented theory should be  relevant  for anyone wishing to explore quantum computing applications of double-bracket flows for tasks other than diagonalization.

Sec.~\ref{context} discusses double-bracket iterations in context of other approaches to constructing quantum algorithms.
In particular quantum dynamic programming~\cite{QDP} can reduce the circuit depth of recursive double-bracket iterations, thus potentially strengthening their role for future quantum computing.
Finally, Sec.~\ref{conclusions} states the specific conclusions established in this work and suggests an outlook of double-bracket flows for scientific quantum computing.

\section{Double-bracket quantum algorithms for diagonalization}
\label{VDB}
This section will propose double-bracket quantum algorithms for diagonalization.
They will be derived from double-bracket iterations.
While similar to discretizations of double-bracket flows, double-bracket iterations allow for more flexibility.
We will explore what happens if  the individual evolution steps have relatively long durations.
Additionally, we will retain a `double-bracket' form of rotation generators which is similar to flows but will consider it as an ansatz for variational purposes.
We will introduce flows only later in Sec.~\ref{GWWreview} because they are less relevant for near-term quantum computing.

\paragraph{Definition of double-bracket iterations.}
Let $\h_0$ be a Hamiltonian and $\D_0,\D_1,\ldots$ be a sequence of hermitian diagonal operators.
We define the respective diagonalization \emph{double-bracket iteration} (DBI) by the recursion starting with $k=0$ and recursion step 
\begin{align}
  \h_{k+1}=   e^{s_k \w_{k}}\h_k e^{-s_k \w_{k}} \ ,
  \label{eq:hk}
\end{align}
where $\w_k$ is given by
\begin{align}
  \w_{{k}} = [\D_k,{\h_{k}}]\ .
  \label{eq:wkp}
\end{align}
Here, the step durations $s_k$ can be chosen variationally for the considered task.
Sec.~\ref{GWWnumsec} will provide numerical examples. Sec.~\ref{GWWreview} will discuss the relation to the less relevant for quantum computing continuous double-bracket flows which can be recovered in the limit of an equidistant and infinitesimal step duration~$s_k$.

When the bracket~\eqref{eq:wkp} involves a diagonal operator then we will say that the DBI is a diagonalization DBI because we will see that such iterations give rise to circuits well-suited for that task.
Note, that non-diagonal (but hermitian) $\D_k$ could be useful for other purposes.

It should be highlighted that for variational purposes, the resulting ansatz will be parametrized by the parametrization of the diagonal operators $\D_k$.
This is different from unstructured parametrizations of circuits because the diagonal evolutions governed by $\D_k$ are `folded' through the brackets~\eqref{eq:wkp}.
As we will see, this form makes it easier to search for the diagonal operators $\D_k$.

Eqs.~\eqref{eq:hk} and \eqref{eq:wkp} allow to understand the `double-bracket' nomenclature.
We will say that
\begin{align}
    \h_k(s) = e^{s \w_k} \h_k(0) e^{-s \w_k}
    \label{DBR}
\end{align}
is a \emph{double-bracket rotation} because it
satisfies a Heisenberg equation involving two, not one, brackets
\begin{align}
    \partial_s\h_k(s) = [  [\D_k,\h_k(0)], \h_k(s)]\ .
    \label{eq:nonlinearHeisenberg}
\end{align}

The monograph by Helmke and Moore~\cite[Ch. 2.3]{helmke_moore_optimization} discusses double-bracket iterations as Lie-bracket recursions.
It shows that if $\D_0=\D_1=\ldots = \hat D^*$ then as long as $s_0=s_1=\ldots=s^*$ are sufficiently short (an explicit bound can be obtained) then the recursion converges to a fixed point $\h_\infty$ with $[\h_\infty, \hat D^*]=0$.

The DBIs considered here will vary the diagonal operators $\D_k$ and recursion step durations $s_k$ so as to increase the diagonalization rate in each step.
Eqs.~\eqref{eq:hk} and \eqref{eq:wkp} allow to change the $\D_0,\D_1,\ldots$ generators, in hope to faster achieve a fixed point $\h_\infty$ which commutes with any diagonal operator, i.e. diagonalize.

We will say that a sequence of double-bracket rotations, or more generally a sequence of approximations to double-bracket rotations, is a DBI.
In Subsec.~\ref{GCIsec}, we will discuss how a DBI can be turned into a unitary circuit.
As in the title of this work, we will refer to a  DBI transpiled into a quantum circuit as a double-bracket quantum algorithm.

\paragraph{Diagonalization in quantum computing.}
A~Hamiltonian $\h_0$ can be considered as the input to the task of finding a diagonalization transformation.
Whenever specific examples will be required, $\h_0$ will be taken to be a Hamiltonian of $L$ qubits 
and the $D=2^L$ dimensional Pauli matrices acting on qubit $i$ will be denoted by $\X_i,\Y_i$ and $\Z_i$. 
Throughout, we choose $\Z_i$ to be diagonal.
Let us now discuss how one can approach diagonalization in quantum computing.

We will speak of \emph{eigenstate-by-eigenstate}  diagonalization if a unitary is applied to a basis vector, obtaining a particular eigenstate or an approximation thereof.
Let us denote the computational basis vectors by $\ket \mu$ with $\msmall{\mu_1,\ldots,\mu_L\in \{0,1\}}$ such that $\Z_i \ket \mu = (-1)^{\mu_i}\ket \mu$.
If a unitary $\u$ is diagonalizing $\h_0$ then if we apply $\u$ to $\ket \mu$, we will obtain an eigenstate.
More specifically, let $\hat H_\infty = \u^\dagger \h_0 \u$ be diagonal then $\ket \mu_\infty = U\ket \mu$ is an eigenstate with  corresponding energy eigenvalue $E_\mu$ because
\begin{align}
  \h_0 \ket \mu_\infty = \u \h_\infty \ket \mu = E_\mu \ket\mu_\infty\ . 
\end{align}
 
A quantum circuit approximation to the diagonalizing transformation $\hat U$  above can be viewed as an eigenstate-by-eigenstate diagonalization quantum algorithm which is global because for all states $\ket\mu$ we use the same circuit.
Brute-force optimization of circuits aimed at preparing individual eigenstate approximations has been studied extensively, see Ref.~\cite{KishorRevModPhys.94.015004} for a review, and could be considered local eigenstate-by-eigenstate diagonalization because the circuit can vary for different target eigenstates.

The aim of this section is to propose that DBIs are well suited  for eigenstate-by-eigenstate diagonalization.
A priori, the argumentation that will be presented will be based on a global cost function.
However, when a diagonalization DBI has not advanced far, it may be viewed as performing local eigenstate-by-eigenstate diagonalization for some selected states $\ket \mu$.
This is because even if the global cost function has not been reduced to zero, it can happen that applying the unitary circuit arising from the few-step DBI will give good eigenstate approximation for some states $\ket \mu$.
In general,  performance of local eigenstate-by-eigenstate diagonalization can be witnessed by measuring the energy fluctuation
\begin{align}
\Xi_k(\mu) = \sqrt{\sandwich \mu {\h_k^2} \mu - \sandwich \mu {\h_k} \mu^2} \ ,
 \label{eq:ef}
\end{align}
which vanishes for eigenstates and can be evaluated through experimentally accessible observables~\cite{kokail2019self}.

As an alternative, one may consider \emph{thermally encoded} diagonalization where one would encode the input Hamiltonian $\h_0$ in a density matrix, e.g. as
\begin{align}
    \hat \rho^{(H_0)} = (\id  + \h_0/ \|\h_0\|)/ \tr(\id  + \h_0/ \|\h_0\|) \ .
\end{align}
Again an approximation to an exact diagonalization transformation $U$ could be considered.
In this case diagonalization effectiveness could be measured by, e.g., evaluating quantum coherence in the computational basis~\cite{coherence_review}.
While thermally encoded diagonalization  may appear appealing, it should be noted  that \emph{i)} all observables would carry a lot of trivial noise due to the close proximity to the maximally mixed state, \emph{ii)} not all of the exponentially many eigenstates are of equal interest and \emph{iii)} certifying success may be difficult because, e.g., quantum coherence is a function of the entire state.
For those reasons this framework will not be considered further in this work but it goes to show that the quantum data structure for diagonalization, in principle, may be chosen in more than one way.

We discussed two ways in which a quantum computer can be set up for the task of diagonalization.
Next, let us discuss two oracles for the quantum computer to query the input $\h_0$.
The matrix dimension $D=2^L$ is expected to be large so, unlike in classical computing, for quantum computation it is more natural to build algorithms querying evolutions $\{e^{-it\h_0}\}_{t\in \mathbb R}$ governed by $\h_0$ rather than one by one inputting  matrix elements of $\h_0$ into a quantum computer.

The \emph{Hamiltonian simulation} oracle  translates classical knowledge of the values of the couplings in the input Hamiltonian $\h_0$ into evolutions $e^{-it\h_0}$.
This should be viewed as being in line with the natural sparsity structure of Hamiltonians which are relevant in physics: The number of couplings will usually scale at most polynomially with the number of qubits $L$, rather than exponentially.
Query access to this evolution oracle can be achieved by Hamiltonian simulation quantum algorithms with an almost linear runtime in the evolution duration~\cite{Su}.

An \emph{oblivious evolution} oracle allows to query the application of unitary evolutions $e^{-it\h_0}$ to a state governed by $\h_0$ for desired durations $t$.
Here one does not need to know the precise couplings in the Hamiltonian, a situation close to the setting of analog quantum simulation where the evolution just `happens'.
If we disregard the knowledge about the couplings in $\h_0$ then the Hamiltonian simulation oracle can be seen as oblivious, simply evolving the state as queried.
For readers familiar with QIBO, this is similar to having the freedom to change the `backend' and being agnostic as to what is the specific implementation at hand~\cite{qiboEfthymiou_2022}.

DBIs reveal how to turn evolutions under the input Hamiltonian $e^{-it\h_0}$ into an approximation of eigenstate-by-eigenstate diagonalization, thus double-bracket quantum algorithms work within the oblivious evolution oracle framework.
The next subsection will  derive analytical conditions for a DBI to tend towards a diagonal fixed point and after that we will discuss how to implement the arising double-bracket quantum algorithms on quantum computers.

\subsection{Diagonalization  double-bracket iterations}

In the convention that all $\Z_i$ operators are diagonal, their products
\begin{align}
\p_\mu=\prod_{j=1}^L \Z_j^{\mu_j}\ , 
\label{eqphaseflips}
\end{align}
with $\msmall{\mu_1,\ldots,\mu_L\in \{0,1\}}$, form a basis of diagonal operators on $L$ qubits.
This means  that any $D=2^L$ dimensional diagonal operator $\D$ can be expressed as the linear combination 
\begin{align}
    \D = D^{-1}\sum_{\mu\in\{0,1\}^{\times L}} \langle  \p_\mu, \D\rangle_\mathrm{HS}\ \p_\mu\ ,
\end{align}
which involves the Hilbert-Schmidt scalar product~\cite{nielsen2010quantum}
\begin{align}
\label{eqHS}
    \langle \A, \B \rangle_\mathrm{HS} = \tr[\A^\dagger \B ]\ .
\end{align}

The hermitian conjugation in the definition of the Hilbert-Schmidt scalar product is important because we will be using anti-hermitian operators $\w = - \w^\dagger$ so that their induced Hilbert-Schmidt norm defined through $\|\A\|_\text{HS}^2 = {\langle { \A, \A} \rangle_{\mtiny{\mathrm{HS}}}}$ is
\begin{align}
    \|\w\|_\text{HS}^2 = -\tr[\w^2 ]\ .
\end{align}
There is no minus for hermitian operators, e.g, $\|\h_0\|_\text{HS}^2 = \tr[\h_0^2 ]$.

Anti-hermitian operators arise naturally when considering brackets, i.e., commutators, of hermitian operators.
For example, we will use brackets of hermitian operators $\A, \B$
\begin{align}
  \w(\A,\B)= [\A, \B]  = \A \B-  \B\A\ .
  \label{eq:wzdef}
\end{align}

Note that $e^{s \w(\A,\B)}$ is unitary for any real duration of the evolution $s$ which can be proven analogously to, e.g., why $e^{-is \h_0}$ is unitary.
In fact $i \h_0$ is anti-hermitian.

Let  $\dzz(\J)$ be a map associating a hermitian and diagonal operator to any hermitian operator $\J$.
We call $\dzz(\cdot)$ a \emph{diagonal association} and it is specifically purposed  to give rise to brackets
\begin{align}
    \wz(\J) = [\dzz(\J),\J]\ 
    \label{wzbracket}
\end{align}
which, when exponentiated, generate unitary evolutions.
We speak of a map because $\dzz(\J)$ could be a result of some optimization procedure applied to $\J$ for every step of a diagonalization DBI and thus $\dzz$ could be used to variationally construct $\D_0, \D_1,\ldots$ on the fly.
Let us next see some examples of diagonal associations.

The dephasing channel for $L$ qubits 
\begin{align}
  \d\J = D^{-1}\sum_{\mu\in\{0,1\}^{\times L}} \p_\mu \J\p_\mu\ 
  \label{eq:pinch}\ 
\end{align}
is a diagonal association; it takes $\J$ and returns an operator with the same diagonal matrix elements and off-diagonal ones equal to zero.
The output of the dephasing channel is again hermitian so the  bracket, which we will call canonical, defined by
\begin{align}
   \wcan(\J) =  \w( \d\J,  \J) = [ \d\J,\J]
   \label{eq:wcan}
\end{align} generates unitary operators because it is anti-hermitian $(\wcan)^\dagger = -\wcan$.

The diagonal association can be chosen to be independent of $\J$. 
For example we can set 
\begin{align}
\label{eq:Zmudiagonassociation}
    \Delta^{(\mu)}(\J) = \p_\mu
\end{align} by sampling a random $\mu$ each time it is queried. 
Again, such an association allows to define a bracket
\begin{align}
    \wmu(\J) = \w( \Z_\mu,  \J) = [ \Z_\mu,\J] \ .
    \label{eqwmu}
\end{align}

For any diagonal association $\dzz(\cdot)$ we define the respective DBI by setting $\D_k=\dzz(\h_k) $.
So to build a DBI, in every step we could dephase as in~\eqref{eq:pinch} or pick a new random operator as in~\eqref{eq:Zmudiagonassociation}.
The diagonal association $\dzz(\J)$ could be extended to be a function of the recursion step number $k$ and be built on the fly to serve the purpose of the iteration. 
For example for each given $\h_k$ one can perform rounds of closed loop optimization varying the parameters of an ansatz for $\dzz(\h_k)$ and aim to obtain $\h_{k+1}$ yielding a low cost function value (e.g., be more diagonal).
For clarity let us summarize this theoretical approach by the following abstract algorithm.

\begin{abstractalgorithm}[H]
\caption{Diagonalization DBI}
\label{algorithmDBI}
	\begin{algorithmic}
	  \Input
  \Desc{$K$:}{\quad total number of recursion steps}
  \Desc{$\S$:}{\quad a vector of $K$ recursion step durations}
  \Desc{$\dzz$:}{\quad diagonal association }
  \Desc{$\h_0$:}{\quad input Hamiltonian}
  \EndInput
  \Output
  \Desc{${\h_{K}}$:}{\quad the result of DBI Eqs.~\eqref{eq:hk} and \eqref{eq:wkp} after $K$ steps}
    \EndOutput	  
  \State
  \textbf{For }$k = 0$ to $K-1$  \\
  \hspace{\algorithmicindent} 
  \textbf{Set} {   $\D_k\leftarrow\dzz(\h_k) $}\\
  \hspace{\algorithmicindent} 
  \textbf{Compute} { $\w_k=[\D_k,\h_k]$ using Eq. \eqref{eq:wkp}}
  \State
\hspace{\algorithmicindent} 
  \textbf{Set } { $\h_{k+1} \leftarrow e^{s_k \w_k} \h_k e^{-s_k\w_k}$} as in Eq.~\eqref{eq:hk}
 \\
   \textbf{Return } { $\h_{K}$} \end{algorithmic}
  
\end{abstractalgorithm}

In the remainder of this section, we will discuss  \emph{i)} why the above examples of diagonal associations \eqref{eq:pinch} and \eqref{eq:Zmudiagonassociation} have diagonalizing tendencies, \emph{ii)} how to implement Eqs.~\eqref{eq:hk} and \eqref{eq:wkp} on a quantum computer, and \emph{iii)} what will be the runtime, i.e. implementation cost to run a diagonalization double-bracket quantum algorithm.

\subsection{Conditions for diagonalization using double-bracket iterations}
Let us next derive a monotonicity relation for diagonalization DBIs.
It is related to a similar relation known for the GWW flow: This predecessing result can be thought of as a special case of the discussion below, relevant in the continuous limit of the DBI iteration with the canonical bracket \eqref{eq:wcan}, i.e. taking $s_k$ infinitesimally small.

The DBI monotonicity relation will be derived by studying what happens in a single recursion step. Let $\J$ be a hermitian matrix and let
\begin{align}
    \J_s^{\mtiny{(Z)}} = e^{s\wz(\J)}\J e^{-s\wz(\J)}
    \label{eqJevols}
\end{align} 
be its double-bracket rotation for time $s$ by a bracket \eqref{wzbracket}  which involves a diagonal association $\dzz(\J)$.

To analyze the effect of this double-bracket rotation, it is useful to define
\begin{align}
    \o {\J} \equiv \J-\d{\J}\ ,
\end{align}
which is the off-diagonal restriction of $\J$.
Here  $\d{\J}$ is the dephasing channel which is in general different from $\dzz(\J)$.
We will see that $\d{\J}$ will appear nonetheless in the analysis and thus the diagonal association via the dephasing channel gives rise to a bracket that we will call canonical~\cite{kehrein_flow}.

We will use the squared Hilbert-Schmidt norm as the measure of the diagonalization progress and define
\begin{align}
    f(s) = \| \o{\J_s^{\mtiny{(Z)}}}\|_\mathrm{HS}^2\ .
\end{align}
With this definition we have
\begin{align}
f'(0) &= {\partial_s\| \o{\J_s^{\mtiny{(Z)}}}\|_\mathrm{HS}^2}_{\vert s = 0}\\
&= -2 \langle\wz(\J),\wcan(\J)\rangle_\mathrm{HS}\ ,
\label{wzdecreasef}
\end{align}
a result which will be proven just below.
This formula uncovers a large degree of flexibility in devising diagonalization DBIs because we merely need that
\begin{align}
\cos(\theta) =  \frac{\langle\wz(\J),\wcan(\J)\rangle_\mathrm{HS} }{\|\wz(\J)\|_\mathrm{HS}\,\|\wcan(\J)\|_\mathrm{HS}} >0 \ ,
\label{csangle}
\end{align}
i.e, the Cauchy-Schwarz angle $\theta$ needs to be restricted to $[0,\pi/2)$.
We thus see that in the geometry induced by the Hilbert-Schmidt scalar product, the canonical bracket is special in that it is a choice of double-bracket rotation  direction which leads to a step of diagonalization
\begin{align}
f'(0) =& -2 \|\wcan(\J)\|_\mathrm{HS}^2 <0\ 
\label{GWWDBI}
\end{align}
as long as $\wcan(\J)\neq 0$.

It may be handy to say that as long as flow step durations $s_k$ are small enough then the dephasing channel (which is related to the continuous GWW flow) gives rise to unconditionally diagonalizing DBIs.
On the other hand, according to~\eqref{csangle} general diagonal associations give rise to conditionally diagonalizing DBIs, where in each step one may need to check whether the sign should be flipped.
If $\wz$ points away from $\wcan$ then redefining the sign of the diagonal operator $\dzz(\J) \mapsto -\dzz(\J)$ will give rise to a bracket which will have a diagonalizing effect.
The reader should note  that this is generic as opposed to an unconstrained ansatz optimized by brute-force.

In a diagonalization DBI we will have the special case $\J = \h_k$.
Once the evolution \eqref{eqJevols} will have progressed far enough to some time $s_k$ then  the decrease of the magnitude of the off-diagonal terms will turn around into an increase.
To proceed further with diagonalization, we set   $\h_{k+1}= \J_{s_k}^{\mtiny{(Z)}}$ and proceed to a rotation of $\h_{k+1}$ with a new bracket.
This assignment of $s_k$ is locally optimal in $k$ but for a DBI with $K>1$ total steps it may not be globally optimal; it is the greedy strategy for optimizing DBI step durations $\S=(s_1,\ldots,s_K)$.
This is a viable strategy but other may be possible and the following definition allows to further conceptualize the task at hand.
\begin{definition}[$\sigma$-decreasing]
 A  sequence of unitaries $\{\qu_k\}_{k=1,\ldots,K}$  of length $K$
 is \sgd for $\h_0$ if for $k=2,\ldots, K$ 
\begin{align}
  \| \o{\qu_k^\dagger\h_0\qu_k}\|_\mathrm{HS}-\| \o{\qu_{k-1}^\dagger \h_0\qu_{k-1}}\|_\mathrm{HS} < 0\ .
  \label{sgddef}
\end{align}
\end{definition}
This definition captures the operational goals of a quantum device because as long as an evolution is effective in lowering the off-diagonal matrix elements then the precision of implementing the DBI step is a secondary matter and  `errors' may potentially be tolerable.
We next derive the monotonicity relation~\eqref{wzdecreasef} which can aid theoretically in experimental explorations.

\begin{lemma}[\sgd double-bracket rotations]
\label{sgdlemma}
We have
\begin{align}\label{DBImonotonicityapp}
	\| \o{\J_s^{\mtiny{(Z)}}}\|_\mathrm{HS}^2 - \| \o{\J}\|_\mathrm{HS}^2 
 =& -2s \langle \wz(\J),\wcan(\J)\rangle_\mathrm{HS}\nonumber\\&+ O(s^2)\ .
\end{align}
\end{lemma}
As a corollary to this statement, the monotonicity relation~\eqref{wzdecreasef} follows by dividing both sides by $s>0$ and taking the limit $s\rightarrow 0$.
\begin{proof}
Using the integral form of the remainder in the Taylor expansion we have
\begin{align}
    \J_s^{\mtiny{(Z)}} = \J + s[\wz(\J),\J] + \mathcal R_s\ ,
\end{align}
where 
\begin{align}
    \mathcal R_s = \int_0^s \mathrm d r(s-r)[[\wz(\J),[\wz(\J), \J_r]]
\end{align}
and
\begin{align}
    \|\mathcal R_s \|_\text{HS}\le 16 s^2 \|\dzz(\J)\|^2_\text{HS}\,\|\J\|^3_\text{HS} = O(s^2)\ .
    \label{remainderNorm}
\end{align}
We  use that the Hilbert-Schmidt norm is induced by the respective scalar product and insert the Taylor expansion obtaining
\begin{align}\label{DBImonotonicityapp2}
	\| \o{\J_s^{\mtiny{(Z)}}}\|_\mathrm{HS}^2 -\| \o{\J}\|_\mathrm{HS}^2  =&2s \langle \sigma([\wz(\J),\J]),\sigma(\J)\rangle_\mathrm{HS}\nonumber\\
&+ O(s^2)\ .
\end{align}
Here for higher-order terms we made analogous bounds as in~\eqref{remainderNorm} and collected them into $O(s^2)$ with an updated constant.
Next we will use that 
\begin{align}
    \left\langle \Delta\left([\wz(\J),\J]\right),\sigma(\J)\right\rangle_\mathrm{HS} = 0
\end{align}
and 
\begin{align}
    \tr\left([\wz(\J),\J] \sigma(\J)]\right) = 
    \tr\left([\J,\sigma(\J)] \wz(\J)]\right)  \ ,
\end{align}
which involves basic linear algebra discussed in the Appendix~\ref{app:discretization}.
Thus
\begin{align}
    \langle \sigma([\wz(\J),\J]),\sigma(\J)\rangle_\mathrm{HS} = - \langle \wz(\J),\wcan(\J)\rangle_\mathrm{HS}\ ,
\end{align}
where the minus sign appeared due to the Hilbert-Schmidt scalar product definition~\eqref{eqHS}.
\end{proof}
This results shows that if the variational bracket is chosen to be the canonical bracket then for small enough duration of the double-bracket rotation we will obtain a $\sigma$-decrease. 
\begin{corollary}
[Unconditional monotonicity of GWW DBI]
For all steps $k$ in the abstract algorithm~\ref{algorithmDBI} we have
\begin{align}
\| \o{\h_k}\|_\mathrm{HS}^2-\| \o{\h_{k-1}}\|_\mathrm{HS}^2 = -2s_k \|\wcan_k\|_\mathrm{HS}^2 +O(s_k^2)\ .
\label{eq:monotnicity_app}
\end{align}
\end{corollary}

For general variational brackets that arise from some diagonal association we can summarize the possibility of obtaining a \sgd sequence for $\h_0$ as follows.
\begin{corollary}[Generic diagonalizing DBIs]
Let  $\dzz(\cdot)$ be a diagonal association.
For any input Hamiltonian $\h_0$ and $\epsilon>0$, there exists a \sgd DBI with an association $\overline \dzz$ satisfying
    \begin{align}
        \|\overline\dzz(\h_k) - \dzz(\h_k)\| \le \epsilon
        \label{epscond}\ 
    \end{align}
    or
        \begin{align}
        \|\overline\dzz(\h_k) + \dzz(\h_k)\| \le \epsilon
        \label{epscond}\ .
    \end{align}
\end{corollary}
This means that once we choose to rotate with brackets involving diagonal operators, it is the \emph{i)} double-bracket and \emph{ii)} recursive structures that lead to diagonalizing properties.
The choice of the diagonal association influences the performance of the corresponding DBI, however.
\begin{proof}
We proceed inductively, constructing the modified association on the fly.
If for some $k$ the original bracket vanishes $[\dzz(\h_k),\h_k]=0$ then we add a diagonal operator $\A_k$ such that $\overline\dzz(\h_k) = \varepsilon_k\dzz(\h_k)+\A_k$ gives rise to a non-zero bracket which either for $\varepsilon_k = +1$ or for $\varepsilon_k=-1$ will be collinear with the canonical bracket.
Taking $s_k$ sufficiently small ensures a $\sigma$-decrease in that step.
This procedure can be repeated for all steps.
Finally, we note that $\|\A_k\|$ can be small enough so that \eqref{epscond} holds.
\end{proof}
Note, that this proof formalizes the claim that the DBIs proposed have generically \sgd properties. 
If $[\dzz(\h_k),\h_k]\approx0$ then instead of sticking to that diagonal operator and modifying it subtly it may be better to choose a completely different one.
In other words, the corollary is concerned with the fact that diagonal associations generically lead to diagonalizing DBIs but it does not make a statement about the rate of diagonalization.
The diagonalization rate should best be analyzed for the given input Hamiltonian $\h_0$ because different sequences of $\D_0,\D_1,\ldots$ will perform differently for different settings.
The next subsection discusses a special case when $\D_0,\D_1,\ldots$  are all equal and so convergence characteristics can be studied analytically.

\subsection{Brockett-Helmke-Mahony-Moore double-bracket iterations}
Moore, Mahony and Helmke~\cite{moore1994numerical} proposed Lie-bracket recursions, see also the monograph by Helmke and Moore~\cite[Ch. 2.3]{helmke_moore_optimization}, which are equivalent to DBIs proposed here but came to the attention of the author only in last stages of finalizing this manuscript.
Eqs.~\eqref{eq:hk} and \eqref{eq:wkp} anticipate that it may be useful to use different diagonal operator $\D_k$ in each recursion step.
However, Brockett~\cite{brockett1991dynamical, brockett1991dynamical2} proposed a double-bracket flow involving a single diagonal operator $\D^*$ and Moore, Mahony and Helmke~\cite{moore1994numerical} have considered DBIs where each double-bracket rotation involves the same diagonal operator $\D^*$.
Thus it appears that Lie-bracket recursions arose by considering numerical diagonalization schemes, similar to this work which considered initially the GWW flow but concluded that generalizing small-step discretizations to large step durations is worthwhile, now also for quantum computing.
In this subsection we will summarize how the mathematical analysis of the Brockett-Helmke-Mahony-Moore (BHMM) DBI~\cite{moore1994numerical} provides insights about the performance of quantum computing applications of DBIs.

Firstly,  all fixed points $\h_\infty$ of the BHMM DBI have the property $[\h_\infty,\D^*]=0$.
Secondly, assuming that $\D^*$ has a non-degenerate spectrum allows to show that $\h_\infty$ is diagonal.
Thirdly, the only stable equilibrium point is $\h_\infty$ which has the same sorting of eigenvalues as $\D^*$.
Finally, because we are not switching $\D_k$ in every recursion step, the BHMM DBI can be seen analytically to converge.
In particular, close to the equilibrium points the convergence is exponential.

This eventual exponential convergence rests on reaching the basin of asymptotic attraction.
In the quantum case we are expecting exponentially large matrices and as we will see in the next subsection double-bracket quantum algorithms have an exponential implementation cost in the number of recursion steps.
What this means is that in quantum computation implementing the initial non-asymptotic steps is key as otherwise the exponential stability of the diagonal fixed point cannot be salvaged.
The variational formula~\eqref{DBImonotonicityapp} is a first step to systematically understand how to vary and select $\D_k$ in each recursion step.

Furthermore, the analysis of the BHMM DBI highlights the need to avoid spectral degeneracies~\cite{moore1994numerical}.
In near-term quantum computing, $\D_k$ should be simple enough to implement so for example the classical Ising model $\D_k=\sum_{i=1}^L B_i\Z_i+\sum_{i=1}^L J_{i,j}\Z_iZ_j$.
Here the `magnetic' fields and qubit couplings $B_i, J_{i,j}$ are the parameters of the restricted variational ansatz and as a rule of thumb should be chosen such as to avoid spectral degeneracies.

Finally, Moore, Mahony and Helmke~\cite{moore1994numerical} provide an idea to compute the scheduling of DBIs.
The most basic formula is $s^*=1/(4 \|\h_0\|_\text{HS} \cdot\|\D^*\|_\text{HS})$ which can be generalized to a step-dependent scheduling~\cite{moore1994numerical} expressed in norms involving $\h_k$.
This is a quantitative formula for `small-enough' step duration obtained by bounds based on Taylor expansion similar to Eq.~\eqref{DBImonotonicityapp}.
For near-term quantum computing, it should be noted that Hilbert-Schmidt norms of innocuous operators can be exponentially large  and thus this scheduling may be unnecessarily  too pessimistic.

Having said that, together with this double-bracket rotation duration  the BHMM DBI provides an example of an unconditionally diagonalizing DBI and thus facilitates obtaining a feasible solution to the  task of compiling a diagonalization unitary.
Ref.~\cite[Thm.~2.3]{smith1993geometric} provides one more, related, perspective which derives conditions for exponential convergence of iterations related to DBIs.
Next, let us look at the details how to implement DBIs on quantum computers.
\subsection{Double-bracket rotations via group commutators}
\label{GWWQA}
A DBI is a sequence of double-bracket rotations.
We will now discuss how to implement individual double-bracket rotations  on a quantum computer.
A DBI is a priori only an analytical ansatz for a diagonalization iteration but we will also see that it is possible to implement iterated double-bracket rotations, thus arriving at double-bracket quantum algorithms.

The group commutator  prominently features in the Solovay-Kitaev quantum compiling scheme~\cite{dawson2006solovay} and is defined by the formula
\begin{align}
    \vgc(\A,\B) = e^{i\A}e^{i\B}e^{-i\A}e^{-i\B} \ ,
    \label{eq:GCAB}
\end{align}
which involves two hermitian operators $\A,\B$.
It yields an evolution step approximately governed by the commutator~\cite{HigherOrderGCPhysRevResearch.4.013191}
\begin{align}
   e^{-[\A,\B]} =\vgc(\A,\B) + \hat E^{\text{(GC)}} \ ,
    \label{eq:GCapprox}
\end{align}
where  $\|\hat E^{\text{(GC)}}\| \le \|[\A,[\A,\B]]\| +\|[\B,[\B,\A]]\|$ in any  unitarily invariant norm, e.g., the Hilbert-Schmidt norm.
This constant captures the non-cummutative character of the operators involved and a self-contained proof using the technique of local-error bounds~\cite{PhysRevA.90.022305,Su,PhysRevX.11.011020} is provided in Appendix~\ref{app:discretization}.

Using the group commutator formula, we see that if we have oracle access to evolution $\{e^{-it\J}\}$ under an evolution generator $\J$ and to evolutions $\{e^{-it\D}\}$ under its   diagonal association $\D=\dzz(\J)$, then we obtain an approximation to the double-bracket rotation with $\w= [\D,\J]$ 
\begin{align}
  e^{-s\w} =& e^{i\sqrt s \D}e^{i\sqrt s\J}e^{-i\sqrt s \D}e^{-i\sqrt s\J}+ \hat E^{\text{(GC)}} \ .
\label{eq:groupcommutator}
\end{align}
The approximation error can be bounded as
$\|\hat E^{\text{(GC)}}\| \le ( \|[\D,\wz(\J)]\|+\|[\J,\wz(\J)]\|) s^{3/2}$.
Interestingly, this local error bound of the group commutator admits a further upper bound which is proportional to the norm of the corresponding generating bracket.
If we refine this $R$ times and scale the evolution time $r= \sqrt{s/R}$, then we can converge onto the exact double-bracket rotation according to
\begin{align}
  e^{-s\w} =& \left(e^{ir \D}e^{ir\J}e^{-ir\D}e^{-ir\J}\right)^R+  O\left(s^{3/2}R^{-1/2}\right) \ .
\label{eq:groupcommutatorRefined}
\end{align}

The group commutator should be viewed as a means of approximately implementing a DBI step using simpler evolution operations.
It is important to recognize that the group commutator formula \eqref{eq:GCAB} constitutes an elemental component for building more complex quantum algorithms.
It involves forward and backward evolution under two different generators.
If the evolution time is short enough, it yields an evolution step governed by the commutator of these two operators.
Equivalently, we may rephrase this in language of queries~\cite{kothari,Low2019hamiltonian}: 
In order to \emph{query} the forward propagation under the commutator of two operators we must \emph{query} their forward and backward evolution interlaced as in Eq.~\eqref{eq:GCAB}.

Note that the group commutator relies on the capability to invert the evolution.
The signs and ordering of these evolution steps decide which sign the resulting commutator in the generator will have.
We will use in the following that the redefinition $\A = -\J$ and $\B=\D$ yields using Eq.~\eqref{eq:GCapprox}
\begin{align}
  e^{-s\w} =& e^{-i\sqrt s \J}e^{i\sqrt s\D}e^{i\sqrt s \J}e^{-i\sqrt s\D}+ \hat E^{\text{(GC)}} \ 
\label{eq:groupcommutatorReordered}
\end{align}
because again $\w= [\A,\B] = [\D,\J]$.
In other words, the same approximation is obtained, but using a 
different order of queries.

\subsection{Group commutator iterations}

This sets the scene for formulating the \emph{group commutator iteration} (GCI) which is meant as the routine that is going to be easiest to implement on a quantum computer.
It arises by prioritizing what really matters in a DBI.

Firstly, the diagonal operators $\D_0,\D_1,\ldots$ of the inline formulation of DBI around Eqs.~\eqref{eq:hk} and \eqref{eq:wkp} should be simple to implement.
This may not be the case for diagonal associations appearing in abstract algorithm~\ref{algorithmDBI} because those may include the dephasing channel.
Secondly,
DBI is formulated using evolutions under exact brackets.
This is useful for analytical derivations but for implementations one should favor hardware practicality.
Finally, the full group commutator is not needed because there is additional structure allowing to reduce the number of queries to the evolution oracles of the input Hamiltonian $\h_0$.

More specifically, firstly we will assume that for each recursion step we have a classical prescription $d_\mu$  for the diagonal operator $\D_k = \sum_{\mu \in \{0,1\}^{\times L}}d_\mu \Z_\mu$.
When this is a concise description, e.g. it could be the local magnetic field Hamiltonian or the classical Ising model, then this evolution can be efficiently implemented using Hamiltonian simulation or Clifford operations~\cite{nielsen2010quantum}.
Secondly, we will not aim to perform exactly a double-bracket rotation but only approximate it using the group commutator in each recursion step.
So, let us set for $k=0,1\ldots$ 
\begin{align}
    \u_k = \vgc(-\sqrt{s_k}\h_k,\sqrt{s_k}\D_k)
\end{align}
and define
\begin{align}
    \h_{k+1} =\u_k^\dagger \h_k \u_k\ .
\end{align}
Finally, we  use the ordering~\eqref{eq:groupcommutatorReordered} to define a reduced group commutator sequence
\begin{align}
    \ugci(\D,\J) = e^{i\D}e^{i\J}e^{-i\D}
    \label{ugci}
\end{align}
and obtain for $\u_k'=\ugci(\sqrt{s_k}\D_k,\sqrt{s_k}\h_k)$ using commutativity
\begin{align}
    ({\u_k'})^\dagger \h_k \u_k'={\u_k}^\dagger \h_k \u_k =\h_{k+1}\ .
\end{align}
This can be summarized as follows.
\begin{abstractalgorithm}[H]
	\caption{Group commutator iteration~(GCI)}
 \label{aagci}
	\begin{algorithmic}
	  \Input
  \Desc{$K$}{total number of recursion steps}
  \Desc{$\S$}{a vector of $K$ recursion step durations}
  \Desc{$\Dbf$}{a vector of $K$ hermitian and diagonal operators}
  \Desc{$\h_0$}{input Hamiltonian}
  \EndInput
  \Output
  \Desc{${\h_{K}}$}{output Hamiltonian}
    \EndOutput	  
  \State
  \textbf{For }$k = 0$ to $K-1$  \\
  \hspace{\algorithmicindent} 
  \textbf{Set} { $\u_k\leftarrow e^{i\sqrt{s_k}\D_k}e^{i\sqrt{s_k}\h_k}e^{-i\sqrt{s_k}\D_k}$}\\
\hspace{\algorithmicindent} 
      \textbf{Set } {$\h_{k+1} \leftarrow\u_k^\dagger \h_k \u_k$}
 \\
   \textbf{Return } { $\h_{K}$}
\end{algorithmic}  
\end{abstractalgorithm}

A GCI can be modified to converge onto a DBI by investing sufficient runtime overhead as in Eq.~\eqref{eq:groupcommutatorRefined}.
However, the notion of \sgd stresses that for diagonalization it need not be our goal per se to implement a particular DBI but rather we can keep ourselves satisfied with a different iteration, as long as it is $\sigma$-decreasing.
There may be a gain in diagonalization efficacy by implementing a DBI more closely but it may not be worth it, at least in near-term, and GCIs should be considered first for hardware applications.
Thus the abstract algorithm serves to stress the view that GCIs  in their own right are  worthwhile for near-term applications.

Above we defined DBIs as having recursion steps that involve an exact double-bracket evolution while GCIs strictly speaking only approximate such an evolution.
However, it should cause no confusion to conceptually think of GCIs as essentially DBIs because every recursion step of a GCI approximates a double-bracket rotation~\eqref{eq:GCapprox}.
If we use the same sequence of diagonal operators $\D_0,\D_1,\ldots$ then the arising  diagonalization DBI and GCI are  naturally corresponding. 
In other words, for conceptualization one should think in terms of DBIs while for applications one should consider GCIs.

\subsection{GCI transpiler for quantum computing}
\label{GCIsec}

In step $k$ of a DBI we need to make an evolution step governed by $\w_k = [\D_k,\h_k]$.
We can approximate the double-bracket rotation of $\w_k$ using the group commutator but for that we need oracle access to  evolutions under the $k$-th iterated Hamiltonian $\{e^{it\h_k}\}$.
Unless we have it, then  a DBI, or GCI, despite being formulated using unitary operations, is not directly deployable on a quantum device. 

For a GCI to be deployable we need to construct the oracle access to $\{e^{it\h_k}\}$ using the basic oblivious evolution oracles of input Hamiltonian $\{e^{it\h_0}\}$ and of the diagonal operators $\{e^{it\D_k}\}$.
Doing this translation is the role of a transpiler which is tasked with listing all queries to the oblivious evolution oracles in the order in which they are composed in the GCI.
The basic evolutions $\{e^{it\h_0}\}$ and $\{e^{it\D_k}\}$ can be further transpiled into primitive gates of a given device using standard quantum compiling methods.

The key idea here is frame shifting.
Let $\v_k$ be the complete GCI unitary such that 
\begin{align}
    \h_k = \v_k^\dagger \h_0 \v_k\ 
\end{align}
then we will use unitarity to express
\begin{align}
  e^{it \h_k} = \v_k^\dagger e^{it \h_0}\v_k\ . 
  \label{frameshifting}
\end{align}
Thus we can shift the frame using $\v_k$ and query the input evolution under the input Hamiltonian with the desired time $t$.
This means that $\v_k$ composed with the group commutator recursion step together with \eqref{frameshifting} yields $\v_{k+1}$
\begin{align}
    \v_{k+1} &= \v_k \ugci(\sqrt{s_k}\D_k,\sqrt{s_k}\h_k)\label{untranspiledstep}\\
     &= \v_ke^{i\sqrt{s_k}\D_k}\v_k^\dagger e^{i\sqrt{s_k}\h_0}\v_k e^{-i\sqrt{s_k}\D_k}\ .
     \label{transpilingstep}
\end{align}
The transpiler will output the circuit $\UDBI_K$ of $\v_K$ for the final step $K$ if it recursively keeps track of the list of queries needed to implement $\v_1,\v_2,\ldots$.
If $\v_k$ in Eq.~\eqref{transpilingstep} is already transpiled into queries to evolution oracles then we would call that form `unfolded' as opposed to \eqref{untranspiledstep} which is still implicit.
The following transpiling algorithm formalizes this idea.

\begin{abstractalgorithm}[H]
	\caption{$\VDBI^\text{(GCI)}$: GCI circuit transpiler}
	\begin{algorithmic}
	  \Input
  \Desc{$K$:}{total number of recursion steps}
  \Desc{$\S$:}{a vector of at least $K$ recursion step durations}
\Desc{$\{e^{-it\D_k}\}$:}{\quad\quad\quad at least $K$ oblivious evolution oracles with time $t\in\mathbb R$ and diagonal generators $\D_k$ labeled by $k=0,1,\ldots$
  }
\Desc{$\{e^{-it\h_0}\}$:}{\quad\quad\quad oblivious oracle access to evolution under $\h_0$ for time $t\in\mathbb R$}
  \EndInput
  \Output
  \Desc{${\UDBI_K}$:}{ ordered list of queries to the input evolution oracle $\{e^{-it\h_0}\}$ or diagonal evolution oracles $\{e^{-it\D_k}\}$ whose ordered composition yields the circuit of unitary $\v_K$ yielding $\h_K =\v_K^\dagger \h_0\v_K$ satisfying $K$ GCI recursion steps }
    \EndOutput	
\\\hrulefill\\
  \State
  \textbf{If }$K = 0$  \\  \hspace{\algorithmicindent} 
  \textbf{Return } $\id$
  \State
  \textbf{Get} GCI query list for $K-1$ steps\\
  \quad\quad$ \UDBI_{K-1}\leftarrow \VDBI^\text{(GCI)}(K-1, \S, \{e^{-it\D_k}\}, \{e^{-it\h_0}\})$\\
    \textbf{Set} frame-shifted evolution query list as in Eq.~\eqref{frameshifting} \\
  \quad\quad$ \PDBI_{K-1} \leftarrow \left( \UDBI_{K-1}^\dagger, e^{i\sqrt{s_k}\h_0} ,\UDBI_{K-1}\right)$\\
 \textbf{Set} group commutator query list following Eq.~\eqref{ugci}\\
  \quad\quad$ \RDBI_{K}\leftarrow \left( e^{i\sqrt{s_k}\D_{K-1}}, \PDBI_{K-1}, e^{-i\sqrt{s_k}\D_{K-1}}\right)$\\
 \textbf{Compose} query lists $\RDBI_{K}$ and $\UDBI_{K-1}$ as in Eq.~\eqref{transpilingstep}  \\
  \quad\quad $\UDBI_K \leftarrow \left(\UDBI_{K-1},\RDBI_K\right)$    \\  
\textbf{Return} $\UDBI_K$
  \end{algorithmic}
  
\label{algorithmtranspiler}
\end{abstractalgorithm}

Below the quantum algorithm environment describes how to use the above abstract algorithms of a GCI and its transpiling to run GCI on a quantum computer.
For this it is meaningful to make the folowing consideration: if
$\h_K = \v_K^\dagger \h_0 \v_K$ is the result of $K$ steps of GCI then $\ket {\psi_K} = \v_K \ket {\psi_0}$ is such that 
\begin{align}
    \bra{{\psi_K} }\h_0\ket {\psi_K} = \bra {\psi_0} \h_K \ket {\psi_0}
\end{align}
and so it approximates an eigenstate as long as the iteration progressed sufficiently far.
It is important to note that we are considering $\v_K$ and so the output of the transpiler above needs to be reversed because it was working in the convention of iterating the Hamiltonian and not the state.
For this we will consider the list of queries to evolution oracles outputed by the GCI transpiler $\UDBI_K$, we will reverse that list ${reverse}(\UDBI_K)$ and apply to the state.
Finally, we assume that when combining lists, such as in the assignment of $\PDBI_{K-1}$ above, the result is one single list which is automatically `flattened' by iterating through all elements of the lists in the appropriate order and creating a new list.

\begin{quantumalgorithm}[H]
\caption{Diagonalization GCI}
	\begin{algorithmic}
 	  \Input
  \Desc{$K$:}{total number of recursion steps}
  \Desc{$\S$:}{a vector of at least $K$ recursion step durations}
\Desc{$\{e^{-it\D_k}\}$:}{\quad\quad\quad at least $K$ oblivious evolution oracles with time $t\in\mathbb R$ and diagonal generators $\D_k$ labeled by $k=0,1,\ldots$
  }
\Desc{$\{e^{-it\h_0}\}$:}{\quad\quad\quad oblivious oracle access to evolution under $\h_0$ for time $t\in\mathbb R$; its queries are applied to the registers of the input state $\ket {\psi_0}$ }  
  \EndInput
  \Output
  \Desc{$\ket {\psi_K}$ }{  ~   such that if $\h_K = \u_K^\dagger \h_0 \u_K$ is the result of $K$ steps of GCI then $\ket {\psi_K} = \u_K\ket {\psi_0}$ }
    \EndOutput	
    \\\hrulefill\\
  \State
  \textbf{For } $\hat T$ in ${reverse}\left(\VDBI^\text{(GCI)}(K, \S, \{e^{-it\D_k}\}, \{e^{-it\h_0}\})\right)$\\
  \hspace{\algorithmicindent} 
  \textbf{Set $\ket {\psi} \leftarrow \hat T \ket {\psi}$}\\

\textbf{Return} $\ket{\psi_K}\leftarrow \ket{\psi}$
  \end{algorithmic}	\label{quantumalgorithmGCI}
\end{quantumalgorithm}

Note that this quantum algorithm aligns with  three design characteristics and the presence of each of them is beneficial.
It is an iteration which in this specific case has the advantage that its outputs can be studied along the way.
Assuming it is \sgd then its fixed point is useful because it yields diagonalization and as the iteration progresses an increasingly better approximation is obtained.
Finally, the operations employed are natural, namely Hamiltonian  and diagonal evolutions  suffice to implement the procedure.

\paragraph{Runtime as query complexity of the oblivious evolution oracle.}

We will assume that the evolution with the operator $e^{is_k \D_k}$ can be implemented using Clifford operations easily and so will not include it in the runtime analysis.
Instead we will count the number of queries to the oblivious evolution oracle of the input Hamiltonian needed to prepare $\v_k$ and will denote it by
$\NC(k)$.
It satisfies the recursion
\begin{align}
    \NC(k+1) =3\NC(k)+1\ .
\end{align}
The solution of this recursion, with $\NC(1) =1$ as the starting point, gives the final result of queries to the evolution oracle to perform $K$ GCI recursion steps
\begin{align}
    \NC(K) \le (3+o(1))^K\ ,
\end{align}
so an exponential scaling in the number of steps.
Ref.~\cite{QDP} proposed a way to reducing this runtime to polynomial at the cost of appropriately large qubit overheads.

\subsection{Dynamical diagonal association and double-bracket recursions}
In this section we shall explore a class of DBIs with `dynamical' diagonal associations.
They will be dynamical in the sense that without classical knowledge of an evolution generator $\J$, it will be possible to transform the oracle access to the evolutions $e^{-is\J}$ into oracle access to the evolution under the associated diagonal generator $e^{-is \dzz(\J)}$.

One example of a diagonal association map which can be implemented `dynamically' on a quantum computer is the dephasing channel.
It is also known as pinching and amounts to the restriction of an operator to its diagonal.
The dephasing channel is a mixed-unitary quantum channel for any $D$ \cite[Eq.(4.113)]{Watrous} but the most important case is for $L$ qubits, so with $D=2^L$.
Then 
\begin{align}
  e^{-is \d\J} = \prod_{\mu\in\{0,1\}^{\times L}} \p_\mu e^{-is\J/D}\p_\mu + \hat E^{(\Delta)}\ ,
  \label{eq:unitarymixing}
\end{align}
where irrespective of the order in this product  $\|\hat E^{(\Delta)}\|\le 8s^2 \|\J\|^2$, see appendix~\ref{app:quantumalgorithm}.
This means that it is possible to perform the evolution under $\d\J$ for time  $\dell$ by interlacing the evolution under $\J$  in steps $s/D$  with phase flips $\p_\mu$ defined in Eq.~\eqref{eqphaseflips}.
These phase flips can be implemented using Clifford operations~\cite{nielsen2010quantum}.

Again, let us rephrase this as an elemental quantum algorithmic component in the language of queries:
We obtain \emph{query} access to evolution governed by the diagonal association given by the pinched generator by repeatedly \emph{querying} evolution under the input generator and conjugating it by phase flips.
This query access is `dynamical' because it obliviously turns Hamiltonian evolution to evolution under the pinched generator.
Note the difference to the above not `dynamical' case of classical prescribed diagonal operators because then the diagonal association depends only on classical knowledge that informs the choice of $\D_0,\D_1,\ldots$.

Finally, we can further generalize this to say that the diagonal association should admit an approximation by a product formula.
More specifically, we want there to exist a finite sequence of $I+1$ unitaries $\hat Q_i$ such that
\begin{align}
  e^{-is \dzz(\J)} =& \hat Q_0 e^{-is_1\J}\hat Q_1 e^{-is_2\J}\hat Q_2\ldots e^{-is_I\J}\hat Q_{I}\nonumber\\&+ \hat E^{(\dzz)} 
      \label{eq:productassociation}
  \end{align}
  with the error term satisfying $\|E^{(\dzz)}\| = O(s^p)$ with $p>0$.
The dephasing channel is an example of such a diagonal association with $\hat Q_i$'s given by appropriate phase flips and $I=2^L$, $s_1=s_2=\ldots=s_I = s/D$  so that  $p=2$.
Unitary-mixing channels in general admit a similar approximation by a product formula.

The product formula approximation  facilitates covariance under shifting of the frame, a property needed for performing multiple recursion steps on a quantum device.
Let $\hat L$ be a unitary which implements a shifting of the frame and consider any generator $\J$ then we have
\begin{align}
      e^{-is \dzz(\L\J\L^\dagger)} =& \hat Q_0 \L e^{-is\J/D}\L^\dagger \hat Q_1 \L e^{-is\J/D}\L^\dagger \hat Q_2\ldots\nonumber \\
      &\ldots \L e^{-is\J/D}\L^\dagger \hat Q_{I+1}+ \L \left( \hat E^{(\dzz)}\right) \L^\dagger
      \label{eq:covariant}
  \end{align}
  and so we automatically get a prescription for the implementation of the diagonal association with the same error.

Covariance under frame shifting allows to progress with a recursion.
If $\v_1=\vgc(\sqrt{s}\dzz(\h_0),\sqrt{s}\h_0)$ is the first recursion step, then we set $\L = \v_1$ and the frame-shifting covariance formula \eqref{eq:covariant} can be used in the group commutator to 
implement the next recursion step
\begin{align}
 \v_2=\vgc(\sqrt{s}\dzz(\v_1^\dagger\h_0\v_1),\sqrt{s}\v_1^\dagger\h_0\v_1)   \ .
\end{align}
Repeating this `unfolding' results in the hardware prescription for the GCI which includes a `dynamical' diagonal association that can be approximated by a product formula~\eqref{eq:productassociation}.

\paragraph{Query complexity for dynamical diagonal associations.}

The dephasing channel is an instructive example of a `dynamical' diagonal association covariant under shifting of the frame because it is involved in the canonical bracket.
According to the monotonicity relation~\eqref{GWWDBI} a double-bracket rotation involving the canonical bracket is \sgd for sufficiently small rotation duration $s$.
This is unconditional, as opposed to taking an arbitary diagonal operator $\D_k$ which generates a diagonalizing double-bracket rotation conditioned on choosing the sign correctly (either $+\D_k$ or $-\D_k$ will yield a variational bracket collinear with the canonical bracket).

However, while a variational operator $\D_k$ may be easy to implement, e.g. by Clifford operations, for `dynamical' diagonal associations we have an approximation by a product formula \eqref{eq:productassociation} through $I$ queries to the evolution oracle of the input Hamiltonian.
Thus the queries to $e^{\pm i\sqrt{s_{k-1}} \dzz(\h_{k-1})}$ dominate and so
\begin{align}
    \NC_I(K) = O(I)^K\ .
\end{align}

When the diagonal association map is given by the dephasing channel we have the GWW DBI and $I = 2^L$ where $L$ is the number of qubits so the runtime of $K$ GWW DBI steps is exponential in both $K$ and $L$.
Indeed, in the worst-case, input matrices can be dense and the full unitary-mixing representation of the dephasing channel is needed to pinch them. 
As will be discussed below conjugation by a polynomial number of phase flips suffices for sparse evolution generators and then   $\NC_I(K) = (\text{poly}(L))^K$.
\begin{figure*}
  \includegraphics[width=\linewidth]{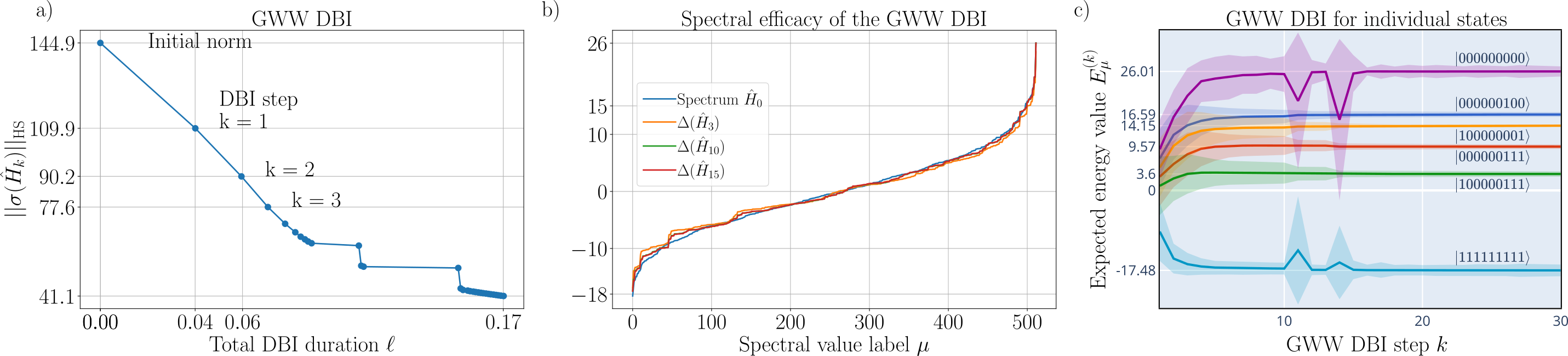}
  
  \caption{GWW DBI aiming to approximate eigenstates of  the transverse and longitudinal-field Ising model $\h_\text{TLFIM}$ of $L=9$ qubits.
  {a)} During the DBI, the decrease of the Hilbert-Schmidt norm of the off-diagonal restriction progresses in several stages.
  Optimization of the DBI step durations $s_k$ results in some steps yielding larger overall decreases than others which can be explained by lifting of degeneracies.
  b) The energy spectrum of the Hamiltonian   overlaid by the sorted entries of diagonal restrictions $\d{\h_{k}}$ after $k=3,15$ and $30$ steps.
  It hints that later DBI stages affect the states in the middle of the spectrum where the energy gaps tend to be small.
  In particular  approximate degeneracies which show up as steps of $\d{\h_{3}}$ become lifted through further recursion steps, as evidenced by the softened progression of $\d{\h_{15}}$.
{c)} Expected values of the energy $E_\mu^{(k)}=\sandwich \mu {\h_k}\mu$ computed for a selection of computational basis states $\ket\mu$ which is equivalent to applying the  DBI cicuit to that state and evaluating the expected value of the input Hamiltonian.
The two fully polarized states were selected because they tend towards states at the edges of the spectrum.
Additionally the selection includes four other computational basis states with small energy fluctuations after $k=15$ DBI steps.
Expected values of the energy $E_\mu^{(k)}$ change mostly during the first few  DBI steps and then stabilize to a constant level.
The colored areas around $E_\mu^{(k)}$ show the measurable energy fluctuation $\Xi_k(\mu)$ which become small for sufficient number of DBI steps.
Just $k=5$ steps yield a surprisingly good approximation to a very low-energy eigenstate and thus seems  appealing for potential explorations of the GWW DBI in a quantum device.
After $k=10$ recursion steps the maximally polarized states show big variations both in expected energy value and fluctuation which should be linked with coupling to other states as these are being reorganized.
These variations stem from GWW DBI being a global diagonalization method aiming to decrease  the  norm of the off-diagonal restriction $\|\o{\h_k}\|_\text{HS}$ so most states lower their energy fluctuation but some may temporarily increase.
}
\label{figeigsGWW}
\end{figure*}

\section{Efficacy of diagonalization DBIs}
\label{GWWnumsec}

The abstract algorithm of the GCI circuit transpiler highlights an important aspect, namely GCIs result in recursive circuits.
Thus as a recursive GCI progresses, each additional recursion step becomes more  costly than the previous ones.
In order to get general understanding of the expressive power of DBIs, in this section we turn our attention to DBIs and abstract from intricacies of group commutator approximations involved in GCIs.
More specifically, the aim is to address what eigenstate approximations can be achieved with only  a few recursion steps.

\subsection{Numerical examples of DBIs}
\paragraph{Numerical example: Preparing eigenstates.}
We will begin the assessment of the efficacy of diagonalizing DBIs by exploring the GWW DBI which is a DBI where $\D_k = \Delta(\h_k)$, i.e., at every step we use the diagonal association given by the the dephasing channel~\eqref{eq:pinch}.
We then use that to compute the canonical bracket~\eqref{eq:wcan} at every step and proceed with the DBI recursion~\eqref{eq:hk} and~\eqref{eq:wkp}.

This will provide intuition about how many recursion steps may appear to be needed in proof-of-principle explorations.
The analysis  includes optimization of recursion step durations by a greedy protocol: For every recursion step we find the step duration which results in the maximal reduction of the Hilbert-Schmidt norm of the off-diagonal restriction of the respective $\h_k$.

Let us consider the GWW DBI applied for the transverse and longitudinal-field Ising model
\begin{align}
  \h_\text{TLFIM} = J_X\sum_{j=1}^{L-1} \X_j \X_{j+1} + \sum_{j=1}^{L} \left(\Z_j + \X_j\right)\ .
\end{align}
For the numerical calculations we will set $J_X=2$ and $L=9$.
Regarding this choice of the model, quantum simulations are often explored by means of the simpler transverse-field Ising model which does not feature the longitudinal on-site field $\X_j$ aligned with the Ising interaction $\X_j\X_{j+1} $ and thus involves two rather than three local Hamiltonian terms.
However, then the model is exactly solvable by a mapping to free fermions~\cite{tFI}, while in presence of the longitudinal field such solution is not directly available.
In other words, considering the logintudinal field $\X_j$ is not a big experimental complication and may make the obtained results more representative of generic interacting spin systems. 
\begin{figure*}
  \includegraphics[trim = 9cm 0 9cm 0, clip,width=\linewidth]{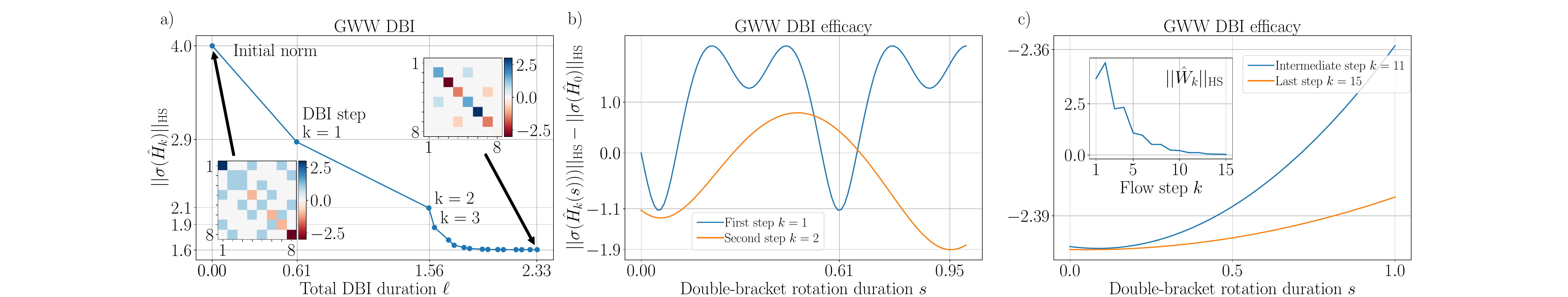}
  \caption{Example of GWW DBI applied to the transverse-field Ising model on $L=3$ sites.
 { a)} In each step $k=1,\ldots, 15$  the canonical bracket \eqref{eq:wcan} generates the evolution and the flow step duration is chosen to be giving the maximal decrease of the norm of the off-diagonal restriction $\|\o{\h_k}\|_\text{HS}$. After about 5 steps of a rapid decline the norm decrease begins to saturate and essentially ceases completely. This is due to the convergence to the block-diagonal form of the Hamiltonian indicated in the upper-right inset.
 { b)} For small  DBI rotation durations the norm decreases linearly in accordance with the monotonicity relation~\eqref{GWWDBI}. Evolution beyond the first minimum can yield another deeper local minimum, e.g., for $k=2$ the second minimum has been chosen over the first.
 If the DBI rotation duration is chosen wrongly, the norm of the off-diagonal restriction may increase, e.g., for $k=1$ most $s$ do not yield a diagonalizing step.  c) The monotonicity relation~\eqref{GWWDBI} implies that there is at least one negative minimum but as the iteration progresses it becomes shallow. Additionally, as shown in the inset, the norms of the canonical brackets decay and so the diagonalization rate slows down because the respective double-bracket rotations havae a weakened influence. }
  \label{fig:illustration}
\end{figure*}

 Fig.~\ref{figeigsGWW}{a)} shows the monotonic decrease of the Hilbert-Schmidt norm of the off-diagonal restriction $\|\o{\h_k}\|_\text{HS}$ as the  GWW DBI progresses.
 The duration of flow steps has been optimized to mitigate against the recursive character of the proposed implementation which is encoded on the horizontal axis.
Fig.~\ref{figeigsGWW}{b)} shows that the diagonal restrictions $\d{\h_{15}}$ and $\d{\h_{30}}$ after $k=15$ and $30$ steps reproduce the energy spectrum of the Hamiltonian obtained by exact diagonalization.
More broadly, this suggests that the GWW DBI is an appealing pre-conditioning method, i.e., unitary transformation increasing the diagonal dominance of an input Hamiltonian.
This can be useful for further quantum information processing because other protocols may have an improved runtime then, e.g. runtime of ground-state preparation using quantum phase estimation depends on the initial overlap with that ground-state~\cite{lin2020near}.

Evaluating expected values of the flowed Hamiltonian $E_\mu^{(k)}=\sandwich \mu {\h_k}\mu$ for a choice of computational basis states $\ket\mu$ is a basic way of extracting physical information from a quantum device.
In an experiment, $E_\mu^{(k)}$ can be estimated in the Schr\"odinger picture by evolving an initial state forward $\ket\mu \mapsto \v_k \ket \mu$ and estimating the expectation value of the input Hamiltonian $\h_0$.
Fig.~\ref{figeigsGWW}{c)} shows examples of initial states $\ket \mu$ whose expected values of energy stabilize to a constant level after relatively few GWW DBI steps.

This investigation revealed the maximally polarized state $\ket{ 11\ldots1}$ to be the candidate of choice for potential explorations using a quantum device.
This is because it tends towards a low-energy state and just $k=5$ steps result in a narrow energy fluctuation.
Thus based on this state a near-term device could probe the low-energy physics of the model.

Finally, the later stages of the iteration result in splitting narrowly positioned energy levels which has to happen in order to diagonalize the quantum many-body system.
Such splitting may be detectable experimentally when the energy levels split beyond the respective range of energy fluctuations but ultimately this effect may be sensitive to implementation imperfections and thus hard to observe in a near-term quantum device.

In summary, Fig.~\ref{figeigsGWW} presents numerical results instructing about possible avenues of exploring diagonalization DBIs in a quantum device.
We find that  even if the norm of the off-diagonal restriction may appear large, selected states may be close to eigenstates.
This means that  for that limited number of recursion steps the DBI achieves local rather than global eigenstate-by-eigenstate diagonalization.
In particular, the maximally polarized state $\ket{ 11\ldots1}$ is an example of a state that should be considered for  exploring DBIs experimentally because it results in very low expected value of energy and reaches a narrow fluctuation range within a small number of steps.

\paragraph{Numerical example: Step duration optimization landscape.}
Fig.~\ref{fig:illustration} presents a minimal example of the GWW DBI for the transverse-field Ising model
\begin{align}
  \h_\text{TFIM} = J_X\sum_{j=1}^{L-1}\X_j \X_{j+1} +\sum_{j=1}^{L} \Z_j\ 
\end{align}
on $L=3$ sites and Ising coupling $J_X=1$.
For this small system we can directly inspect a color-coded visualization of the Hamiltonian matrix, see Fig.~\ref{fig:illustration}a).
The initial few steps lead to a rapid decrease of the norm of the off-diagonal terms, however, further down the line even taking very many steps does not allow to fully diagonalize the problem and the iteration results in block-diagonalization.
In such special cases, it is possible to choose a different diagonal association and pursue diagonalization further.

The upper-right inset shows that the Hamiltonian matrix has been put into a block-diagonal form and the $1\times 1$ blocks signify eigenstates.
This can be better understood by considering the dependence of the off-diagonal norm on the flow step duration $s$.
For the initial two steps, Fig.~\ref{fig:illustration}b), we find a rapid decrease of the norms of the off-diagonal restrictions.
However, for the later steps, Fig.~\ref{fig:illustration}c), the diagonalizing iteration loses efficiency and saturates which is due to increasingly more shallow minima of the off-diagonal norm in dependence on the flow step duration.
This is a distinctive feature of the method: The iteration aims to decrease the canonical bracket which may not result in full diagonalization in presence of spectral degeneracies~\cite{BROCKETT199179,kehrein_flow}.
Full diagonalization can be obtained by a variational search for a diagonal association different from dephasing.
\begin{figure}
  \includegraphics[trim = .2cm 0.5cm 1.5cm 1cm, clip,width=\linewidth]{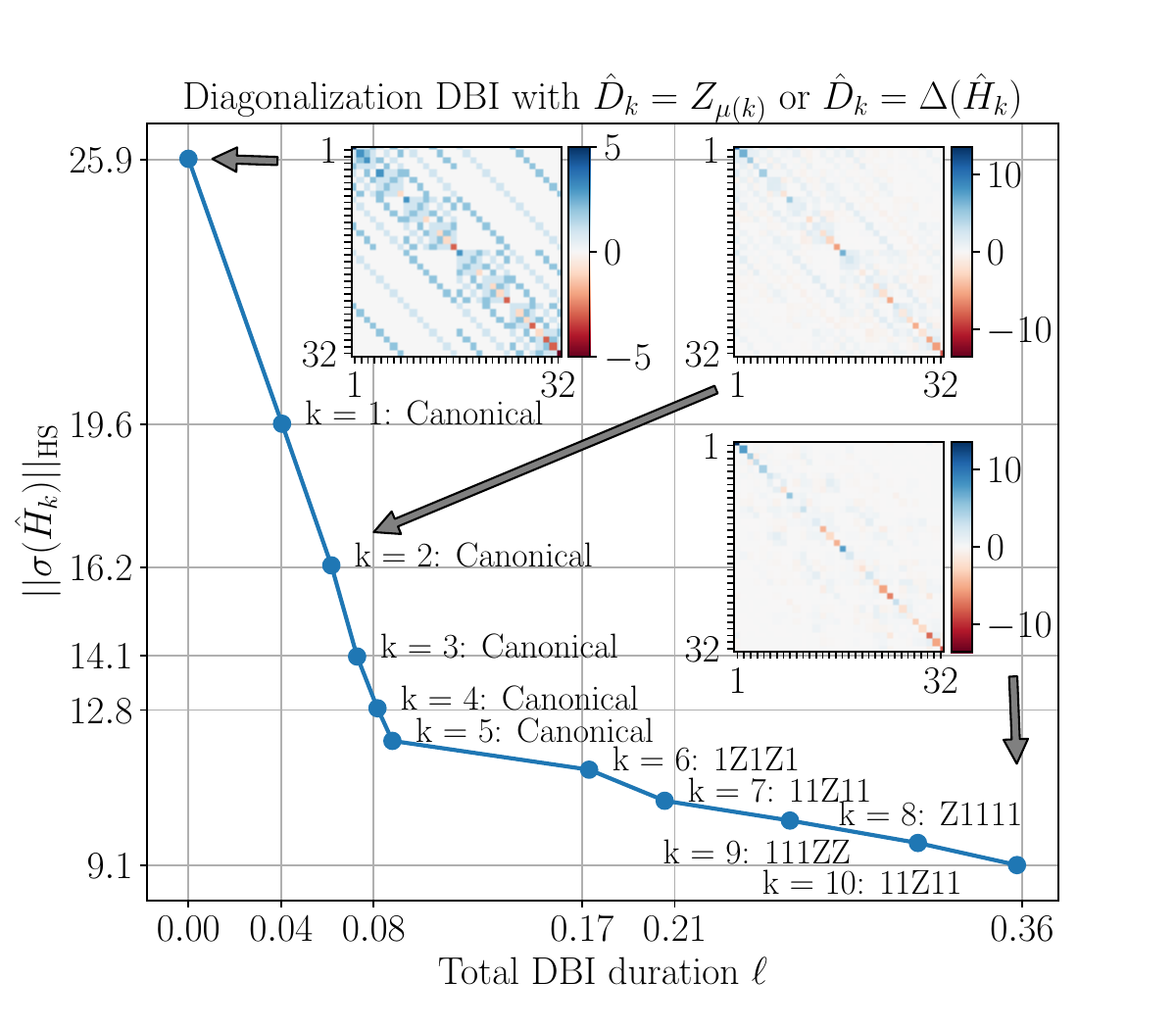}
  \caption{Illustration of variational DBI choosing either the canonical bracket with $\D_k = \d{\h_k}$ or $\D_k = \p_{\mu}$, whichever yields a larger decrease of $\|\o{\h_k}\|_\mathrm{HS}$ applied for TLFIM on $L=5$ sites and $J_X=2$.
  This example shows that by going beyond the dephasing channel which gives rise to the GWW DBI, it is  possible to proceed with the diagonalization because for $k\ge 6$ the product operators yield better diagonalizing rotations.}
  \label{figL5}
\end{figure}

\paragraph{Numerical example: Variational $\D_k$ for DBI.}
Fig.~\ref{figL5} shows the numerical results for flowing the Hamiltonian after optimizing the flow step duration for all generators in a selection that included the canonical bracket and $\wmu$ for all possible choices of phase flip operators $\Z_\mu$ in Eq.~\eqref{eqwmu}.
The insets show the full matrices of the initial and flowed Hamiltonians, demonstrating gradually increasing diagonal-dominance of the flowed Hamiltonians.
Due to rapidly growing Hilbert space dimension, the visualization of the full Hamiltonian matrix becomes impractical for increasing system sizes so a state-selective analysis of the type shown in Fig.~\ref{figeigsGWW} should be preferred but appendix~\ref{app:numerics} presents an analogous plot for $L=7$ qubits.

\begin{figure*}
  \includegraphics[trim = 0 0 0 0, clip,width=\linewidth]{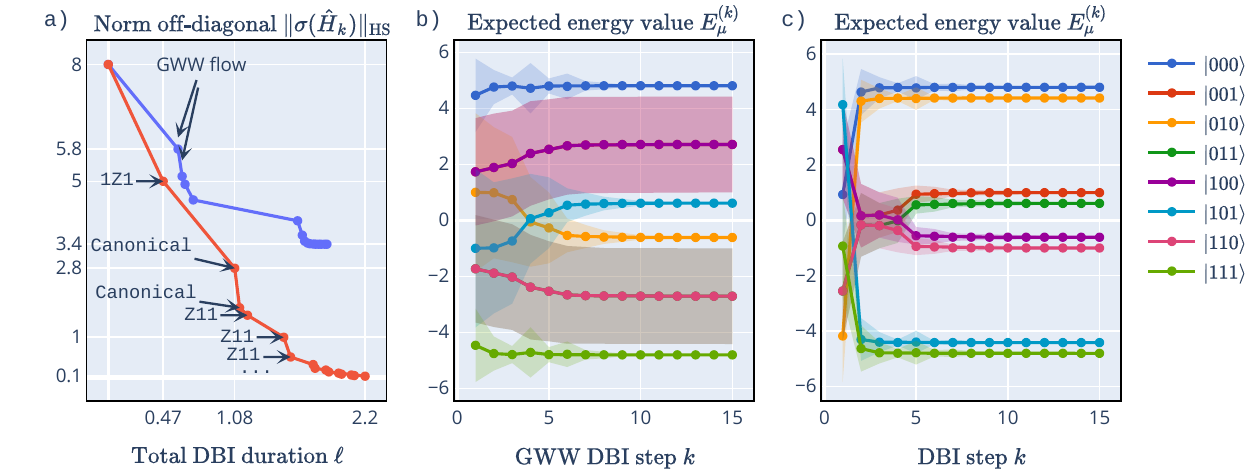}
  \caption{Comparison of variational and GWW DBIs. 
  a) The Hilbert-Schmidt norm of the off-diagonal restriction decreases more rapidly for the variational DBI.
  The appearance of the variational bracket $\w^{(010)}$ has a significant impact, namely to lift degeneracy.
  b)  Each line is the expected energy value during the progression of the GWW DBI evaluated for all $2^L=8$ computational basis initial states.
The large energy fluctuation of the states with intermediate expected energy value can be explained by convergence towards a superposition of two degenerate states.
c) The degeneracy of the initial computational basis states becomes fully lifted after $k=4$ flow steps and the expected energy fluctuation becomes negligible in the course of the DBI.
  }
  \label{figL3diag}
\end{figure*}

\paragraph{Numerical example: Full diagonalization.}
Fig.~\ref{figL3diag} shows a comparison between  the GWW DBI and a DBI with variationally selected brackets.
The decrease of the Hilbert-Schmidt norm of the off-diagonal restriction is  faster for the DBI exploring more types of double-bracket rotations.
Additionally, after only $k=4$ steps the degeneracy of the initial computational basis states is lifted.
Thus, the variational DBI  lifts all degeneracies and suppresses the off-diagonal interaction terms faster.
It appears that the canonical bracket has the role of \sgd and $\wmu$ brackets lifts degeneracy.
It should be noted that for the extremal fully polarized states the GWW flow converged faster than in the variational diagonalization DBI.

\paragraph{Numerical example: 2d Anderson model.}
The Anderson model is a minimal toy-model for studying disorder-induced breakdown of transport.
Using the notation 
\textless$x,y$\textgreater  ~to denote all nearest-neighbor positions on the square $L \times L$ lattice $[L]^{\times 2}$, the Anderson Hamiltonian reads
\begin{align}
    \h_\text{AI} = \sum_{<x,y>}\a_x^\dagger a_y + \sum_{x\in[L]^{\times 2}}B_x\a_x^\dagger \a_x\ ,
\end{align}
where $B_x\in [0,1.5]$ are independent and identically distributed according to the uniform distribution.
Here we consider $\a$ to be the annihilation operator of either a fermionic system where
$\a_x \a_y^\dagger + a_y^\dagger a_x = \delta_{x,y}$ or of a bosonic system when $\a_x \a_y^\dagger - a_y^\dagger a_x = \delta_{x,y}$ for all $x,y\in [L]^{\times 2}$.
This Hamiltonian can be written as $\h(h) = \sum_{x,y} h_{x,y} \a^\dagger_x \a_y$ for an appropriate matrix $\h=\h^\dagger$ collecting all couplings, e.g., for $\h_\text{AI}$ we have $h_{x,y}=1$ for nearest-neighbor sites, $h_{x,x}= B_x$ and zero otherwise.

\begin{figure*}
    \centering    \includegraphics[width=\linewidth]{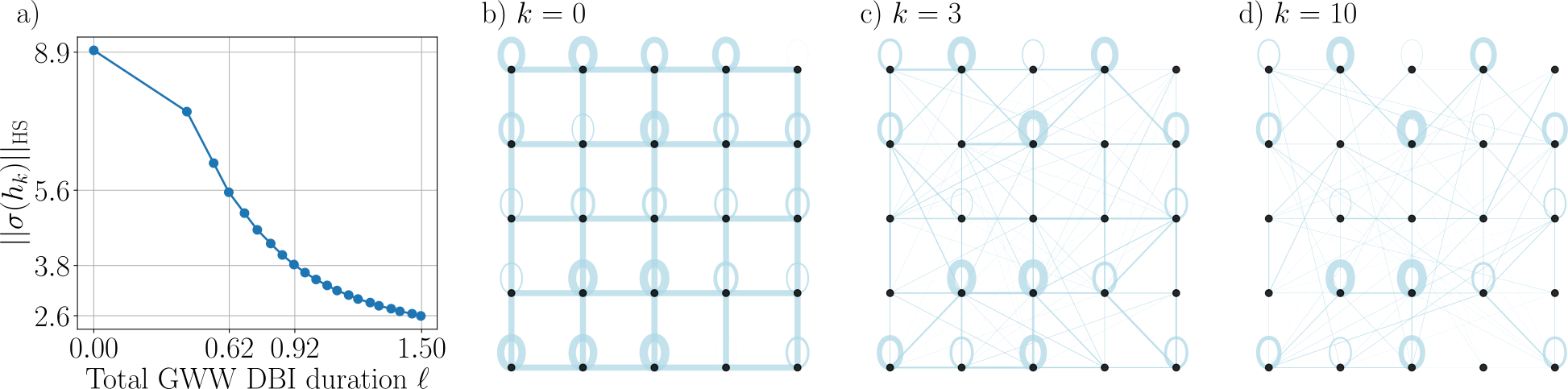}
    \caption{Results of the GWW DBI for the 2d Anderson model. a) Decrease of the Hilbert-Schmidt norm of the off-diagonal restriction $\|\o{h_k}\|_\text{HS}$. 
    b) A representation of the initial coupling matrix $h_0$ of the input model $\h_0=\h(h_0)$ as a 2d square graph with edge widths encoding the coupling strengths. 
    According to the 2d geometry, there are only nearest-neighbor couplings and additionally we consider random uniform disorder which is encoded in the width of the self-links.
    c) Just a handful of DBI steps suffices to nearly decouple the neighbors.
    d) As the DBI progresses, distant sites become coupled but the strength is negligible and the system Hamiltonian becomes nearly diagonal, with sites of the lattice being nearly decoupled.
    The width of the self-links encodes the approximation to the single-particle eigenvalues.
    }
    \label{figAI}
\end{figure*}

For both fermions and bosons, these coupling matrices generate the Gaussian representation of the unitary group and in particular we have the Lie-algebra homomorphism
\begin{align}
    \wcan(\h(h)) &= [\d{\h(h)},\h(h)]\\
    &=\h([\d {h},h]) = \h(\wcan(h))\ .
\end{align}
Thus, we can consider a Gaussian DBI in the single-particle space with
\def\miniwcan{ w}
\begin{align}
    \miniwcan(h) = [\d h,h]\ 
\end{align}
such that
\begin{align}
    h_{k+1} = e^{s_k\miniwcan_k} h_k e^{-s_k\miniwcan_k}
\end{align}
where $\miniwcan_k=[\d{h_k},h_k] $ and $h_0$ is the initial coupling matrix.
If we solve the Gaussian DBI then we obtain the full Hamiltonians acting on the Hilbert space by setting $\h_k = \h(h_k)$.

Fig.~\ref{figAI} shows the results of performing a GWW DBI with optimized recursive step durations $s_k$ for a randomly sampled instance of the Anderson Hamiltonian.
Panel a) shows the decrease of the norm of off-diagonal restrictions and panels b,c,d) show graph representations of the coupling matrices during the iteration.
This routine can in principle be performed for much larger systems (even $100 \times 100$ on a regular laptop) but visualization would become difficult due to clutter.
For the modest $5 \times 5$ lattice size, we see all elements of the diagnalization via a DBI: Initially the coupling between nearest-neighbor sites are substantial and successive double-bracket rotations lower these coupling strength.
In the process further off-diagonal couplings become generated but carry almost negligble strength.
\subsection{Heuristic strategies for lowering implementation cost}

\begin{figure}
  \includegraphics[trim = 0cm 0 0cm 0, clip,width=.99\linewidth]{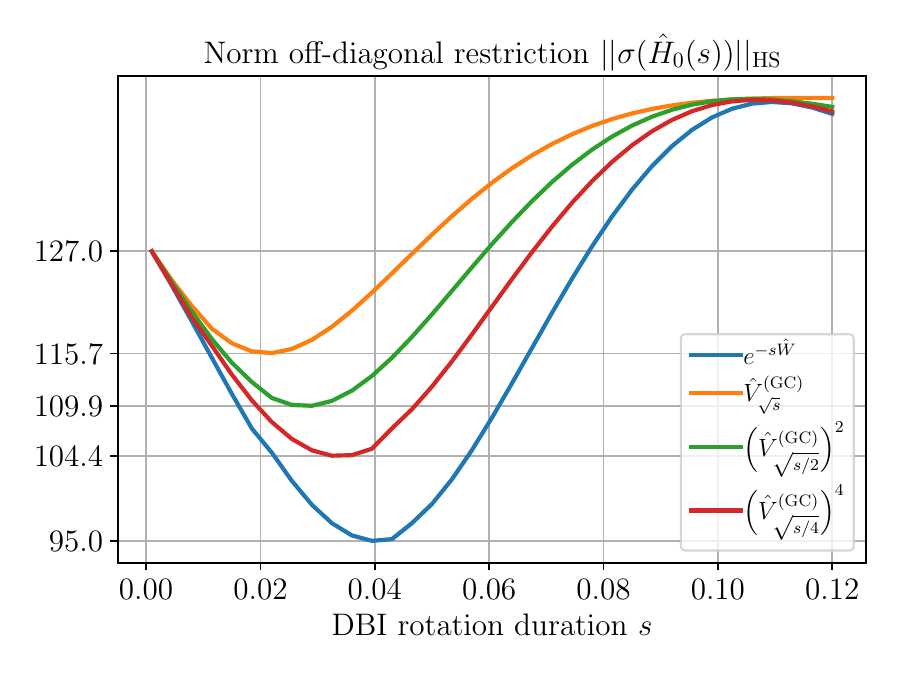}
  \caption{Repetitions of a group-commutator rotation step allow to increase $\sigma$-decrease. In most common applications of the group commutator formula~\eqref{eq:groupcommutator} a small flow step duration~$s$ is considered because then an evolution governed by a commutator is approximated. However, heuristic group commutator steps for sizeable durations $s$ are examples of \sgd unitaries. The $\sigma$-decrease is largest in the limit of repeating many times the group commutator parametrized by a short flow step duration~$s$.}
  \label{figmoresteps}
\end{figure}

\paragraph{Fewer recursive repetitions: Repeating GCI steps.}
The numerical examples presented results for DBIs but for hardware implementation, the GCI quantum algorithm would be used.
Given, $\D_k$ one may optimize $s_k$ to give the highest $\sigma$-decrease from a single group commutator rotation~\eqref{eq:GCAB}.
However, it turns out that it may be more useful to repeat twice the same group commutator formula to target more closely a double-bracket rotation of the corresponding DBI.

Fig.~\ref{figmoresteps} illustrates the  $\sigma$-decrease obtained by targeting approximations of a double-bracket rotation so $\bigl(\vgc_{\sqrt{s/m}}\bigr)^m\approx e^{-s\w}$.
The model is the TLFIM as in Fig.~\ref{figeigsGWW}  on $L=9$ sites with $J_X=2$.
The results show that the GWW double-bracket rotation gives a larger $\sigma$-decrease than rotations by the group commutator.

This means that in a GCI, instead of performing a single group commutator step to obtain $\h_{k+1}$ and proceed onwards, one should rather obtain $\h_{k+1}$ by several repetitions of the group commutator involving $\D_k$ and $\h_k$.
This should be more efficient because prematurely switching the generator $K$ times will necessitate $\NC(K)=(3+o(1))^K$ queries to the evolution under the input Hamiltonian while repeating the same GCI step $K$ times would scale linearly $\overline \NC(K)=O(K)$
For this reason, in practice, this heuristic strategy should be receive preferential consideration.

\paragraph{Efficient pinching: Making use of locality.}

The pinching scheme based on Eq.~\eqref{eq:unitarymixing} uses exponentially many phase flips.
This may be needed in the worst-case, i.e., for dense matrices~\cite{childs2010limitations} but is not practically applicable. 
This number can be reduced if $\o{\h_k}$ is sparse, i.e., a sum of a small number of Pauli operators.
Appendix~\ref{app:sparsity} shows that restricting the product in Eq.~\eqref{eq:unitarymixing} to a smaller number of phase flips selected uniformly at random allows to  with high probability accurately approximate the evolution under a pinched generator.
Local input Hamiltonians are sparse and for early stages of DBIs Lieb-Robinson bounds apply and thus efficient pinching heuristics can be employed.
Note, that Hastings conjuctured that Lieb-Robinson bounds do not hold for all double-bracket flows~\cite{hastings2022lieb} so sparsity may be reduced very quickly also in DBIs.

More specifically, the App.~\ref{app:sparsity} provides a bound on how many phase flips suffice to achieve a desired approximation.
This depends quadratically on the sparsity $S$ of the given Hamiltonian.
For initial Hamiltonians we have $S=2L$ for $\h_\text{TFIM}$ and $S=3L$ for $\h_\text{TLFIM}$.
In both these cases the sparsity is linear in the number of qubits $S = aL$, with $a>0$,  and the bound in  App.~\ref{app:sparsity} ensures an accurate pinching approximation using $R = a^2 L^{2} J^2 \epsilon^{-2}\log(2aL/\delta)$ phase flips, where $J$ is the largest coupling strength, $\epsilon$ is the required pinching precision and $1-\delta$ is the success probability.

In general the sparsity is the number of non-negligible off-diagonal interaction terms and cannot be computed a priori for evolved Hamiltonians beyond using arguments based on Lieb-Robinson bounds.
Fig.~\ref{figeigsGWW}a) shows an example where surprisingly short DBI steps have been locally optimal.
By comparison, if one would consider a quench evolution of duration of the ordered as considered in the DBI, then the Lieb-Robinson cone would be much smaller than the total system size.

Regarding runtime, we may additionally consider the intuition provided by the notion of $\sigma$-decreasing: Accurate pinching is only required for approximating the GWW DBI but other generators may also decrease the off-diagonal terms.
One may choose to explore heuristics wherein pinching uses a number $R$ of phase flips independent of the encountered sparsity.
If $R$ depends polynomially on the system size, the  runtime would  still be exponential in the number of DBI steps $K$ due to the recursion but if $K$ is kept constant it would be efficient in the system size.
Finally, on can  experimentally check if $R$ independent of the system size is enough to 
obtain a \sgd DBI.

\paragraph{Classical-quantum hybrids: Precompiling.} 
GWW flow has received a wealth of attention by quantum many-body physics theorists~\cite{kehrein_flow,wegner2006flow}.
These existing classical algorithms can aid the quantum computation wherein the first steps of the GWW DBI could be computed numerically and the classical prescription of the canonical brackets would be implemented directly by Hamiltonian simulation.
More specifically, a typical approximation is to restrict $\h_k$ to be a quantum Ising model with couplings depending on the recursion step $k$. 
Each of such models then can give rise to a respective bracket $\w_k$ which can be used for a double-bracket rotation.

For example, quantum computed dephasing is not needed to run the first step of the GWW DBI for Ising models because we know $\d{\h_\text{TFIM}}$.
Generalizing this to precompiling, which replaces the DBI by direct implementation of $s^{-s_k\w_k}$, is efficient in terms of the sparsity of the DBI Hamiltonians $\h_k$ and, if viable, it should be preferred over running the full scheme which includes dephasing.
On the other hand, once the assistance of the classical computations will reach its limits the actual quantum computation can be used to extend the diagonalizing iteration further.

Note, that for systems with sizes considered in the numerical examples presented in this work essentially the entire GWW DBI can be precompiled.
For $L=9$  the decomposition of any $\h_k$ or $\w_k$ involves at most $2^{2L} \approx  2.6\times 10^5$ different Pauli operators.
This is the worst case upper bound on the number of interaction terms that may be necessary to precompile steps of the classically prescribed GWW DBI circuit.
Of course, for first DBI steps that number is much smaller and it may be a very useful heuristic to perform Hamiltonian simulation for first few steps of the GWW DBI
\begin{align}
  \v^\text{(pGWW)} = e^{-s_1\w_1}\ldots e^{-s_N\w_N}\ .
  \label{eq:precompiling}
\end{align}
When precompiling, one should directly include the heuristic presented above and locally optimize \sgd flow step durations $s_1,\ldots, s_N$.
Inspecting the precompiled GWW DBI circuit \eqref{eq:precompiling} makes apparent its practicality in near-term: It suffices to just simply use Hamiltonian simulation for Hamiltonians $-i\w_k$ as generators.
Doing this accurately would achieve the same results as presented, e.g., in Fig.~\ref{figeigsGWW} and thus GWW DBI allows to compile approximate diagonalization circuits.

Precompiling will eventually cease being practical  but further quantum computation starting from the shifted frame
\begin{align}
  \h_0 \rightarrow \J_0 = \v^{\text{(pGWW)}\dagger}\h_0 {\v}^\text{(pGWW)}
\end{align}
may be possible.
For some fixed number of extra DBI steps a \sgd DBI may proceed via quantum computation: Starting from the shifted frame a variational assignment of the DBI step duration such that $\J_s = \vgc_s(\J)^\dagger \J \vgc_s(\J)$ is maximimally $\sigma$-decreasing can be considered.

\paragraph{Comments on the necessary number of recursion steps.}

The GWW flow acting on an arbitrary product state will yield an approximation to an eigenstate of $\h_0$ and for local Hamiltonians we know certain properties of  generic eigenstates.
For example their bi-partitions are expected to have extensive entanglement cost  which suggests that correlations should be allowed to span substantial portions of the system~\cite{HUANG2019594}.
If the GWW flow as it establishes correlations will be subject to a Lieb-Robinson bound~\cite{hastings2022lieb,lieb1972finite,YoungReviewLR} then taking the total number of flow steps in proportionality to the system size $K\sim L$ may be necessary.

When exploring qualitative rather than quantitative questions about quantum matter, taking $K$ independent of the system size $L$ may be possible.
The eigenstate thermalization hypothesis, see Refs.~\cite{kuwahara,QECETH,alhambra} for proven instances, suggests that the reductions of eigenstates agree with a thermal density matrix.
These, in turn, obey a so-called area law for the mutual information~\cite{PhysRevLett.100.070502} which is related to decaying correlations, i.e., in physically relevant cases choosing $N$ independent of the system size $L$ may be considered as long as it is sufficiently large to match the required range of correlations.
Many-body localized systems have eigenstates which are lowly entangled and indeed applications of double-bracket flow methods have been reported~\cite{MBLPRLPhysRevLett.119.075701}.
Regular quantum statistical systems exhibit a correlation length which decays with the temperature.

In summary this discussion is qualitative and for practical implementations one can instead make a case-by-case analysis like in Fig.~\ref{figeigsGWW}.
Having said that qualitatively entanglement properties of eigenstates can provide a guideline to identify systems where DBIs will perform particularly well.
Many-body localized spin~\cite{MBLPRLPhysRevLett.119.075701} or fermionic~\cite{ScipostMBL10.21468/SciPostPhys.14.5.125} systems seem particularly appealing for potential explorations.

\section{Głazek-Wilson-Wegner flow}
\label{GWWreview}
The continuous Głazek-Wilson-Wegner (GWW) flow is a smooth family of unitary operators $\U_\ell$ parametrized by the total flow duration  $\ell\in[0,\infty)$.
Given $\U_\ell$, we define the GWW flowed Hamiltonian 
\begin{align}
  \H_\ell = \U_\ell^\dagger \h_0 \U_\ell\ ,
  \label{eq:hell}
\end{align}
where $\h_0$ is the input Hamiltonian.
As we will discuss below, generically the flowed Hamiltonian $\H_\ell$ is increasingly more diagonal as $\ell$ increases and so the GWW flow can be used to obtain approximations to a unitary diagonalization transformation.

The GWW flow unitaries $\Uell$ are defined using the requirement that the flowed Hamiltonian $\Hell$ solves the GWW flow equation~\cite{PhysRevD.48.5863,PhysRevD.49.4214,wegner1994flow,wegner2006flow,kehrein_flow}
\begin{align}
  \partial_\ell \Hell = [ \Well, \Hell]\ ,
  \label{eq:derivativehell}
\end{align}
where
\begin{align}
  \Well=[\dHell,\oHell]\ 
  \label{eq:well}
\end{align}
is the commutator of the diagonal $\dHell$ and off-diagonal $\oHell$ restrictions of the flowed Hamiltonian.
The choice \eqref{eq:well}  will be referred to again as the canonical bracket.

The GWW flow allows to sequentially diagonalize Hamiltonians which follows from the fact that the squared Hilbert-Schmidt norm of the off-diagonal restriction is non-increasing~\cite{kehrein_flow,wegner2006flow}
\begin{align}
  \partial_\ell \|\o{\Hell}\|_\text{HS}^2 = -2\,\|\Well\|_\text{HS}^2 \le 0\  .
  \label{eq:wegnerrate}
\end{align}
This equation states that the magnitude of the off-diagonal restriction of flowed Hamiltonians generically decreases and the negative rate is equal to twice the Hilbert-Schmidt norm of the canonical bracket.
Appendix~\ref{app:existence} provides a detailed discussion and self-contained derivation of this key monotonicity feature. 
Eq.~\eqref{eq:wegnerrate} can be also derived similarly to~Eq.~\eqref{GWWDBI}  because we may consider $\J = \Hell$ and take the limit of infinitesimally short double-bracket rotation $s\rightarrow 0$ and use that $\wcan(\J) = \Well$.
The only difference is that in this case we would be performing the Taylor expansion for the GWW flow.

The monotonicity relation \eqref{eq:wegnerrate} is most crucial because it shows why the flow is diagonalizing: The magnitude of off-diagonal terms in the flowed Hamiltonian can stop changing only if the Hilbert-Schmidt norm of $\Well$ vanishes.
In turn this implies $\Well=0$ which by definition~\eqref{eq:well} happens only if the diagonal restriction commutes with the off-diagonal restriction.
This is true if one of these restrictions vanishes and by~\eqref{eq:wegnerrate} it is the off-diagonal restriction that will be flowed away.
Alternatively, on some occasions the two restrictions can both be non-trivial but then they must be commuting so simultaneously diagonalizable.
Thus, in generic situations the fixed point of the GWW flow will be a diagonal form of $\h$ but it is also possible that the flow will reach only a block-diagonal form due to, e.g., spectral degeneracies~\cite{deift1983ordinary,helmke_moore_optimization}.

\subsection{GWW flow limit of the GWW DBI}
\label{GWWQA}

In this subsection we present claims proved in the appendix, namely with only small overhead the GWW GCI can be made to converge onto the GWW flow too.
We next describe the overhead needed to setup a converging GWW GCI.

We begin the formulation of the modified GWW GCI  by fixing the desired flow duration~$\ell$ and the total number of recursive steps $K$.
Based on these we set the durations of the recursive steps to be equidistant and take $s= \ell/K$.

As above, we will use the approximation to the evolution under the pinched generator denoted by
\begin{align}
  \vd_{\dell}(\J) = \prod_{\mu\in\{0,1\}^{\times L}} \p_\mu e^{-i\tfrac\dell D\J}\p_\mu\ 
  \label{pinching_component}
\end{align}
and define the modified group commutator step by
\def\vgww{\hat U^{\mtiny{\mathrm{(GWW)}}}}
\begin{align}
   \vgww_s(\J) = \left( \vd_{\mtiny{r}}(\J )^\dagger\right)^R e^{i\sqrt {s} \J}\left(\vd_{\mtiny{r}}(\J )\right)^R e^{-i\sqrt{s} \J}\ ,
  \label{eq:gcmain}
\end{align}
where $r= \sqrt{s}/R$ and $R$ is an integer which sets the refinement scale needed to converge towards the GWW flow.

We next define the  modified GWW GCI unitaries $\v_k$ after $k$ steps by 
\begin{align}
  \v_k = \v_{k-1} \vgww_s(\J_k) \ ,
  \label{uk22main}
\end{align}
and we set
  $\v_{0} = \id$ and $\J_0=\h_0$.
The iterated Hamiltonian is given by
\begin{align}
  \J_{k}=  \v_{{k}}^\dagger\J_0 \u_{{k}}\ 
  \label{hk22main}\ .
\end{align}

This modified GWW GCI  converges to the GWW flow $\lim_{K\rightarrow \infty} \v_K=\Uell$ which is captured by the following proposition.
\begin{proposition}[Quantum algorithm for GWW flow]
\label{proplimitGWW}
For fixed total flow duration $\ell\in \mathbb [0,\infty)$ and $K$ equidistant flow steps, the refined GWW GCI quantum algorithm converges to the continuous GWW flow according to
\begin{align}
 \| \U_{\ell} - \v_K\|  
  \le& O\left( \sqrt{\ell} e ^{16 \|\h_0\|^2 \ell} K^{-1/2}\right)\ .
 \end{align}
\end{proposition} 
The proof is in appendix~\ref{app:quantumalgorithm}.
The proposition shows conceptually that despite the GWW flow being non-linear, it can be implemented in the quantum computing framework.
Indeed, in this work the GWW flow plays the role of providing the foundation for the GWW DBI by setting $\D_k = \Delta(\h_k)$ which for sufficiently short durations $s_k$ is promised to reduce the norm of the off-diagonal restrictions.

Quantitatively,  the obtained convergence rate precludes meaningful practical applications.
The inverse square root dependence on the number of recursion steps stems from the scaling of the error term of the group commutator approximation~\eqref{eq:groupcommutator}.
This can be improved by considering higher-order product formulas~\cite{HigherOrderGCPhysRevResearch.4.013191} but the exponential dependence on $\ell$ may be a characteristic feature of recursive discretizations because Prop.~\ref{proplimitGWW} in app.~\ref{app:quantumalgorithm} proves the convergence of the GWW GCI to the GWW DBI and a similar scaling appears  there too.
Thus it appears that continuous double-bracket flows are useful conceptually and DBIs are a better suited notion for experimental implementations.

\section{Relation to prior work and open questions}
\label{context}

\subsection{Towards quantum hardware applications}
\paragraph{Regarding variational quantum algorithms.}
Similar to the design principle of regular variational approaches, running double-bracket quantum algorithms relies directly on primitive gates~\cite{KishorRevModPhys.94.015004}.
Being algorithmic, a diagonalization DBI may sometimes appear less efficient than a circuit ansatz optimized by brute-force.
This is expected especially for  miniature systems of a handful qubits because a less restrictive ansatz allows to `overfit' and achieve an input-specific compression which is unlikely to generalize to larger systems.
Similarly, if a device is very noisy and only miniature circuits involving just a handful of layers of entangling gates are available then brute-force optimization may produce better results than the more algorithmic DBI approach~\cite{larose2019variational,Zeng_2021,commeau2020variational,cirstoiu2020variational,gibbs2022long}.

For larger quantum computations (larger number of qubits and larger circuit depths) diagonalization DBIs will perform better than circuits optimized by brute-force.
Finding the right circuit is computationally hard~\cite{PhysRevLett.127.120502,stilck2021limitations} and in practice barren plateaus are prohibitive~\cite{KishorRevModPhys.94.015004}.
In contrast, a DBI can be set up easily and can provide a feasible solution to the quantum compiling task of approximating an eigenstate.
If thousands of layers of entangling gates are available then DBIs will fare better as soon a brute-force optimization will become unpractical for the mere reason of being able to use the quantum hardware meaningfully.

An unstructured circuit ansatz optimized by brute-force can be used as preconditioning for the more sophisticated DBI approach which can take over for a second stage of state preparation.
It is open whether DBIs, together with short-depth variational circuits, will suffice for advancing material science without resorting to the even more advanced quantum algorithms based on phase estimation.
See Ref.~\cite{riemannianflowPhysRevA.107.062421} for a recent work on optimizing variational circuits using Riemannian gradient flow methods related to double-bracket flows.

\paragraph{Regarding quantum phase-estimation algorithms.}
For the central task of obtaining physical properties of quantum many-body systems through quantum computation it will be useful to prepare ground and thermal states~\cite{KnillOptimal,PoulinWocjan,temme2011quantum,ge2019faster,lin2020near,gilyen2019quantum,motta2020determining,tan2020quantum,dong2022ground,PRXQuantum.3.010318,epperly2021theory}.
With few exceptions, see, e.g., Ref.~\cite{motta2020determining}, most explicit quantum algorithms rely crucially on controlled-unitary operations~\cite{kitaev1995quantum,lin2022lecture} and then tackle the resulting probabilistic component by amplifying the success probability in phase estimation~\cite{brassard2002quantum,PoulinWocjan,kothari,parrish2019quantum,stair2020multireference}.

Without preconditioning, quantum algorithms based on phase estimation need an exponential runtime for ground state preparation~\cite{ge2019faster,lin2020near}.
For  double-bracket flow it is difficult to predict when a target accuracy will be reached.
For similar reasons, there is no runtime guarantee for variational  double-bracket iterations and understanding it better is a key theoretical challange to advance this approach further.
For proof-of-principle applications and benchmarking for material science applications, DBIs can be effective just after a few recursion steps.
This may be indicative of an exponential convergence  which appears for the double-bracket flow of the QR algorithm~\cite{deift1983ordinary, helmke_moore_optimization}.

\paragraph{Regarding DBI parametrization.}
The variational brackets studied here have been simple, namely products of Pauli-$Z$ operators.
Is it more advantageous to use linear combinations of Pauli-$Z$ operators?
For both TFIM and TLFIM $\d{\h_0}=\sum_{i=1}^L \Z_i$ is in form of such linear combination and in Fig.~\ref{figL3diag} is outperformed by a single product operator.

The canonical bracket loses efficiency in presence of degeneracies in $\h_0$ and in general lack of degeneracies allows to ensure to eventually obtain diagonalization~\cite{BROCKETT199179,helmke_moore_optimization,deift1983ordinary}.
Are spectral gaps important when aiming for fast heuristic approximations?

Another aspect of interest is choosing the computational basis state which will offer the most insight about the physical properties of the system.
What needs to happen during the iteration for the fully polarized state $\ket{11\ldots1}$ to not flow towards the ground state?

\subsection{Towards new quantum algorithms}

\paragraph{Regarding other double-bracket flows.}
There appear to be further possibilities for creating novel quantum algorithms by building on the various continuous flows developed in the field of dynamical systems.
This includes the early work in Ref.~\cite{deift1983ordinary} which discusses continuous flow equations implementing the QR decomposition and proposes  a continuous version of the Golub-Kahan algorithm for computing the singular value decomposition~\cite{golub1965calculating}.
In Ref.~\cite{BROCKETT199179} double-bracket flows for diagonalization, linear programming and sorting have been formulated, see also Ref.~\cite{BROCKETT1989761} for matching optimization and for simulating any finite automaton~\cite{brockett2005smooth}.
Thus, there appears to be a vast literature on analog classical computing which explores how non-linear differential equations can encode various computational tasks~\cite{Chu_iterations,helmke_moore_optimization,bloch1985completely,BLOCH1985103,bloch1990steepest,bloch1992completely,smith1993geometric}.
In this work the focus has been solely on the task of preparing eigenstates of quantum many-body systems but quantizing these ideas can add to the body of explicit quantum algorithms that we know of.

\paragraph{Regarding analog quantum simulation.}
It is very likely that when  repurposing  other dynamical equations for quantum computing one will face challenges similar to those discussed in Sec.~\ref{GWWreview}.
However, we may ask if a fully analog implementation of a double-bracket flow is viable? 
Here we resorted to a discretization which raises the question if recursive circuits can be avoided in the digital approach?
Positive answers to these questions would likely affect an entire family of, possibly quantum, algorithms.
Ref.~\cite{QDP} takes steps towards using input quantum states as quantum instructions for double-bracket quantum algorithms but the input system size needed to grow exponentially with the duration of evolution.

\paragraph{Regarding dephasing as a quantum algorithmic component.}
The key physical property needed to embed the GWW flow into the quantum computing framework has been irreversible dephasing.
The approach taken here parallels dynamical decoupling~\cite{ezzell2022dynamical}.
The presented implementation proposal uses that pinching is a mixed-unitary channel but the preparation complexity has been exponential in the system size if no structure of the problem is promised.
Arguably this is natural because of the size of the input matrix but it would be useful to avoid this preparation inefficiency.
The appendix~\ref{app:sparsity} explores the case of pinching sparse Hamiltonians, however, if the recent conjecture regarding Lieb-Robinson bounds of double-bracket flows is true~\cite{hastings2022lieb} then the flowed Hamiltonian could quickly become dense.
It should be noted that the Brockett double-bracket flow~\cite{brockett1991dynamical, brockett1991dynamical2} does not feature dephasing and yet is also diagonalizing~\cite{helmke_moore_optimization}.

\paragraph{Regarding quantum algorithms using mixed-unitary channels.}
It occurs worthwhile to consider the question: What use can we make of mixed-unitary maps which involve polynomial rather than exponential number of unitary conjugations?
Here, pinching was a costly circuit component which broadly speaking had the role of implementing memory loss.
Can polynomial mixed-unitary maps give rise to controllable irreversibility and are they useful for creating purposeful quantum algorithms?

Note, that recent experiments have shown that circuits of the type involved in the pinching component \eqref{pinching_component} have an expressive power which suffices to prepare eigenstates experimentally~\cite{kokail2019self}.
One could speculate if the role of the local gates assigned in that study by variational optimization has also been to dephase couplings.

It seems that recursive flows admit a discretization instability.
Proposition~\ref{proplimitGWW} states convergence of the proposed quantum algorithm towards the continuous GWW flow in the limit of vanishing flow step durations.
However, the constants involved in the overall bound depend exponentially rather than polynomially on the total flow duration.
An abstraction of this is lemma~\ref{recursionlemma} in appendix~\ref{app:discretization} which suggests that this exponential dependence stems from the deviations of the generator because it depends on the accuracy of the previous steps.
If this exponential dependence is a true instability then the more adaptive \sgd setting should be considered in practice.

\section{Conclusions and outlook}
\label{conclusions}

Double-bracket iterations (DBIs) reveal in what order to compose evolutions under an input Hamiltonian $\h_0$ with diagonal evolutions in order to obtain  approximations of its eigenstates.
The approach has an analytical character and  offers a way to find a feasible solution to a large scale quantum compiling task which is untractable when relying exclusively on brute-force optimization~\cite{PhysRevLett.127.120502,stilck2021limitations}.

The essential ingredient is to perform a double-bracket rotation, namely if $\D_0$ is a diagonal operator then for short enough $s_0$ either  $\exp(-s_0[\D_0,\h_0])$ or $\exp(s_0[\D_0,\h_0])$ will transform $\h_0$ into $\h_1$ which will have a decreased magnitude of off-diagonal terms as expressed by the Hilbert-Schmidt norm. 
The abstract algorithm~\ref{algorithmDBI} summarizes the recursive iteration of this double-bracket rotation ansatz by feeding in another diagonal operator $\D_1$ to transform $\h_1$ into $\h_2$ which again is more diagonal etc.
Thus we obtain  an ansatz which \emph{i)} is variational ($\hat D_k$ and $s_k$ can be optimized) and \emph{ii)} has implementation requirements similar to unstructured approaches of using readily available gates and composing them such as to prepare eigenstate approximations.

In contrast to other variational approaches, though, for double-bracket iterations obtaining an additional diagonalizing step is as simple as picking a new diagonal operator and learning whether to use the induced bracket for forward or backward evolution.
The knowledge how to perform effective double-bracket rotations can be learned using quantum hardware which then can be used to progressively study further the system of interest with a next recursive step.
This is different from quantum phase estimation which relies on random outcomes of measurements.
In terms of implementation requirements, the presented double-bracket approach can be viewed as being positioned between unstructured variational approaches and those based on phase estimation.

The group commutator iteration (GCI) algorithm~\ref{aagci} formalizes the main proposal of this work, namely to implement experimentally iterations  which approximate double-bracket rotations by the corresponding group commutator formula.
Algorithm~\ref{algorithmtranspiler} transpiles GCIs into a sequence of evolutions governed by the input Hamiltonian and diagonal operators.
The involved evolutions can be implemented by Hamiltonian simulation or Clifford operations and thus double-bracket quantum algorithms consist of operations natural for quantum computing.
The proposed quantum algorithm~\ref{quantumalgorithmGCI} is as simple as applying the evolutions prescribed by the transpiler algorithm ~\ref{algorithmtranspiler}.

The key ingredient of the GCI algorithm~\ref{aagci} are evolutions governed by the input Hamiltonian $\h_0$.
The presented unfolding approach leads to exponentially many queries to the input evolution in the number of recursion steps.
However, Ref.~\cite{QDP} shows that in certain cases the circuit depth can be reduced to polynomial via quantum dynamical programming.
This comes at the cost of additional overheads and thus the number of recursion steps must be kept low.

Fig.~\ref{figeigsGWW} addresses the question what eigenstate approximations can be achieved with few recursive DBI steps.
The fully polarized state $\ket{11\ldots 1}$ stands out as the candidate of choice for experimental studies because it can be expected to tend to a state representative of the low-energy physics of the model even for larger system sizes.
Additionally Fig.~\ref{figeigsGWW} shows that while the overall norm of the off-diagonal restriction of the iterated Hamiltonian may remain substantially non-zero after a dozen of DBI steps, some computational basis states do tend towards eigenstates.
This is evidenced by experimentally measurable energy fluctuations.

This work arose based on existing theory of double-bracket flows and Sec.~\ref{GWWreview} supplements a demonstration that the Głazek-Wilson-Wegner (GWW) flow can be implemented in the quantum computing framework.
For this it is necessary to approximate evolutions governed by the restriction to the diagonal $\d{\h_k}$ of iterated evolution generators $\h_k$.
Generator dephasing is irreversible  and approximating it is costly when implemented by unitary transformations.
However, as has also been demonstrated, interventions to reduce the runtime are possible.

While double-bracket quantum algorithms are oblivious in the sense that classical information about the input is not necessary, they are not blind as the number of queries to $e^{-it\h_0}$ may suffice to learn $\h_0$.
An interesting physical behavior ensues in that the Hamiltonian `guides' its own diagonalization.
Double-bracket quantum algorithms use a `quantum' oracle because instead of querying `classical' data, e.g., a value of a matrix element, we query an application of a quantum evolution operation to the quantum registers.

A quantum device capable of double-bracket flows for systems with sizes larger than are practical on a laptop computer will likely resort to the proposed heuristics {\it i)} extending durations of double-bracket rotations {\it ii)}, efficient pinching or diagonal associations and {\it iii)} precompiling.
The first two will need to be used for large systems but for small numbers of qubits precompiling, proposed in Eq.~\eqref{eq:precompiling}, can include them directly.
Thus for proof-of-principle implementations of the double-bracket flow quantum algorithm precompiling reduces the experimental demands  to merely accurate Hamiltonian simulation.

\paragraph{Acknowledgments.}
The author is grateful to Jeongrak Son for sharing his insights and spotting necessary mathematical corrections.
Valuable conversations with Raul García-Patron, Stefan Kehrein, Nelly Ng and Steven Thomson and helpful comments on the manuscript by  Ashwinie Ghanesh, Kay Giang, Janusz Gluza, Siong Thye Goh, Zo\"e Holmes, Johannes Kn\"orzer, Akash Kundu and Samantha Xiaoyue Li are gratefully acknowledged.
 
This work has been supported by the start-up grant of the Nanyang Assistant Professorship of the Nanyang Technological University in Singapore which was awarded to Nelly Ng and additionally MG has received support by the Presidential Postdoctoral Fellowship of Nanyang Technological University in Singapore.

\bibliographystyle{quantum}


\onecolumn\newpage
\appendix
\section*{Appendix}
This appendix contains five sections and presents details of calculations that extend the discussion from the main text.
The overall aim has been to provide a self-contained exposition which includes technical details.

This is especially true in the first section~\ref{app:discretization}, which focuses on operator relations and inequalities which are basic for working with double-bracket quantum algorithms.
It starts with elementary norm upper bounds, then derives a second-order upper bound on the group commutator and finally presents solution to an abstract recursion in lemma \ref{recursionlemma} which impacts bounds on the convergence of flow discretizations.
The proofs of these results are technical.

Section~\ref{app:quantumalgorithm} discusses details of the elementary circuit components of the proposed quantum algorithm of the double-bracket flow.
Firstly, the subsection~\ref{discretizedGWW} proves that the discretized GWW flow defined in Eqs.~\eqref{eq:hk} and~\eqref{eq:wkp} converges to the continuous GWW flow.
This is a precursor to understanding the quantum algorithm for the GWW flow whose steps have additional deviations besides those stemming from discretization.
We find that the recursive character of the discretized GWW flow results in an exponential sensitivity on the length of the flow which is different from discretizations of evolutions with an independently prescribed generator where the error bound has a polynomial dependence on the total duration.
Unless this characteristic is an artifact of the proof technique, a quantum device cannot approximate the continuous GWW flow in polynomial runtime.
Next, subsection~\ref{Errortermsapp} proves error bounds on the elemental components of the proposed quantum algorithm in relation to the discretization of the GWW flow.
Subsection~\ref{appconvergence} provides a proof of the proposition~\ref{proplimitGWW} stated in the main text regarding the limit of the quantum algorithm to the continuous GWW flow.

Section~\ref{app:sparsity} resorts to large deviation bounds to prove that a random selection of phase flips leads to accurate pinching wherein the success probability depends on the number of phase flips versus the sparsity of the Hamiltonian.

Section~\ref{app:numerics} presents additional results of numerical simulations.

Section~\ref{app:existence} discusses details about the continuous GWW flow.
While this reproduces the known facts also reviewed elsewhere~\cite{wegner2006flow,kehrein_flow}, a mathematically inclined reader may find it easier to understand because certain subtleties are highlighted rather than suppressed.
In particular, a self-contained constructive proof of the existence of the GWW flow will be provided.
The motivation for considering that problem is summarized in subsectin~\ref{app:existencecontext}.
Subsection~\ref{app:monotonicity} derives the key monotonicity property \eqref{eq:wegnerrate} of the GWW flow on which rest statements about its diagonalizing properties.
Finally, subsection~\ref{appunitaryapprox} shows that DBI discretizations converge to the continuous flow.

\section{Helpful technical calculations}
\label{app:helpfultechnical}
\label{app:discretization}
As in the main text, the Hilbert-Schmidt norm of any two linear operators $\A,\B\in \linops$ induced by the Hilbert-Schmidt scalar product $\langle \A, \B \rangle_\text{HS} = \tr[ \A^\dagger \B]$ will be denoted by $\| \A \|_\text{HS}=\sqrt{\langle \A,\A \rangle_\text{HS}} $. 
We will derive certain results for an unitarily invariant norms which will be denoted without a subscript as $\|\A\|$.
The operator norm (maximum singular value) is the most relevant example other than the Hilbert-Schmidt norm which is unitarily invariant too.

\begin{lemma}[Commutator formulas]
For any $\A,\B,\C,\D\in \linops$ we have that the difference of commutators is a sum of commutators of differences
\begin{align}
  [\A,\B]-[\C,\D]=[\A-\C,\B]+[\C,\B-\D]\ .
  \label{commdiff}
\end{align}
Moreover the trace induces cyclicity as follows
\begin{align}
  \tr\left(\A[\B,\C]\right)=\tr\left(\B[\C,\A]\right) 
  \label{eqcyclicity}
\end{align}
which in particular yields
\begin{align}
  \tr\left(\A[\B,\A]\right)=0 \ .
  \label{comspecial}
\end{align}

\end{lemma}

\begin{proof}
A direct use of linearity of commutators gives
\begin{align}
  [\A,\B]-[\C,\B]+[\C,\B]-[\C,\D]=[\A-\C,\B]+[\C,\B-\D]\ .
\end{align}
We obtain the cyclicity property by using the cyclicity of the trace in two directions
\begin{align}
  \tr\left(\A[\B,\C]\right)&=
  \tr\left(\B\C\A\right) - \tr\left(\B\A\C\right)\\&=
  \tr\left(\B[\C,\A]\right) \ .
\end{align}
We obtain \eqref{comspecial} by using \eqref{eqcyclicity} which gives $\tr\left(\B[\A,\A]\right)=0$.

\end{proof}

\begin{lemma}[Bounding canonical brackets]\label{bndcom}

Let $\A,\B\in \linops$. 
We have the upper bounds on the bracket
\begin{align}
\|[\d{\A},\o{\B}]\| \le 4\|\A\|\, \|\B\|
\label{bndcom1}
\end{align}
and on the difference of brackets
\begin{align}
\|[\d{\A},\o{\A}]-[\d{\B},\o{\B}]\| \le 8 \max\{ \|\A\|,\|\B\|\}\, \|\A-\B\| \ .
\label{bndcom2}
\end{align}\end{lemma}
\begin{proof}
We have 
\begin{align}
\|[\d{\A},\o{\B}\| \le 2\|\d{\A}\|\,\|\o{\B}\|\le 4\|\A\|\, \|\B\|\ ,
\end{align}
where we used that via the pinching inequality in general $\|\d\A\|\le \|\A\|$~\cite[Eq.~(IV.51) on p. 97]{bhatia} which together with the triangle inequality implies that $\|\o \B\|\le 2 \|\B\|$.

Notice, which is useful generally, that $\o\cdot$ is linear namely 
\begin{align}
  \o\A-\o\B=\A-\d\A-\B+\d\B = \A-\B -\d{\A-\B} = \o{\A-\B}\ .
\end{align}
Using \eqref{commdiff} and \eqref{bndcom1} we obtain
\begin{align}
\|[\d{\A},\o{\A}]-[\d{\B},\o{\B}]\|  
\le&\|[\d{\A-\B},\o{\A}]\|+\|[\d{\B},\o{\A-\B}]\|\\
\le& 4\|\A-\B\|\, \|\A\|+\|[\d{\B},{\A-\B}]\|+\|[\d{\B},\d{\A-\B}]\|\\
\le&   8 \max\{ \|\A\|,\|\B\|\}\, \|\A-\B\|\, \ .
\end{align}
\end{proof}
\begin{lemma}[Diagonal-off-diagonal orthogonality]\label{orthogonality}
Let $\A,\B\in \linops$. 
We have the Hilbert-Schmidt orthogonality of the diagonal and off-diagonal restrictions
\begin{align}
  \tr\left(\o \A\d\B\right)=0\ .
  \label{eqorthogonality}
\end{align}

\end{lemma}

\begin{proof}
Let us consider a basis $\{\ket k\}_{k=1,\ldots,D}$ and by direct computation we find that an off-diagonal operator $\sigma(\A) = \sum_{k,k'=1}^D (1-\delta_{k,k'}) \sandwich k \A {k'}$ multiplied by a diagonal operator $\d \B=\sum_{k,k'=1}^D  \delta_{k,k'} \sandwich k\B{k'}$ remains off-diagonal
\begin{align}
  \o \A\d\B &= \sum_{k,k'=1}^D\sum_{j=1}^D (1-\delta_{k,k'}) \delta_{k',j} \sandwich k \A {k'}\sandwich j \B{j} \ketbra{k}{j} 
  \\&= \sum_{k,j=1}^D (1-\delta_{k,j})  \sandwich k \A {j}\sandwich j \B{j} \ketbra{k}{j} \ .
\end{align}
Then $\tr{\ketbra{k}{j}} = \delta_{k,j}$ and $(1-\delta_{k,j}) \delta_{k,j}=0$ imply
\begin{align}
  \tr\left(\o \A\d\B\right)= \sum_{k,j=1}^D (1-\delta_{k,j}) \delta_{k,j} \sandwich k \A {j}\sandwich j \B{j}  =0 \ .
\end{align}
\end{proof}

The relations \eqref{eqorthogonality} and  \eqref{eqcyclicity} turn out to be quite crucial.
Let $\A,\B,\C,\D\in \linops$ be hermitian then we have 
\begin{align}
  \tr\left(\o\A \o{[[\d\B,\o\C],\D]}\right)
  &  = \tr\left(\o\A [[\d\B,\o\C],\D]\right)\\
  &  = \tr\left([\d\B,\o\C][\D,\o\A]\right) \ .
  \label{orthoabstract}
\end{align}
Applications for this purely linear-algebraic relation can be found by rephrasing the relation~\eqref{orthoabstract} using the Hilbert-Schmidt scalar product
\begin{align}
  \left\langle \o\A^\dagger, \o{[[\d\B,\o\C],\D]}\right\rangle_\text{HS}
  &  = \left\langle [\d\B,\o\C]^\dagger, [\D,\o\A]\right\rangle _\text{HS}\ .
\end{align}
This means that certain projections of off-diagonal restrictions of brackets can be evaluated as Hilbert-Schmidt overlaps of reshuffled brackets.
Or, put differently, the nested bracket simplifies to the first order bracket under the trace.
This is not directly valid for higher-order nesting.

In particular setting all the involved operators to be the Hamiltonian $\h$ we find
\begin{align}
  \left\langle \o\h, \o{[[\d\h,\o\h],\h]}\right\rangle_\text{HS}
  &  = -\left\langle [\d\h,\o\h], [\d\h,\o\h]\right\rangle _\text{HS}\\
  & = - \|[\d\h,\o\h]\|_\text{HS}^2\ .
\end{align}

This formula comes into play when analyzing short flow steps generated by the canonical bracket.
If we consider a variational bracket with some diagonal operator $\B=\dzz(\h)$ then we would find that the canonical bracket comes into play nonetheless
\begin{align}
  \left\langle \o\h, \o{[[\dzz(\h),\o\h],\h]}\right\rangle_\text{HS}
  &  = -\left\langle [\dzz(\h),\o\h], [\d\h,\o\h]\right\rangle _\text{HS}\ .
\end{align}
This discussion stresses the mathematical structures that lead to the appearance of the canonical bracket when considering such classess of double-bracket flows.
Lemma~\ref{sgdlemma} in the main text summarizes the application to.
 double-bracket quantum algorithms.
 
\begin{lemma}[Group commutator]
Let $\A=-\A^\dagger,\B=-\B^\dagger\in \linops$ be anti-hermitian and let $\|\cdot\|$ be a unitarily invariant norm.
We have
\begin{align}
  \|e^{\A} e^{\B} e^{-\A}e^{-\B} -e^{[\A,\B]}\|\le \| [\A,[\A,\B]]\| + \| [\B,[\B,\A]]\|\ .
\end{align}
\label{groupcom}
\end{lemma}
Let us mention that we will typically use this bound for $\A= i \sqrt s \J$ and $\B=i \sqrt s \J'$ for some real evolution duration $s\in \mathbb R$ and hermitian $\J,\J'$.
We then get that the scaling of the deviation of the exact bracket-evolution is given by
\begin{align}
  \|e^{i\sqrt s\J} e^{i\sqrt s\J'} e^{-\sqrt s\J}e^{-i\sqrt s\J'} -e^{-s[\J,\J']}\|\le 4 s^3\left(\max(\|\J\|,\|\J'\|)\right)^3  = O(s^{3/2})\ .
\end{align}
Due to a cancellation, this is a reduced power (and thus higher error for $s\ll 1$) compared to, for example, the first order composition
\begin{align}
  \|e^{i s\J} e^{i s\J'} -e^{is(\J+\J')}\|\le s^2\|[\J,\J']\| \le 2s^2\left(\max(\|\J\|,\|\J'\|)\right)^2   = O(s^{2})\ ,
\end{align}
which thus converges at a faster rate.
\begin{proof}
Define
\begin{align}
  Q_x = e^{x\B_{\A}}e^{-x\B}e^{(1-x)[\A,\B]}\ ,
\end{align}
where
\begin{align}
 \B_{\A}= e^{\A}\B e^{-\A} \ .
\end{align}
It will be useful to consider smooth rescaling $\A\mapsto{x\A}$ by some $x\in \mathbb R$ to give rise to a smooth matrix-valued function $\B_{x\A}$  for which the integral form of the remainder in the Taylor expansion theorem gives
\begin{align}
 \B_{\A}= \B + [\A,\B] + \int_0^1 \dd x (1-x)e^{x\A}[\A,[\A, \B]] e^{-x\A} \ .
 \label{taylorBA}
\end{align}
In the following we will use that  by unitary invariance of the norm we get the bound for the remainder
\begin{align}
 \| \B_{\A}- \B + [\A,\B]\| \le \|[\A,[\A, \B]] \| \ .
\end{align}
We first check that
\begin{align}
  \Q_1 = e^{\B_{\A}}e^{-\B}=e^{\A} e^{\B} e^{-\A}e^{-\B}
\end{align}
and
\begin{align}
  \Q_0 = e^{[\A,\B]}\ .
\end{align}
By an explicit calculation we may check that the derivative of $\Q_x$ reads
\begin{align}
  \partial_x\Q_x = e^{x\B_{\A}}(\B_{\A}-\B-e^{-x\B}[ \A,\B]e^{x\B})e^{-x\B}e^{(1-x)[\A,\B]} \ .
\end{align}
By the fundamental theorem of calculus, triangle inequality and unitary invariance of the operator norm we obtain
\begin{align}
  \|\Q_1-\Q_0\| 
  &\le \int_0^1\dd x \|\B_{\A}-\B-[ e^{-x\B}\A e^{x\B},\B]\|\ .
\end{align}
Next, we perform a telescoping step pivoting around $[\A,\B]$ obtaining
\begin{align}
  \|\Q_1-\Q_0\| 
  &\le \int_0^1\dd x \|\B_{\A}-\B-[ \A,\B]\|+\int_0^1\dd x \|[\A,\B]-[e^{-x\B}\A e^{x\B},\B]\| \\
  &\le  \|[\A,[ \A,\B]]\|+\int_0^1\dd x \|[\A-e^{-x\B}\A e^{x\B},\B]\| 
\end{align}
We next consider the linear $n=1$ Taylor expansion  for  $ \A_{x\B} = e^{-x\B}\A e^{x\B} = \A + \int_0^x \dd y [\B,e^{-y\B}\A e^{y\B}]$ again with the remainder in integral form and 
arrive at
\begin{align}
  \|\Q_1-\Q_0\| 
    &\le  \|[\A,[ \A,\B]]\|+\int_0^1\dd x \int_0^x \dd y \|e^{-y\B}[[\B,\A] ,\B]e^{y\B}\| \\
    &\le  \|[\A,[ \A,\B]\|+ \|[\B,[\B,\A]]\| \ .
  \end{align}

\end{proof}

As an alternative, we may consider using more standard bounds together with the Taylor expanasion formula~\eqref{taylorBA}.
Our strategy will be that
\begin{align}
    e^{\B_{\A}}e^{-\B}\approx e^{\B_{\A}-\B}= e^{{[\A,\B] +\hat R}}\approx e^{[\A,\B]}e^{\hat R}\approx e^{[\A,\B]}\ .
\end{align}
We will use that 
\begin{align}
    \|e^{\hat R}  - \id\| \le \|\hat R\|\ ,
\end{align}
which, we may mention, can be proven for any two anti-hermitian $\A,\B$ by considering $\Q(x) = e^{x\A}e^{(1-x)\B}$ with $\partial_x\Q(x) = e^{x\A}(\A-\B)e^{(1-x)\B}$ such that by using unitary invariance
\begin{align}
\|e^{\A}  - e^{\B}\| \le \int_{0}^1 \dd x \| \partial_x \Q(x) \| \le \|\A-\B\|\ .
\end{align}

Similarly, for any anti-hermitian $\A,\B$ we get the bound
\begin{align}
    \|e^{\A}e^{\B} - e^{\A+\B}\| \le \|[\A,\B]\|\ ,\label{boundeasyGC}
\end{align}
which we prove analogously by defining
$\Q(x) = e^{(1-x)\A}e^{(1-x)\B}e^{x(\A+\B)}$ so that 
\begin{align}
\partial_x\Q(x) = -e^{(1-x)\A}e^{(1-x)\B}\left(e^{-(1-x)\B}\A e^{(1-x)\B} +\B -(\A+\B)\right)e^{x(\A+\B)}
\end{align}
such that using 
\begin{align}
e^{-(1-x)\B}\A e^{(1-x)\B}= \A +  \int_0^x \dd y e^{-(1-y)\B}[\B, \A] e^{(1-y)\B}
\end{align}
we get a simplification of the bound~\eqref{boundeasyGC} obtaining
\begin{align}
\|e^{\A}e^{\B}-e^{\A+\B}\| \le \int_{0}^1 \dd x \| \partial_x \Q(x) \| \le \|[\A,\B]\|\ .
\end{align}

To proceed with the bound on the group commutator approximation, for notational purposes we denote
\begin{align}
 \hat R =  \int_0^1 \dd x (1-x) e^{x\A}[\A,[\A, \B]] e^{-x\A} \ .
\end{align}
By using telescoping and triangle inequality we find
    \begin{align}
  \|e^{\A} e^{\B} e^{-\A}e^{-\B} -e^{[\A,\B]}\|
  &\le \|e^{\B_{\A}}e^{-\B} -e^{{[\A,\B] +\hat R}}\|+\|e^{[\A,\B] +\hat R} -e^{[\A,\B]}\|\\
  &\le \|[\B_{\A},\B]\|+\|e^{[\A,\B] +\hat R} -e^{[\A,\B]}e^{\hat R}\|+\|e^{[\A,\B]}\left( e^{\hat R} -\id\right)\|\\
  &\le \|[\B_{\A},\B]\|+\|[[\A,\B] ,\hat R]\|+\| \hat R\| \\
  &\le \|[[\B,\A],\B]\|+\|[\hat R,\B]\|+\|[[\A,\B] ,\hat R]\|+\| \hat R\| \\
  &\le \|[[\B,\A],\B]\|+(1+2\|\A\|\,\|\B\|+2\|\B\|)\|[\A,[\A,\B]]\| \ .
\end{align}
Thus, if both $\A$ and $\B$ contribute $O(s)$ in the bounds then this derivation yields also a third-order bound $O(s^3)$.

\def\a{a}
\def\b{b}

\begin{lemma}[Recursion for recursive discretization]
\label{recursionlemma}
Let $\a,\b\in \R$. 
The sequence $\{\alpha_k\}_{k\in\{0,1\ldots\}}$
with entries
\begin{align}
  \alpha_k = \frac{(1+\a)^k-1}{\a}\b\ .
  \label{eqrecursionresult}
\end{align}
solves the recursion
\begin{align}
\alpha_{k+1} =&  (1+\a)\alpha_k+ \b\ 
\label{recursionsequencebound}
\end{align}
for $k\ge1$ and $\alpha_0 = 0$.
\end{lemma}
\begin{proof}
The ansatz is based on the observation that
\begin{align}
   \sum_{j=0}^{k-1} (1+\a)^{j}=\frac{(1+\a)^k-1}{\a}\ 
\end{align}
and let us show by induction that it is the solution to the recursion.
For $k=0$ we have $\alpha_0=0$ as should be.
Let us assume the relation \eqref{recursionsequencebound} holds for some $k\ge 0$ and show that hence it must hold for $k+1$.
The right-hand side of the recurrence for $\alpha_{k+1}$ reads
\begin{align}
 (1+\a)\alpha_k+ \b\ &= (1+\a)\frac{(1+\a)^k-1}{\a}\b +\b \\
  &= \frac{(1+\a)^{k+1}-1-\a}{\a}\b +\b\\
  &= \frac{(1+\a)^{k+1}-1}{\a}\b = \alpha_{k+1} \ . 
\end{align}
\end{proof}

\section{GWW flow as a quantum algorithm}
\label{app:quantumalgorithm}

\subsection{GWW flow as a GWW DBI limit}
\label{discretizedGWW}

We will consider the GWW flow unitary $\U_\ell$ for some $\ell\in \R$ and compare it to the unitary obtained after  $\n$ steps of the GWW DBI.
More specifically, we will set the recursion step duration to be $\dell = \ell /\n$ and it will be useful to compare to the GWW flow at flow parameters $\{\ell_k = k \dell\}_{k\in\{0, 1,\ldots \n\}}$.

As in the main text, we start with $k=0$ and define
\begin{align}
  \u_{0} = \id\ 
  \label{u0}
\end{align}
and also to generate the first recursion step approximation we define the bracket
\begin{align}
  \w_{0} = [\d{\h_0}, \o{\h_0}]\ .
\end{align}
For any $k\in\{1,\ldots \n\}$ we define
\begin{align}
  \u_k = \u_{k-1}e^{-\dell \w_{k-1}} \ .
  \label{uk}
\end{align}
This allows us to define the iterated Hamiltonian approximation
\begin{align}
  \h_{k}=  \u_{{k}}^\dagger\h_0 \u_{{k}}\ 
  \label{hk}
\end{align}
and with this the approximation to the canonical bracket to generate the next step becomes
\begin{align}
  \w_{{k+1}} = [\d{\h_{k}},\o{\h_{k}}]\ .
  \label{wkp}
\end{align}

\begin{proposition}
\label{propdisc}
  We have
  \begin{align}
 \| \U_{\ell} - \u_K\|  \le&  16   \|\h_0\|^2  \ell e^{20 \|\h_0\|^2\ell}  K^{-1}\ .
\end{align}
\end{proposition}

\begin{proof}
It will be helpful to freely use certain properties of $\Uell$ or more specifically of $\Well$.

By explicit computation we have
\begin{align}
  \U_{\elk} 
  &= \U_{\elkm} +s\,(\partial_\ell\U_{\ell})_{\vert \ell = \elkm}+\int_0^s\dr (s-r)\,(\partial_\ell^2\U_{\ell})_{\vert \ell = \elkm+r}\\ 
  &= \U_{\elkm} -s\,\U_{\elkm}\W_{\elkm}+\int_0^s\dr (s-r)\,\left(\U_{\elkm+r}\W_{\elkm+r}^2+\U_{\elkm+r}\partial_{\elkm+r}\W_{\elkm+r}\right)\ .
\end{align}
Using the GWW equations we have
\begin{align}
    \partial_{\ell}\Well 
    &=  \d{\partial_\ell \Hell} \Hell +\d \Hell \partial_\ell\Hell - \Hell\d{\partial_\ell \Hell} - \partial_\ell\Hell\d \Hell \\
    &=  [\d{ [ \Well, \Hell]} ,\Hell ] + [\d{\Hell} , [ \Well, \Hell] ]\ .
\end{align}
For all $\ell$ we have 
  \begin{align}
   \| \Hell\| = \|\h_0\|
  \end{align}
  and 
  \begin{align}
    \|\d{\Hell}\| \le \|\h_0\|\ 
  \end{align}
and additionally using \eqref{bndcom1} in lemma \ref{bndcom} for $\A = \B = \Hell$
  \begin{align}
   \| \Well\| = 4\|\h_0\|^2\ .
  \end{align}
Using these relations,  we obtain the bound
\begin{align}
    \|\partial_{\ell}\Well \|
        &\le 16 \|\h_0\|^4\ .
\end{align}
This means that
\begin{align}
  \left \| \int_0^s\dr (s-r)\,\left(\U_{\elkm+r}\W_{\elkm+r}^2+\U_{\elkm+r}\partial_{\elkm+r}\W_{\elkm+r}\right)\right\| \le 32 \|\h_0\|^4 s^2
\end{align}
  thanks to $\int_0^s\dr |s-r| = \int_0^s\dr (s-r)= s^2/2$.
  
Thus, we have
\begin{align}
  \| \U_{\elk} - \u_k\| 
  \le&\| \U_{\elkm} - \u_{k-1}\| + s  \|\,\U_{\elkm}\W_{\elkm}-\u_{k-1}\w_{k-1}\| +64 \|\h_0\|^4 s^2\\
  \le&(1+s\|\w_k\|)\| \U_{\elkm} - \u_{k-1}\| +  s\|\W_{\elkm}\,-\w_{k-1}\| +64 \|\h_0\|^4 s^2
  \ .
\end{align}
Using lemma~\ref{bndcom} we have 
  \begin{align}
\|[\d{\H_{\elkm}},\o{\H_{\elkm}}]-[\d{\h_{k-1}},\o{\h_{k-1}}]\| \le 8 \|\h_0\|\, \|\H_{\elkm}-\h_{k-1}\|
\end{align} 
which can be bounded by the fidelity of the GWW flow approximation in the previous step
\begin{align}
  \| \H_{\elkm}-\h_{k-1}\| \le 2 \|\h_0\|\,\| \U_{\elkm} - \u_{k-1}\|\ .
\end{align}
This means that
\begin{align}
\|\W_{\elkm}\,-\w_{k-1}\| \le  16 \|\h_0\|^2\,\| \U_{\elkm} - \u_{k-1}\|\
\end{align}
and so altogether we arrive at
\begin{align}
 \| \U_{\elk} - \u_k\| \le& (1+ 20 \|\h_0\|^2\dell) \| \U_{\elkm} - \u_{k-1}\|  +   64\|\h_0\|^4 s^2\ .
\end{align}
From lemma~\ref{recursionlemma} we obtain
\begin{align}
\| \U_{\ell} - \u_K\| \le& 16\|\h_0\|^2\left(1+ \frac{20 \|\h_0\|^2\ell}{K}\right)^K   \frac\ell K
\end{align}
and we arrive at the proposition statement by using the inequality $(1+x)^r\le e^{xr}$ with $x,r>0$ so that
\begin{align}
\| \U_{\ell} - \u_K\| 
 \le& 16   \|\h_0\|^2  \ell e^{20 \|\h_0\|^2\ell}  K^{-1} \\
 \le& O\left(\ell e^{20 \|\h_0\|^2\ell}  K^{-1}\right)\ .
\end{align}

\end{proof}

\subsection{Error terms of elemental components}
\label{Errortermsapp}
\begin{proposition}[Pinching elemental component]
For $s\in \R$ set $r = \frac s{D}$. Then
\begin{align}
\vd_s(\J) =  \prod_{k=1}^D \pch_k e^{-ir\J} \pch_k^\dagger
\label{aj}
\end{align}
satisfies
\begin{align}
 \|\vd_s(\J)-e^{-is\d\J}  \| \le 4 s^2 \max_{k,k'}\|[\pch_k \J \pch_k^\dagger,\pch_{k'} \J \pch_{k'}^\dagger]\|\le 8s^2 \|\J\|^2\ .
  \label{v1}
\end{align}
\label{pec}
\end{proposition}
As evident from \eqref{aj} this can be implemented using $D$ times the evolution $e^{-ir\J}$ and $2D$ flips $Q_k$ or $Q_k^\dagger$.
\begin{proof}
We can bound \eqref{v1} by fixing an ordering of the terms in \eqref{eq:pinch} and sequentially merging the terms together.
In the first step we obtain for $k\in \{1,\ldots, \lfloor D/2\rfloor\}$~\cite{Su}
  \begin{align}
  \| e^{ir\pch_{2k-1} \J \pch_{2k-1}^\dagger}e^{ir\pch_{2k} \J \pch_{2k}^\dagger}-e^{ir\pch_{2k-1} \J \pch_{2k-1}^\dagger+ir\pch_{2k} \J \pch_{2k}^\dagger}| \le r^2 \|[\pch_{2k-1} \J \pch_{2k-1}^\dagger,\pch_{2k} \J \pch_{2k}^\dagger]\|\ .
\end{align}
In the next step when merging $e^{ir\pch_{2k-1} \J \pch_{2k-1}^\dagger+ir\pch_{2k} \J \pch_{2k}^\dagger}$ with $e^{ir\pch_{2k+1} \J \pch_{2k+2}^\dagger+ir\pch_{2k+3} \J \pch_{2k+4}^\dagger}$ the norm used in the upper bound will increase by 4 after using the triangle inequality.
In general we have for all $\I,\I'\subset\{1,\ldots, D\}$
\begin{align}
 \| [\sum_{k\in\I} \pch_k \J \pch_k^\dagger,\sum_{k\in\I'} \pch_k \J \pch_k^\dagger]\| \le|\I|\,|\I'|\max_{k\in\I,k'\in\I'}\| [ \pch_k \J \pch_k^\dagger, \pch_{k'} \J \pch_{k'}^\dagger]\|\ .
 \label{II}
\end{align}
The sequence of mergings of this type can be kept track of using a tree with $\{1,\ldots, D\}$ as its leaves.
When iterating over the tree we have to use the bound \eqref{II} at most $\lceil \log_2(D)\rceil$ times and get
\begin{align}
  \|\vd_s(\J)-e^{-is\d\J}  \| &\le r^2 \left(\sum_{k=0}^{\lceil \log_2(D)\rceil}2^{2k}\right) \max_{k,k'}\|[\pch_k \J \pch_k^\dagger,\pch_{k'} \J \pch_{k'}^\dagger]\|\\
  &\le r^2 4D^2 \max_{k,k'}\|[\pch_k \J \pch_k^\dagger,\pch_{k'} \J \pch_{k'}^\dagger]\|\ .
\end{align}

For the special case $D=2^L$ this can be tightened to
\begin{align}
  \|\vd_s(\J)-e^{-is\d\J}  \|   &\le s^2 \max_{k,k'}\|[\pch_k \J \pch_k^\dagger,\pch_{k'} \J \pch_{k'}^\dagger]\|\ .
\end{align}
\end{proof}

\begin{lemma}[Refined pinching step]
For any $s\in \mathbb R$ if $R = \Omega( s^{-1/2})$ then
\begin{align}
\|\hat E^{(\Delta)}\|= \|\vd_{\sqrt{s}/R}(\J)^R-e^{-is\d\J}  \| \le 8s^{3/2} \|\J\|^2\ .
  \label{v221}
\end{align}
\label{pinchingstepmain}
\end{lemma}

\begin{proof}
Let $R\in \mathbb N>0$. 
Using lemma~\ref{pec} we have
\begin{align}
 \|\vd_{\sqrt{s}/R}-e^{-i\sqrt{s}/R\d\J}  \| \le 8s R^{-2} \|\J\|^2\ .
  \label{v1}
\end{align}
We next use inductively telescoping
\begin{align}
\|\vd_{\sqrt{s}/R}(\J)^R-e^{-is\d\J}  \| 
&\le \|\vd_{\sqrt{s}/R}(\J)^{R-1}e^{-i(s-\sqrt{s}/R)\d\J}-e^{-is\d\J}  \| + 8s R^{-2} \|\J\|^2\\
&\le \|\vd_{\sqrt{s}/R}(\J)^{R-2}e^{-i(s-2\sqrt{s}/R)\d\J}-e^{-is\d\J}  \| + 2\times 8s  R^{-2} \|\J\|^2\\
&\le \ldots\\
&\le    8s R^{-1} \|\J\|^2\ .
\end{align}
Thus if we set $R = \lceil s^{-1/2}\rceil = \Omega(s^{-1/2})$
then we obtaine the desired $O(s^{3/2})$ bound
\begin{align}
\|\vd_{\sqrt{s}/R}(\J)^R-e^{-is\d\J}  \| 
&\le    9  s^{3/2} \|\J\|^2\ .
\end{align}
\end{proof}

\begin{proposition}[Group commutator elemental component]
Set $r=\sqrt{s}/R$. Then
\begin{align}
   \vgww_s(\J) = \left( \vd_{\mtiny{r}}(\J )^\dagger\right)^R e^{i\sqrt {s} \J}\left(\vd_{\mtiny{r}}(\J )\right)^R e^{-i\sqrt{s} \J}\ .
  \label{eq:gcmain}
\end{align}
satisfies
\begin{align}
\|\hat E^{(\mathrm{GC})}\| =  \|\vgww_s(\J)-e^{s[\d\J,\o\J]}  \| &\le274 s^{3/2}   \|\J\|^2 \times \max\{1, \|\J\|\} \ .
\end{align} 
\label{gcerror}
\end{proposition}
\begin{proof}
By telescoping, we have
\begin{align}
\|\hat E^{(\mathrm{GC})}\| \le&
\|\left({\vd_r(\J)}\right)^R - e^{-ir\d\J}\|+\|\left({\vd_r(\J)}^\dagger\right)^R  - e^{ir\d\J}\| \nonumber\\
&+\|e^{i\sqrt{s}\d\J}e^{i\sqrt{s}\J} e^{-i\sqrt{s}\d\J} e^{-i\sqrt {s}\J}-e^{s[\d\J,\o\J]}  \| \ .
\end{align} 
Lemma~\ref{groupcom} and proposition~\ref{pec} yield
\begin{align}
\|\hat E^{(\mathrm{GC})}\| 
&\le 64 s^{3/2}  \|\J\|  \,\|[\d\J,\o\J]\|+2\|\hat E^{(\Delta)}\|\\
&\le 256 s^{3/2}   \|\J\|^3 + 18 s^{3/2} \|\J\|^2 \\
&\le 274 s^{3/2}   \|\J\|^2 \times \max\{1, \|\J\|\} \ .
\end{align} 
\end{proof}

\subsection{Convergence of the quantum algorithm to the GWW flow}
\label{appconvergence}

When we considered the GWW DBI in subsec.~\ref{discretizedGWW} above, we disregarded the fact that evolution according to the double-bracket must be approximated on a quantum computer.
Paralleling the notation of the GWW DBI unitaries $\u_k$, we define the GWW GCI unitaries $\v_k$ as follows 
\begin{align}
  \v_{0} = \id\ \text{ and } \J_0
  \label{u0222}
\end{align}
and For any $k\in\{1,\ldots \n\}$ we define
\begin{align}
  \v_k = \v_{k-1} \vgww_s(\J_k) \ .
  \label{uk22}
\end{align}
This allows us to define the iterated Hamiltonian rotation as
\begin{align}
  \J_{k}=  \v_{{k}}^\dagger\J_0 \u_{{k}}\ 
  \label{hk22}\ .
\end{align}
Note that we are using the canonical bracket $\wcan$ exactly for the GWW DBI and approximately (via the group commutator) for the GWW GCI.

\begin{proposition}[Continuous flow limit]
\label{propdiscapp}
For fixed flow duration $\ell\in [0,\infty)$, the proposed quantum algorithm and the discretization of the GWW flow converge according to
\begin{align}
 \| \u_N - \v_N\|     \le&27 e ^{16 \|\h_0\|^2 \ell}\max\{1, \|\h_0\|\}s^{1/2}\\
  \le& O\left( \sqrt{\ell} e ^{16 \|\h_0\|^2 \ell} K^{-1/2}\right)\ .
\end{align}
Moreover, in relation to the continuous flow we have
\begin{align}
 \| \U_{\ell} - \v_K\|  
 \le& e^{28 \|\h\|^2\ell} (1+8\|\h\|^2)\dell
 +e ^{16 \|\h\|^2 \ell}\left( 16 \|\h\| + \sqrt{s}\right) \sqrt{s}\\
  \le& O\left( e ^{16 \|\h\|^2 \ell} \sqrt{\frac{\ell}N}\right)\ .
 \end{align}
\end{proposition}

\begin{proof}
We begin with
\begin{align}
  \| \u_k  - \v_k\|
  \le&\| \u_{k-1} - \v_{k-1}\| + \|e^{-s\wcan(\h_{k-1})}-\vgww_s(\J_{k-1})\| \\
  \le&\| \u_{k-1} - \v_{k-1}\| + \|e^{-s\wcan(\h_{k-1})}-e^{-s\wcan(\J_{k-1})}\|+\|e^{-s\wcan(\J_{k-1})}-\vgww_s(\J_{k-1})\| \ .
\end{align}
Using lemma~\ref{bndcom} we have for the first term
\begin{align}
   \|e^{s\wcan(\h_{k-1})}-e^{s\wcan(\J_{k-1})}\|&\le s\| \wcan(\h_{k-1}) - \wcan(\J_{k-1})\|\\
   &\le  8 s\|\h_0\|\, \|\h_{k-1}-\J_{k-1}\|\,
   \\&\le  16 s\|\h_0\|^2\, \|\u_{k-1}-\v_{k-1}\|\ .
\end{align}
For the second term we can use proposition~\ref{gcerror} obtaining
\begin{align}
  \| \u_k  - \v_k\|
  \le&(1+16 \|\h_0\|^2s)\| \u_{k-1} - \v_{k-1}\| + 274 s^{3/2}   \|\h_0\|^2 \times \max\{1, \|\h_0\|\}\ .
\end{align}
This recursive bound can be solved using lemma~\ref{recursionlemma} such that
\begin{align}
  \| \u_K  - \v_K\|
  \le&\frac{(1+16 \|\h_0\|^2 \ell / N)^N - 1} {16 \|\h_0\|^2 s} 274 s^{3/2}   \|\h_0\|^2 \times \max\{1, \|\h_0\|\}\\
  \le&27 e ^{16 \|\h_0\|^2 \ell}\max\{1, \|\h_0\|\}s^{1/2}\\
  \le& O\left( \sqrt{\ell} e ^{16 \|\h_0\|^2 \ell} K^{-1/2}\right)\ .
\end{align}

For the second part, we use proposition~\ref{propdisc} obtaining
 \begin{align}
 \| \U_{\ell} - \v_K\|  \le&  \| \U_{\ell} - \u_K\| + \| \u_K  - \v_K\|\\
 \le& 16   \|\h_0\|^2   e^{20 \|\h_0\|^2\ell}  s
 +27 e ^{16 \|\h_0\|^2 \ell}\max\{1, \|\h_0\|\}s^{1/2}\\
  \le& O\left( \sqrt{\ell} e ^{16 \|\h_0\|^2 \ell} K^{-1/2}\right)\ .
 \end{align}
\end{proof}

Let us comment that the convergence to the GWW flow demonstrates the diagonalizing properties of the quantum algorithm but this convergence is a weak notion due to a peculiar subtlety.
In principle, one could attempt to prove that a quantum algorithm implements a  \sgd sequence for $\h$.
The proof for this would be harder because along of the flow one would need to lower bound the monotonicity gap in Eq.~\eqref{eq:monotnicity_app} to ensure that the \sgd property holds.
Instead, the above proposition shows convergence to the continuous flow and, by proxy, diagonalizing properties follow.
The additional demonstration of \sgd properties would hint at convergence rates because the monotonicity gaps yield the block-diagonalization rate.

\section{Approximate pinching in locally interacting qubit systems}
\label{app:sparsity}
For $\mu,\nu \in \bc$ where $L$ is the number of qubits denote
\begin{align}
  \zm = \otimes_{k=1}^L \begin{pmatrix}1&0\\0&-1\end{pmatrix} ^{\mu_k}
\end{align}
\begin{align}
  \xn = \otimes_{k=1}^L \begin{pmatrix}0&1\\1&0\end{pmatrix} ^{\nu_k} \ .
\end{align}
These operators form an orthonormal Hilbert-Schmidt basis of linear operators $\linops$ with $D=2^L$
\begin{align}
  \langle \zm\xn, \zmp\xnp\rangle_\text{HS} =\delta_{\mu,\mu'}\delta_{\nu,\nu'}\ .
\end{align}
This allows us to write any Hamiltonian $\h$ as
\begin{align}
 \h = \sum_\mnbc \hmn \zm\xn\ .
\end{align}
In this notation, for $\etak 1,\ldots,\etak R\in\bc$ let us define approximate pinching as acting on some input operator $\hat A\in\linops$ by
\begin{align}
 \dmu(\hat A)  =\frac 1 R \sum\nolimits_{k=1}^R \zmc k \hat A \zmc k^\dagger\ .
\end{align}
If $R=D$ and all $\etak i$'s are different then this is simply the pinching channel $\dmu = \Delta$.
\begin{proposition}[Typically balanced flips]
Fix $\nu \in \bc$. Let $\etak1,\ldots,\etak R\in\bc$ be $R$ choices of phase flips drawn independently from the uniform distribution on $\bc$.
The probability of the norm of the approximate pinching exceeding $\epsilon$ is bounded as
\begin{align}
  \Pr\left(\left \|  \dmu(\xn) \right\|> \epsilon \right)\le 2 e^{-R\epsilon^2} \ .
\end{align}
\end{proposition}

The proof is based on a basic combinatorial observation that for any fixed $\nu\in\bc$  the number of $\mu \in\bc$ which have an even overlap with it $|\mu\cap\nu|:=\sum_{i=1}^L \mu_i\nu_i$  equals the number of those with an odd overlap.
This implies that when drawing  $\mu$ uniformly at random the probability of an odd overlap $|\mu\cap\nu|$ with $\nu$ is exactly one half $\Pr(|\mu\cap\nu| \text{ is odd}) = 1/2$.
This also comes into play in other settings, e.g., is the reason why for certain states fidelity can be estimated independent of the number of qubits \cite{directfidelity}.
\begin{proof}
By an explicit computation we find that
\begin{align}
  \zm \xn \zm^\dagger&=(-1)^{|\mu\cap\nu|} \xn=
  \begin{cases}
  \xn, & |\mu\cap\nu| \text{ is even}\\
  -\xn, & |\mu\cap\nu| \text{ is odd} 
  \end{cases}
\end{align}
and this implies that
\begin{align}
  \dmu(\xn) =   \frac 1R \left(\sum_{k=1}^R (-1)^{|\etak k\cap\nu|} \right)\xn\ .
\end{align}
The summands are distributed as independent fair coins because $\Pr((-1)^{|\etak k\cap\nu|}=1) = 1/2$.
Thus, by applying Hoeffding's inequality to $R$ tosses of a fair coin, i.e., a Bernoulli trial using the Rademacher distribution, we obtain
\begin{align}
\Pr\left( \left| \frac 1 { R} \sum\nolimits_{k=1}^R (-1)^{|\etak k\cap\nu|} \right|\ge \epsilon \right) \le 2 e^{-R\epsilon^2}\ .
\end{align}
\end{proof}

Typical balance of sequences of phase flips means that off-diagonal interaction terms in $\o\h$ become small very fast individually.
Let us define the off-diagonal support set $\S = \{\mu,\nu \in \bc \text{ s.t. } \hmn\neq 0 \text{ and } \nu \neq 0\}$ and the sparsity  $S = |\S|$ to be its number of entries.
Additionally, let us use the max-norm to denote the largest interaction strength $\|\J\|_\text{max}  = \max_{\mu,\nu \in \bc} \overline |J_{\mu,\nu}|$.
\begin{proposition}[Approximate pinching]
For any $\J$ with sparsity $S$  pinching $\d\J$ can be approximated in operator norm up to error $\epsilon>0$ and with failure probability $\delta > 0$ using $\dmu(\J)$ by independently drawing phase flips  $R = \frac{S^2\|\J\|_\mathrm{max}^2}{\epsilon^2} \log( 2S/\delta)$ times from the uniform distribution on $\bc$
\begin{align}
  \Pr\left(\left \|\d\J-  \dmu(\J) \right\|\le \epsilon \right)\le 1-\delta \ .
\end{align}
\end{proposition}
\begin{proof}
By the triangle inequality we have
\begin{align}
  \|\d{\J} - \dmu(\J)\|
  \le& S \|\J\|_\text{max}\max_{\mu,\nu\in\S}\left\|\dmu(\xn)\right\|
\end{align}
so if  $\etak i$'s are such that for all $\mu,\nu\in \S$ the sequence is sufficiently balanced such that $\left\|\dmu(\xn)\right\|\le \varepsilon:= \frac{\epsilon}{S\|\J\|_\text{max}}$ then $\|\d{\J} - \dmu(\J)\| \le \epsilon$ for some $\epsilon >0$.
The probability of this failing to occur can be bounded as
\begin{align}
\Pr\bigl(\exists \mu,\nu \in \S \text{ s.t. }\bigl \|  \dmu(\xn) \bigr\|> \varepsilon \bigr)\le\sum_{\mu,\nu\in \S}\, \Pr\bigl(\bigl \|  \dmu(\xn) \bigr\|> \varepsilon \bigr)\ . 
\end{align}
This sufficient condition implies that 
\begin{align}
  \Pr\bigl( \|\d{\J} -  \dmu(\J)\|\le  \epsilon \bigr)
  &\ge \Pr\bigl(\forall {\mu,\nu\in \S}\text{ holds }\,\bigl \|  \dmu(\xn) \bigr\|\le \varepsilon \bigr)\\  
 &\ge 1-S\max_{\mu,\nu\in \S}\, \Pr\bigl(\bigl \|  \dmu(\xn) \bigr\|> \varepsilon \bigr)\\
 &\ge1-  2 Se^{- R\varepsilon^2} \ .
\end{align}
We arrive at the statement of the proposition by using the prescribed dependance of the number of trials $R$ on the approximation error $\epsilon$ and failure probability $\delta$.
\end{proof}

\section{Additional plots}
\label{app:numerics}

This appendix provides examples of figures for other parameters than in the main text.
Firstly, Fig.~\ref{appfigL7}a discusses the analog of the Fig.~\ref{figL5} but now for $L=7$ sites. 
It also presents two more plots, one for the TFLIM and one for TFIM with coupling strength $J_X=1$.
Secondly, Fig.~\ref{appfigeigsGWW2} shows eigenstates in analogy to Fig.~\ref{figeigsGWW} but now for TFIM as opposed to TFLIM which showcases the difference in flow when integrability is present.
Finally, Fig.~\ref{appfigL3diag} shows the analogy to Fig.~\ref{figL3diag} but focuses on the coupling $J_X=1$ as opposed to $J_X=2$ in the main text.

The numerical code written in python is freely available~\cite{github_dbf} and includes a class implementing a suite of functionalities for simulating double-bracket flows together with evaluation of observables and visualization.
The entire project is accompanied by jupyter notebooks, one for each figure appearing in this work.
The optimization of the flow step durations is performed by either a grid or binary search and the value yielding the largest $\sigma$-decrease is selected.
Each of the figures can be produced in a matter of a few minutes.
This timing is significantly larger for larger system sizes $L>9$ than considered here and thus such simulations would carry extra cost in terms of greenhouse gases~\cite{climate}.
\begin{figure*}
  \includegraphics[trim = 0cm 0 0cm 0, clip,width=.33\linewidth]{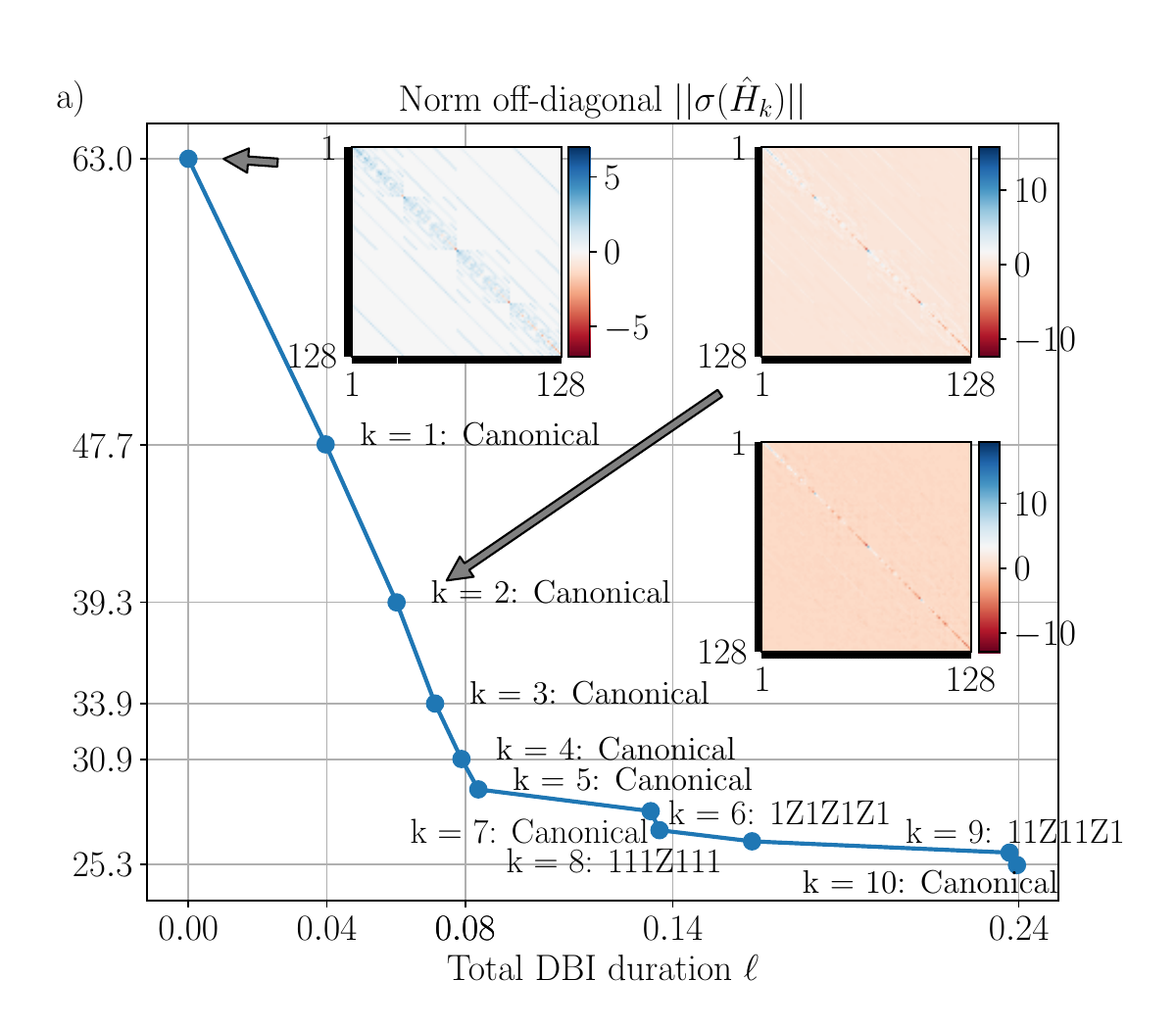}
  \includegraphics[trim = 0cm 0 0cm 0, clip,width=.33\linewidth]{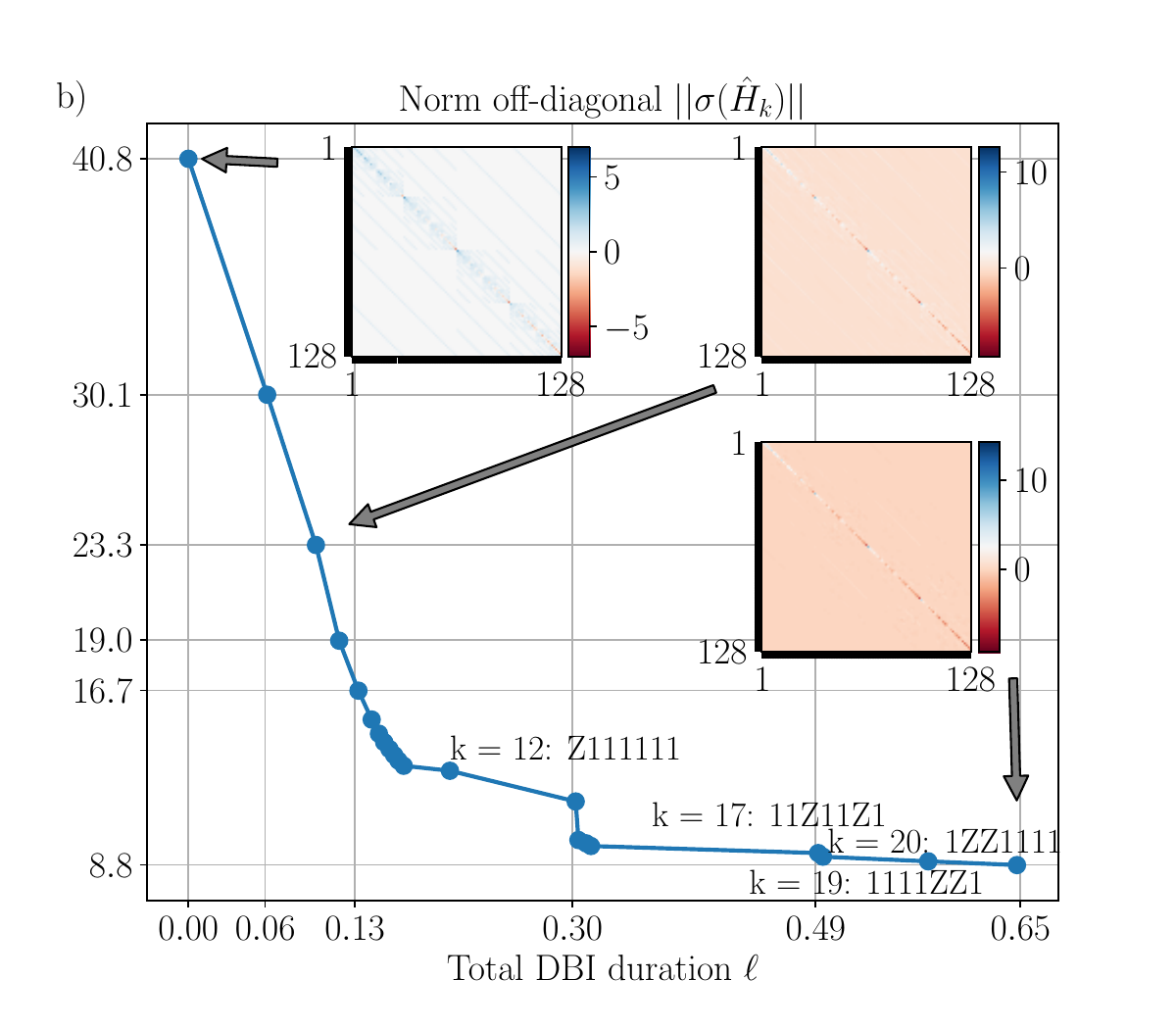}
  \includegraphics[trim = 0cm 0 0cm 0, clip,width=.33\linewidth]{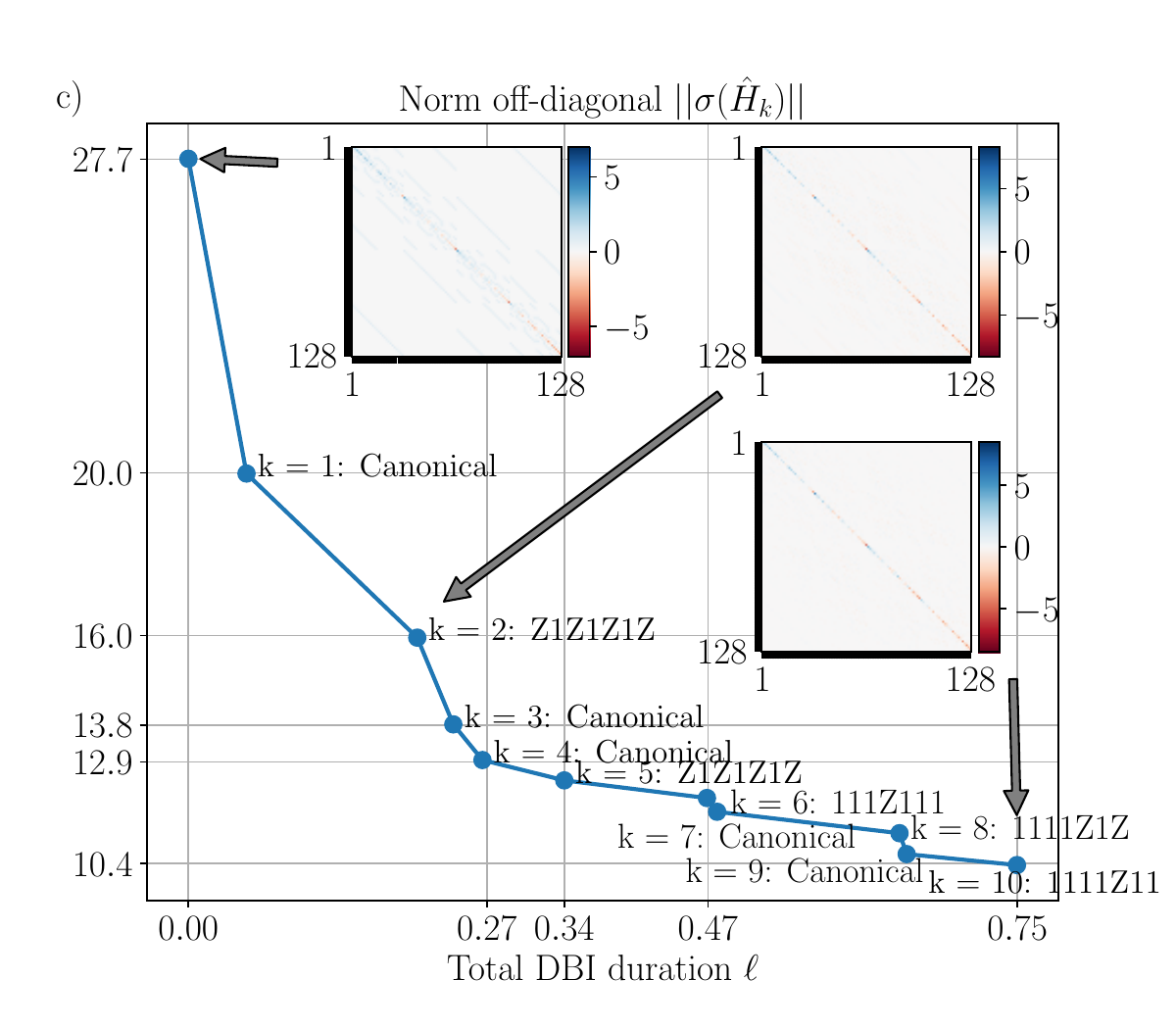}
  \caption{Examples of variational flows for $L=7$ sites. a) TFLIM with $J_X=2$ (Same as Fig.~\ref{figL5} but two more sites). b) TFLIM with $J_X=1$ which converges more slowly (All unlabeled steps used the canonical bracket). c) The integrable TFIM model with $J_X=1$.}
  \label{appfigL7}
  \end{figure*}

\begin{figure*}
  \includegraphics[trim = 0 0 0 0cm, clip,width=.33\linewidth]{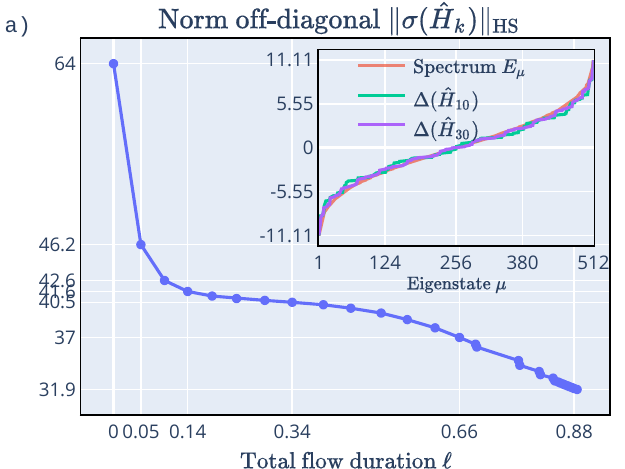}
  \includegraphics[trim = 0cm 0 0cm 0cm, clip,width=.33\linewidth]{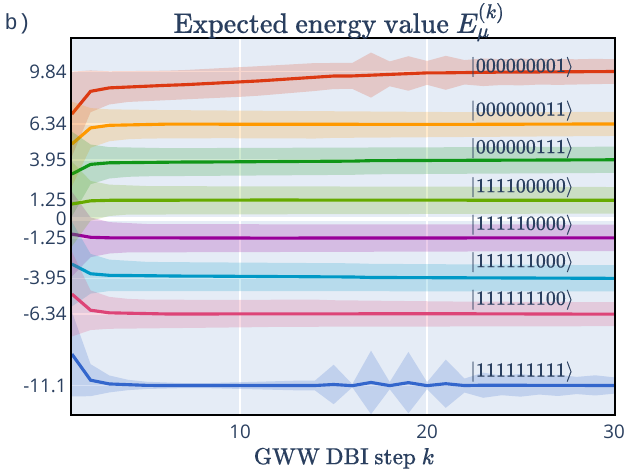}
  \includegraphics[trim = 0 0 0 0, clip,width=.33\linewidth]{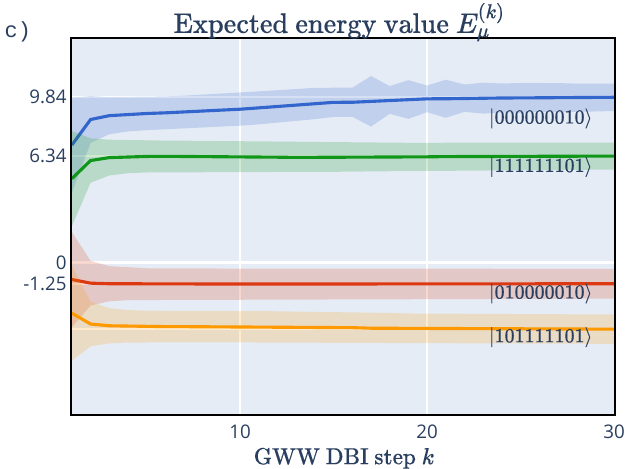}
  \caption{Analogous to Fig.~\ref{figeigsGWW}, GWW flow of  computational basis states towards eigenstates but for the transverse-field Ising model $\h_\text{TFIM}$ instead of TLFIM.
  {a)} The decrease of the norm of the off-diagonal restriction progresses smoothly instead of in stages.
  A lower value of the norm of the off-diagonal restriction is reached.
  The degenerate steps in the diagonal restrictions (inset) are less pronounced and thus match the spectrum faster.
{b)} Expected values of the energy  tend towards their asymptotic values faster than for TLFIM.
States flowing towards outer edges of the spectrum exhibit periodic increases of energy fluctuations which does not appear for TLFIM and may be due to integrability.
This signifies coupling and decoupling to other states and may be necessary to lift degeneracies in general.
{c)} States that reach the smallest energy fluctuations after $k=30$ flow steps accelerate the convergence  after about $k=10$ steps when the overall norm from panel a) moves away from the slowly decreasing plateaux and the oscillations of energy fluctuations begin to appear.
}
\label{appfigeigsGWW2}
\end{figure*}
\begin{figure*}
  \includegraphics[trim = 0 0 0 0, clip,width=\linewidth]{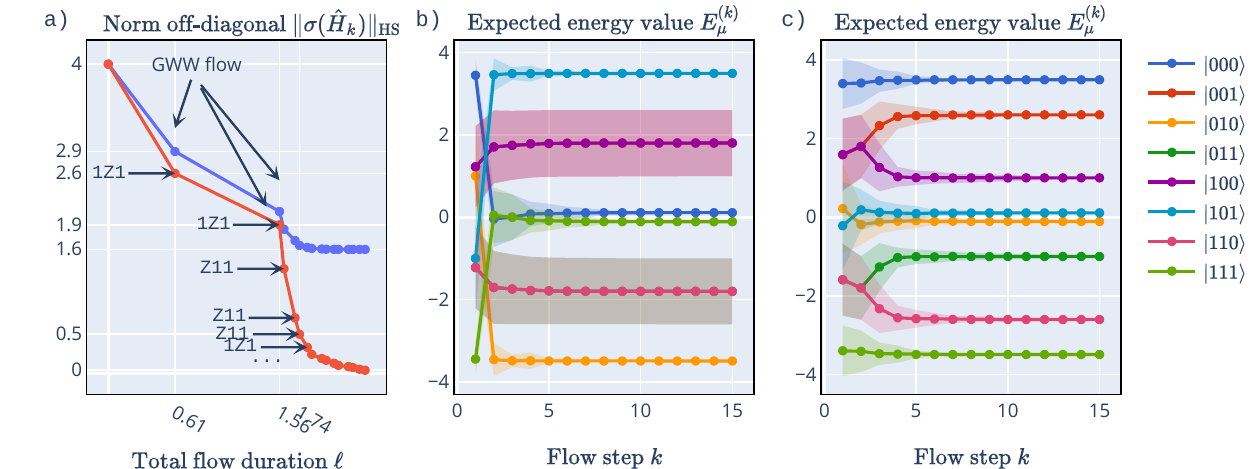}
  \caption{Comparison of variational and GWW flows in analogy to Fig.~\ref{figL3diag} but for $J_X=1$ and not $2$. 
  a) The Hilbert-Schmidt norm of the off-diagonal restriction decreases more rapidly for the variational double-bracket flow but sometimes locally optimal choices lead to smaller decreases in the long run (as evidence by the two $\sigma$-decrease curves meeting after two steps.
  b)  The fully polarized states converge very fast when using the GWW flow.
  c) Again, the degeneracy of the initial computational basis states becomes fully lifted after only a handful of flow steps and the expected energy fluctuation becomes negligible in the course of the flow.
  }
  \label{appfigL3diag}
\end{figure*}

\section{Notes on the GWW flow}
\label{app:existence}
\subsection{GWW flow as a system of ordinary differential equations}
\label{app:existencecontext}
The system of equations
  \eqref{eq:derivativehell} and \eqref{eq:well} can be written as a double-bracket differential equation
\begin{align}
  \partial_\ell \Hell = [ [\dHell,\oHell], \Hell]\ .
\end{align}
We can equivalently look at individual matrix elements with $i,j={1,\ldots, D}$
\begin{align}
  \partial_\ell (\Hell)_{i,j} =  \sum_{k=1}^D\left( ([\dHell,\oHell])_{i,k} (\Hell)_{k,j} - (\Hell)_{i,k}([\dHell,\oHell])_{k,j}\right)\ .
\end{align}
After we expand the matrix elements of the canonical bracket further we have
\begin{align}
  ([\dHell,\oHell])_{i,k} = (\Hell)_{i,k} ((\Hell)_{i,i} - (\Hell)_{k,k}) (1-\delta_{i,k})\ .
\end{align}
Thus
\begin{align}
  \partial_\ell (\Hell)_{i,j} =  \sum_{k=1}^D(\Hell)_{i,k}(\Hell)_{k,j}\left(  ((\Hell)_{i,i} - (\Hell)_{k,k}) (1-\delta_{i,k}) - ((\Hell)_{k,k} - (\Hell)_{j,j}) (1-\delta_{j,k})\right)\ .
  \label{eq:ODE}
\end{align}
This means that the prescribed function is a multi-variate polynomial of degree $3$.

{
\def\A{\mathcal \W}
\def\B{\mathcal \H}
Polynomials are locally Lipschitz-continuous and thus by the Picard-Lindel\"of theorem for an appropriately restricted domain of flow durations $\ell$ and matrix element values of $(\Hell)_{i,j}$ there is a unique solution~\cite{hirsch2012differential}.
Within that domain of flow durations $\Hell$ exists and hence $\Well$ exists.
Then the Dyson-Picard-Lindel\"of  series is given by
\begin{align}
\label{eqDPL}
    \Hell = \h_0 +
    \sum_{n=1}^\infty \int_{\snt}\dtn 
    [\W_{s_n},[\ldots,[\W_{s_1},\h_0]\ldots]\ ,
\end{align}
where  $\sntc n s=\{r\in \R^{\times n}\,:\,0\le r_1\le r_2\le \ldots\le r_n\le s\}$ is the $n$ dimensional simplex.
This series defines a unitary transformation $\h_0\mapsto \Hell$~\cite{YoungReviewLR} which we may denote by $\Uell$.
Alternatively, we may treat $-i\Well$ as an input time dependent Hamiltonian and we may write it in short-hand notation of time-ordered exponantials  as
\begin{align}
    \Uell = \mathcal T e^{ - \int_0^\ell \W_r \dr}\ .
    \label{eq:timeorderedGWW}
\end{align}

Polynomials are not globally Lipschitz-continuous so the Picard-Lindel\"of theorem does not directly yield a unique solution for $\ell\rightarrow \infty$.
Ref.~\cite{hirsch2012differential} stresses that lack of global Lipschitz-continuity can at times have drastic repercussions and exemplifies it by the ordinary differential equation $ d x(t)/dt = 1 +x(t)^2$ whose solution is $x(t)=\tan(c+t)$ where $c$ is a constant set by the initial condition and we get $x(-c\pm \pi/2)\rightarrow \pm \infty$.
Thus, the above discussion using basic existence theorems from the theory of dynamical systems substantiaties a priori the existence of the GWW flow only for a certain flow duration restricted by the local Lipschitz-continuity constant.

The monograph by Helmke and Moore~\cite{helmke_moore_optimization} provides a proof of the existence of the solution to double-bracket flow equations for all $\ell$ using a differential-geometric approach.
Thus, Eq.~\eqref{eq:timeorderedGWW} is well-defined for all $\ell\in \mathbb R$.

Subsec.~\ref{appunitaryapprox} proves the existence of GWW flow  for all flow durations $\ell$ without resorting to the theory of ordinary differential equations but rather by showing that the GWW DBI converges to the GWW flow.
While this necessitates technical calculations, the proof will be self-contained and constructive.

\subsection{Monotonicity of GWW flow}
\label{app:monotonicity}
When it exists, the GWW flow is smooth and so we can derive in the continuous case the precursor to the DBI monotonicity lemma~\ref{sgdlemma} from the main text.
\begin{proposition}[Monotonicity of GWW flow]
We have
\begin{align}
\partial_\ell\| \oHell\|_\mathrm{HS}^2 = -2\| {\Well}\|_\mathrm{HS}^2\ .
\end{align}
\end{proposition}

\begin{proof}
  We need to begin by studying the derivative of the off-diagonal terms which is defined as
  \begin{align}
    \partial_\ell \oHell = \lim_{\dell\rightarrow 0} \frac{\o{\H_{\ell+\dell}}-\oHell}{\dell} \ .
  \end{align}
The Hamiltonian satisfies the GWW flow equation~\eqref{eq:derivativehell} and using linearity of $\o\cdot$ we get
  \begin{align}
    \partial_\ell \oHell &= \lim_{s\rightarrow 0} \frac{\o{\H_{\ell}}-\o{\H_{\ell+s}}}{s}\\
   &=\sigma\left( \lim_{s\rightarrow 0} \frac{{\H_{\ell}}-{\H_{\ell+s}}}{s}\right)\\
    &= \o{[\W_\ell,\H_\ell]}\ .
  \end{align}
Using this expression together with the product differentiation rule and trace cyclicity we have
\begin{align}
\partial_\ell\| \oHell\|_\text{HS}^2
&= \tr\left(\oHell\,\partial_\ell\oHell\right)
+ \tr\left(\partial_\ell\oHell\, \oHell\right)\\
&= 2\,\tr\left(\oHell\,\o{[\W_\ell,\H_\ell]}\right)\ .
\end{align}
Using the definition of $\o\cdot$
\begin{align}
  \partial_\ell\| \oHell\|_\text{HS}^2 
=&2\,\tr\left(\oHell\,[\W_\ell,\H_\ell]\right)
-2\,\tr\left(\oHell\,\d{[\W_\ell,\H_\ell]}\right)\ .
\end{align}
The second term vanishes as anticipated by lemma~\ref{orthogonality}.
Finally we use the cyclicity formula \eqref{eqcyclicity} and definition of the canonical bracket~\eqref{eq:well}  obtaining
\begin{align}
  \partial_\ell\| \oHell\|_\text{HS}^2
&=2\,\tr\left(\oHell\,[\W_\ell,\H_\ell]\right)\\
&=2\,\tr\left(\Well\,[\H_\ell,\oHell]\right)\\
&=2\,\tr\left(\Well\,[\d{\H_\ell},\oHell]\right)\label{eq:drop}\\
&=2\,\tr\left(\Well^2\right)\ .
\end{align}
The crucial sign appears because the GWW flow generator is anti-hermitian
\begin{align}
  \Well^\dagger &= ( [\dHell,\oHell])^\dagger \\
  &= \oHell^\dagger \dHell^\dagger -\dHell^\dagger \oHell^\dagger\\
  &= -\Well\ 
\end{align}
which coincidentally implies that it does generate a unitary flow $\Uell$. 
Using that the GWW flow generator is anti-hermiatian and the definition of the Hilbert-Schmidt norm yields the proposition.
\end{proof}

\subsection{Convergence of DBI discretizations of the GWW flow}
\label{appunitaryapprox}
\begin{proposition}[Existence of GWW flow]
The GWW flow unitaries $\Uell$ exist for any $\ell\in [0,\infty)$.
\label{GWWexistenceprop}
\end{proposition}

The proof of this proposition will build upon several intermediate results.
The exact unitary of the GWW flow $\Uell$ will be defined as the limit of composing unitary operations with piece-wise constant generators.
These generators will be built recursively from previous steps.
We will  show that this sequence converges and that the limit unitary satisfies the GWW flow equation.

The sequence of approximations enumerated by $N=1,2,\ldots$ will based on flow parameters $\elln_k =k \ell 2^{-N}$ with $k\in\{0, 1,\ldots 2^N\}$.
The approximations of the GWW flow unitaries will be denoted as $\{\un_{k}\}_{k\in\{0, 1,\ldots 2^N\}}$,
of the flowed Hamiltonians as $\{\hn_{k}\}_{k\in\{0, 1,\ldots 2^N\}}$ and of the canonical bracket as $\{\wn_{k}\}_{k\in\{1,\ldots 2^N\}}$.
The upper index indicates refinements of the approximations and  $\{\unm_{k}\}_{k\in\{0, 1,\ldots 2^{N-1}\}}$ will be compared to $\{\un_{k}\}_{k\in\{0, 1,\ldots 2^N\}}$.

Fix $N\in \mathbb N_>$ and for $k=0$ define
\begin{align}
  \un_{0} = \id \ 
\end{align}
and
\begin{align}
  \hn_0 = \h\ .
\end{align}
The first approximate flow step will be generated by
\begin{align}
  \wn_{1} = [\d{\h},\o{\h}]\ .
\end{align}
For any $k\in\{1,\ldots 2^N\}$ let us further define
\begin{align}
  \un_k = \un_{k-1}e^{-\delln \wn_{k}}
\end{align}
with $\delln = \ell 2^{-N}$.
This leads to the definition of the flowed Hamiltonian approximation
\begin{align}
  \hn_{k}=  {\un_{{k}}}\,^\dagger\h \un_{{k}}\ .
\end{align}
For $k<N$ the approximate canonical bracket for the following step is
\begin{align}
  \wn_{{k+1}} = [\d{\hn_{k}},\o{\hn_{k}}]\ .
\end{align}

The GWW flow unitary will be defined as
\begin{align}
  \Uell = \lim_{N\rightarrow \infty} \un_{2^N}\ .
\label{eqlimitUell}
\end{align}
This can be interpreted as a limit of the GWW DBI, only we make explicit the dependence on the number of discretization steps $N$.
Point-wise (i.e., for fixed $\ell$) convergence  in operator norm  is implied by the sequence of GWW flow approximations being a Cauchy sequence
\begin{align}
  \lim_{N\rightarrow \infty} \|  \un_{2^N} - \unm_{2^{N-1}}\| =0 \ .
\label{cauchy}
\end{align}
Notice the notation used subsequently: The approximations labeled by $N-1$ have flow step twice larger than those labeled by $N$ and for both GWW flow approximations in \eqref{cauchy} the lower index signifies that the GWW  flow at flow parameter $\ell$ is approximated because $\ell = \delln 2^N = \dellnm 2^{N-1}$.

\begin{proposition}[Reduction to canonical bracket approximation]
For any $k\in\{0,\ldots,2^N-1\}$ we have
  \begin{align}
\|\unm_{{k+1}} -\un_{{2k+2}}\| \le&  \|\unm_{{k}} -\un_{{2k}}\|+
  {\dellnm} \|\wn_{{2k+1}}- \wnm_{{k+1}} \|
  +{\delln} \|\wn_{{2k+2}}-\wn_{{k+1}} \|\ .
  \label{eqh}
\end{align}
\end{proposition}

\begin{proof}
  
In general we have
\begin{align}
  \un_{{2k+2}}= \un_{{2k}} e^{-\delln \wn_{2k+1}}e^{-\delln \wn_{2k+2}} 
\end{align}
and let us use the triangle inequality to obtain
\begin{align}
\|\unm_{{k+1}} -\un_{{2k+2}}\| \le&
 \| \un_{{2k}}e^{-\delln \wn_{2k+1}}e^{-\delln \wn_{2k+2}}
 - \unm_{{k}}e^{-\delln \wn_{2k+1}}e^{-\delln \wn_{2k+2}}\|
 \\&+\| \unm_{{k}}e^{-\delln \wn_{2k+1}}e^{-\delln \wn_{2k+2}}
-\unm_{{k}}e^{-\dellnm \wnm_{k+1}}\|\ 
\\\le&
  \|\unm_{{k}} -\un_{{2k}}\| +\|e^{-\delln \wn_{2k+1}}e^{-\delln \wn_{2k+2}}
-e^{-\dellnm \wnm_{k+1}}\| \ .
\end{align}
We will use that for any anti-hermitian $\Xi,\Xi' \in \li(D)$ it holds that
\begin{align}
  \|e^\Xi-e^{\Xi'}\| \le  \|\Xi-{\Xi'}\| \ .
\end{align}
Conveniently $\dellnm=2\delln$ and we get
\begin{align}
\|e^{-\delln \wn_{2k+1}}e^{-\delln \wn_{2k+2}}
-e^{-\dellnm \wnm_{k+1}}\|\le&
\|e^{-\delln \wn_{2k+2}}-e^{-\delln \wn_{2k+1}}\|\\\nonumber
&+\|e^{-2\delln \wn_{2k+1}}-e^{-\dellnm \wnm_{k+1}}\|\\
  \le&
  {\dellnm} \|\wn_{{2k+1}}- \wnm_{{k+1}} \|
  +{\delln} \|\wn_{{2k+2}}-\wn_{{2k+1}} \|\ .
\end{align}

\end{proof}

\begin{proposition}[Canonical bracket approximation]
\label{bracketapproximation}
  The canonical bracket approximation obeys
  \begin{align}
  \|\wn_{{2k+1}}- \wnm_{{k+1}} \|
  \le&16\|\h\|\|{\unm_{k}} -{\un_{2k}}\|\ .
  \end{align}
  and is Lifschitz continuous
  \begin{align}
   \|\wn_{{2k+2}}-\wn_{{2k+1}} \|
    \le& 64\|\h\|^2 \delln \ .
  \end{align}
\end{proposition}
\begin{proof}
  
Using lemma~\ref{bndcom} we bound it as
\begin{align}
 \|\wn_{{2k+1}}- \wnm_{{k+1}} \|
  \le&8\|\h\|\,\|  \hn_{{2k}}-\hnm_{{k}} \| 
\end{align}
and
\begin{align}
  \|\wn_{{2k+2}}-\wn_{{2k+1}} \|\ \le&8\|\h\|\,\| \hn_{{2k+1}} -\hn_{{2k}}\| \ .
  \end{align}
Here we additionally used that $\|\hnm_k\|=\|\hn_{2k-1}\|=\|\hn_{2k-2}\|=\|\h\|$ due to unitary invariance.

For the first term we observe that by definition $\hn_{2k} = {\un_{2k}}\,^\dagger \h \un_{2k}$
and so
\begin{align}
  \| \hnm_{{k}} - \hn_{{2k}}\| \le&
  \|{\unm_{k}}\,^\dagger \h \un_{2k}-{\un_{2k}}\,^\dagger \h \un_{2k}\|+
  \|{\unm_{k}}\,^\dagger \h \un_{2k}-{\unm_{k}}\,^\dagger \h \unm_{k}\|\\
  \le&2\|\h\|\|\unm_{k} -\un_{2k}\|\ .
\end{align}

For the second term we proceed similarly, noticing that $\hn_{{2k+2}} =e^{\delln \wn_{2k+2}}\hn_{{2k+1}}e^{-\delln\wn_{2k+2}}$
so
\begin{align}
 \| \hn_{{2k+1}} -\hn_{{2k+2}}\| \le& 2\|\h\|\|e^{\delln \wn_{2k+2}}-\id\|\\
  \le& 2\|\h\| \, \|\delln\wn_{2k+2}\|\\
  \le& 8\|\h\|^2 \delln \ .
\end{align}

\end{proof}
\begin{proposition}[Convergence]
  GWW flow approximation $\un_{2^N}$ as a sequence in $N$ converges.
\end{proposition}
\begin{proof}
By the preceeding propositions we have the recurrence  
\begin{align}
\|\unm_{{k+1}} -\un_{{2k+2}}\| \le&  (1+16\delln \|\h\|^2)\|\unm_{{k}} -\un_{{2k}}\|+ 8\|\h\|^2 \delln^2\ .
  \end{align}
We will denote the subject of th recursion of the upper bound as $\{\alpha_k\}_{k\in\{0,\ldots,2^N\}}$ so that
\def\a{a}
\def\b{b}
\begin{align}
\alpha_{k+1} =&  (1+\a)\alpha_k+ \b\ .
\end{align}
where we defined $\a=16 \|\h\|^2\delln$ and $\b=  64\|\h\|^2 \delln^2$
Observing that $\b=4 \a$ we obtain from lemma~\ref{recursionlemma} the bound
\begin{align}
\|\unm_{{2^{N-1}}} -\un_{{2^N}}\| \le& 4 \left(1+\frac{16 \ell \|\h\|}{2^N}\right)^{2^N}\frac{ \ell}{2^{N}} \le O(\delln)\ .
  \end{align}\
\end{proof}

We continue by proving continuity of the unitary $\Uell$.
\begin{proposition}[Continuity of GWW flow]
For any $\ell\in[0,\infty)$ we have
\begin{align}
 \lim_{\epsilon\rightarrow 0} \U_{\ell+\epsilon} = \Uell\ .
  \label{continuity}
\end{align}
\end{proposition}

\begin{proof}
Using the discretization scheme for some fixed $N$ above we set
\begin{align}
 \| \U_{\ell+\epsilon} - \Uell\|
 &= \lim_{N\rightarrow \infty} \| e^{-s\w_1}\ldots e^{-s\w_{2^N}}- e^{-s'\w_1'}\ldots e^{-s'\w_{2^N}'}\|
\end{align}
where we define the discretization step sizes for each target flow duration $s=\ell/N$, $s'=(\ell+\epsilon)/N$ and $\w_k$, $\w_k'$ are the approximants to the canonical bracket for $\Uell$ and  $\U_{\ell+\epsilon}$, respectively.
We will arrive at a bound independent of $N$ which will allow to take the limit $\epsilon\rightarrow 0$.

We use appropriate telescoping which allows to iteratively relate to previous discretization steps and then again use proposition~\ref{bndcom}
\begin{align}
 \| e^{-s\w_1}\ldots e^{-s\w_{2^N}}- e^{-s'\w_1}\ldots e^{-s'\w_{2^N}'}\|
 &\le \| e^{-s\w_1}\ldots e^{-s\w_{2^N-1}}- e^{-s'\w_1'}\ldots e^{-s'\w_{2^N-1}'}\|\\
 &\quad+ \|s\w_{2^N}-s'\w_{2^N}'\|\nonumber\\
 &\le \| e^{-s\w_1}\ldots e^{-s\w_{2^N-1}}- e^{-s'\w_1'}\ldots e^{-s'\w_{2^N-1}'}\|\\
 &\quad+ s\|\w_{2^N}-\w_{2^N}'\| +(s'-s)\|\w_{2^N}'\|\nonumber\\
&\le (1+16\|\h\|^2s)\| e^{-s\w_1}\ldots e^{-s\w_{2^N-1}}- e^{-s'\w_1'}\ldots e^{-s'\w_{2^N-1}'}\|\\
 &\quad+ 64\|\h\|^2 \epsilon/2^N\ .\nonumber
\end{align}
For both flows the total duration differs but the starting direction is the same so $\w_1=\w_1'$ and so using lemma~\ref{recursionlemma}
we find the bound
\begin{align}
 \| \U_{\ell+\epsilon} - \Uell\|
 &\le \lim_{N\rightarrow \infty} 16 \frac{e^{16\|\h\|^2 \ell}-1}\ell  \epsilon
 \end{align}
which linearly converges towards zero for all $\ell\in[0,\infty)$.

\end{proof}

We next use continuity of the GWW flow to prove the following one-parameter group law of GWW flow.
In the following it will be useful to extend the notation above to indicate which initial operator is being flowed.
More specifically, we will write $\Uell(\J)$ for the GWW flow unitary of $\J\in\linops$ or $\un_{{k}}(\J)$ for the entries of the discretization sequence considered.
\begin{proposition}[Group law of GWW flow]
For any $\ell,\ell'\in[0,\infty)$ and any $\J\in\linops$ we have
\begin{align}
  \U_{\ell+\ell'}(\J) = \Uell(\J)\U_{\ell'}(\J_{\ell})\ ,
  \label{grouplaw}
\end{align}
where $\J_\ell = \Uell(\J)^\dagger\J\Uell(\J)$.
\end{proposition}
In other words the GWW flow of $\J$ with total flow duration $\ell +\ell'$ is the GWW flow of $\J$ with duration $\ell$ followed by the GWW flow of $\J_\ell$ with duration $\ell'$.

\begin{proof}
It suffices to show a simpler property, namely
\begin{align}
  \U_{\ell}(\J) = \U_{\ell/2}(\J)\U_{\ell/2}\bigl(\U_{\ell/2}(\J)^\dagger\J\U_{\ell/2}(\J)\bigr)\ .
  \label{binarysplit}
\end{align}
This is because we can iterate the bound according to a convergent binary representation of $0\le \ell/(\ell+\ell')\le 1$.
The above equality holds for any finite binary approximation to $\ell/(\ell+\ell')$ and so by continuity of the GWW flow the limit yields the group law \eqref{grouplaw} for any $\ell,\ell'\in[0,\infty)$.

For the proof of this simpler relation, let us indicate which operator is the starting point of the flow also for the entries of the sequence coverging onto the GWW flow, e.g., $\un_{{k}}(\J)$.
The flow duration discretization is binary and so we find the prototype of the group law to be
\begin{align}
  \un_N(\J) = \un_{2^{N-1}}(\J)\un_{2^{N-1}}(\un_{2^{N-1}}(\J)^\dagger\J\un_{2^{N-1}}(\J))\ .
\end{align}
Indeed, the first operator on the right hand side converges to
\begin{align}
  \un_{2^{N-1}}(\J)\rightarrow \U_{\ell/2}(\J)
\end{align}
and it remains to show that 
\begin{align}
  \un_{2^{N-1}}(\un_{2^{N-1}}(\J)^\dagger\J\un_{2^{N-1}}(\J))\rightarrow  \U_{\ell/2}(\J_{\ell/2})\ .
\end{align}
\def\Jn{\J^{\mtiny{(N)}}}
Let us define
\begin{align}
  \Jn = \un_{2^{N-1}}(\J)^\dagger\J\un_{2^{N-1}}(\J)
\end{align}
which allows us to write
\begin{align}
  \| \un_k(\Jn) - \un_k(\J_{\ell/2})\|
  &\le \| \un_{k-1}(\Jn) - \un_{k-1}(\J_{\ell/2})\| +\delln \| \w_k(\Jn) - \w_k(\J_{\ell/2})\|\\
  &\le (1+16 \|\J\|^2 \delln) \| \un_{k-1}(\Jn) - \un_{k-1}(\J_{\ell/2})\| \ 
\end{align}
in analogy with proposition~\ref{bracketapproximation}.
Iterating this, we obtain
\begin{align}
  \| \un_{2^{N-1}}(\Jn) - \un_{2^{N-1}}(\J_{\ell/2})\|
  &\le (1+16 \|\J\|^2 \delln)^{2^N} \| \un_{1}(\Jn) - \un_{1}(\J_{\ell/2})\| \\
  &\le 16\|\J\|^2 e^{16 \|\J\|^2 \ell} \| \un_{2^{N-1}}(\J)^\dagger\J\un_{2^{N-1}}(\J)-\J_{\ell/2}\| \delln \ .
\end{align}
The convergence result above shows that this upper bound converges towards zero which proves formula~\eqref{binarysplit}.
\end{proof}

The group law of the GWW flow allows us to prove that the limit unitaries $\Uell$ defined in Eq.~\eqref{eqlimitUell} are indeed GWW flow unitaries because they statisfy the
\begin{proposition}[GWW flow equation]For any $\ell\in[0,\infty)$
\begin{align}
  \partial_\ell \Uell = -\Uell \Well\ .
\end{align}
 \end{proposition}
\noindent This indeed implies the GWW flow equation for the Hamiltonian because
\begin{align}
  \partial_\ell \Hell &= -\Uell^\dagger \h \Uell \Well+  \Well\Uell^\dagger \h \Uell \Well\\
  &= [\Well,\Hell]\ .
\end{align}
\begin{proof}
  Using the group law \eqref{grouplaw} we have
  \begin{align}
  \partial_\ell \Uell(\h)
  &= \lim_{\epsilon\rightarrow 0} \frac{ \U_{\ell+\epsilon}(\h)-\Uell(\h)}\epsilon\\
  &=\Uell(\h) \lim_{\epsilon\rightarrow 0} \frac{ \U_{\epsilon}(\Hell)-\id}\epsilon\ .
  \end{align}
The proposition statement is equivalent to demonstrating that the following vanishes
\begin{align}
  \lim_{\epsilon\rightarrow 0} \left\|\frac{ \U_{\epsilon}(\Hell)-\id}\epsilon+\Well\right\|
  &\le \lim_{\epsilon\rightarrow 0}
  \left( \left\|\frac{ e^{-\epsilon \Well} -\id}\epsilon+\Well\right\| +
   \left\|\frac{ \U_{\epsilon}(\Hell)-e^{-\epsilon \Well} }\epsilon\right\| \right) \\
  &\le \lim_{\epsilon\rightarrow 0}
   \epsilon^{-1}\left\| \U_{\epsilon}(\Hell)-e^{-\epsilon \Well}\right\| \ .
\end{align}
Here the first term was the regular Lie derivative of the matrix exponential and so we need to establish that the deviation from a constant generator vanishes for $\epsilon\rightarrow 0$.
To prove this, let us first notice that $\Well = \wn_1(\Hell)\equiv\wn_1$ and so we have
\begin{align}
\| \U_{\epsilon}(\Hell)-e^{-\epsilon \Well}\| = \lim_{N\rightarrow \infty} \|e^{-\delln \wn_1}\ldots e^{-\delln \wn_{2^N}} - e^{-\epsilon \wn_1}\|  \ .
\end{align}
Through telescoping we have for any fixed $N$
\begin{align}
\|e^{-\delln \wn_1}\ldots e^{-\delln \wn_{2^N}} - e^{-\epsilon \wn_1}\|  
&\le \sum_{k=1}^{2^N-1} \left\|
e^{- k \delln \wn_1}e^{-  \delln \wn_{k+1}}\ldots e^{-\delln \wn_{2^N}}\right. \\
&\quad\quad\quad\quad-\left.
e^{- (k+1) \delln \wn_1}e^{-  \delln \wn_{k+2}}\ldots e^{-\delln \wn_{2^N}} \right\|\nonumber\\
&\le \sum_{k=1}^{2^N-1} \|\wn_1-\wn_{k+1}\| \delln \ .
\end{align}
Again by telscoping we have
\begin{align}
  \|\wn_k - \wn_1\|&\le \sum_{k'=2}^k\|\wn_{k'}-\wn_{k'-1}\|\\
  &\le 16\|\h\|\sum_{k'=2}^k\|e^{-\delln\wn_{k'-1}}-\id\|\\
  &\le 1024\|\h\|^3 \delln k \ 
\end{align}
and so
\begin{align}
\|e^{-\delln \wn_1}\ldots e^{-\delln \wn_{2^N}} - e^{-\epsilon \wn_1}\|
&\le  1024\|\h\|^3 ( \delln)^2  \sum_{k=1}^{2^N-1} k\\
&\le  1024\|\h\|^3 \epsilon^2\ .
\end{align}
This bound is independent of $N$ and quadratic in $\epsilon$ and thus it implies that 
\begin{align}
  \lim_{\epsilon\rightarrow 0}
   \epsilon^{-1}\left\| \U_{\epsilon}(\Hell)-e^{-\epsilon \Well}\right\| = 0\ .
\end{align}
\end{proof}

\begin{proof}[Proof of proposition~\ref{GWWexistenceprop} (Existence of GWW flow)]
Summarizing, the limit unitaries $\Uell$ defined in Eq.~\eqref{eqlimitUell} exist and they satisfy the GWW flow equation which proves the existence of the GWW flow.
\end{proof}

\end{document}